\documentclass[11pt,aps,pra,notitlepage,nofootinbib,tightenlines]{revtex4-1}

\usepackage[T1]{fontenc}
\usepackage[english]{babel}
\usepackage[utf8]{inputenc}

\usepackage[dvipsnames]{xcolor}
\usepackage{amsmath}
\usepackage{amssymb}
\usepackage{amsthm}
\usepackage{mathtools}
\usepackage{tikz}
\usepackage[colorlinks=true,linkcolor=blue,citecolor=blue]{hyperref}
\usepackage{pgfplots}
\usepackage{enumitem}
\setlist[enumerate]{noitemsep,partopsep=0pt,parsep=0pt}
\setlist[itemize]{noitemsep,partopsep=0pt,parsep=0pt}
\usepackage{booktabs}
\newtheorem{thm}{Theorem}
\newtheorem{lem}{Lemma}
\newtheorem{cor}{Corollary}

\newtheorem*{cor:limit-type1_direct_repeated}{Corollary~\ref{cor:limit-type1_direct} (restated)}
\newtheorem*{cor:direct_repeated}{Corollary~\ref{cor:direct} (restated)}
\newtheorem*{cor:limit-type1_strongconverse_repeated}{Corollary~\ref{cor:limit-type1_strongconverse} (restated)}
\newtheorem*{cor:strong-converse_repeated}{Corollary~\ref{cor:strong-converse} (restated)}
\theoremstyle{remark}
\newtheorem{rem}{Remark}

\numberwithin{equation}{section}

\DeclareMathOperator{\tr}{tr}
\DeclareMathOperator{\rank}{rank}
\DeclareMathOperator{\spec}{spec}

\DeclareMathOperator{\supp}{supp}
\DeclareMathOperator*{\argmin}{\arg\min}

\DeclareMathOperator{\sym}{sym}
\newcommand{\ket}[1]{|#1\rangle}
\newcommand{\bra}[1]{\langle#1|}
\newcommand{\ketbra}[2]{\ket{#1}\!\bra{#2}}
\newcommand{\proj}[1]{\ketbra{#1}{#1}}

\allowdisplaybreaks

\begin{document}

\title{Doubly minimized Petz and sandwiched R\'enyi mutual information: Operational interpretation 
from binary quantum state discrimination}
\author{Laura Burri}
\affiliation{Institute for Theoretical Physics, ETH Zurich, Zurich, Switzerland}

\begin{abstract}
The doubly minimized Petz R\'enyi mutual information of order $\alpha$ is defined as the minimum of the Petz divergence of order $\alpha$ of a given bipartite quantum state relative to all product states. The doubly minimized sandwiched R\'enyi mutual information is defined analogously, with the Petz divergence replaced by the sandwiched divergence. In this work, we study certain binary quantum state discrimination problems related to correlation detection. We show that the corresponding direct exponent is determined by the doubly minimized Petz R\'enyi mutual information of order $\alpha\in (1/2,1)$, and that the strong converse exponent is determined by the doubly minimized sandwiched R\'enyi mutual information of order $\alpha\in (1,\infty)$. This provides an operational interpretation of these types of R\'enyi mutual information and generalizes previous results for classical probability distributions to the quantum setting. For completeness, we also study the corresponding moderate deviation regime both below and above the threshold, and determine the Stein exponent and the second-order asymptotics.
\end{abstract}

\maketitle

\section{Introduction}\label{sec:1}

The interplay between information-theoretic tasks and information measures is a central theme in both classical and quantum information theory. 
This paper focuses on the task of i.i.d. (independent and identically distributed) correlation detection and explores how its error exponents -- the direct exponent and the strong converse exponent -- relate to R\'enyi generalizations of the mutual information. 
In classical information theory, this relation has been studied in previous work, yielding operational interpretations of the doubly minimized R\'enyi mutual information~\cite{tomamichel2018operational}. 
Here, we extend this relation to the quantum setting. 
We begin by reviewing the relevant results from classical information theory.

\emph{Classical setting.} Let $P_{XY}$ be the joint probability mass function (PMF) of two random variables $X$ and $Y$, and let $P_X$ and $P_Y$ be the marginal PMFs of $X$ and $Y$, respectively.  
Based on the R\'enyi divergence $D_\alpha$, we consider the following types of R\'enyi mutual information (RMI) for $\alpha\in [0,\infty]$.
\begin{align}
I_\alpha^{\uparrow\uparrow}(X:Y)_P&\coloneqq D_\alpha (P_{XY}\| P_X P_Y)
\label{eq:i-classical-0}
\\
I_\alpha^{\uparrow\downarrow}(X:Y)_P&\coloneqq \inf_{R_Y}D_\alpha (P_{XY}\| P_X R_Y)
\label{eq:i-classical-1}
\\
I_\alpha^{\downarrow\downarrow}(X:Y)_P&\coloneqq \inf_{Q_X,R_Y}D_\alpha(P_{XY}\| Q_X R_Y)
\label{eq:i-classical-2}
\end{align}
We call them the \emph{non-minimized RMI}, the \emph{singly minimized RMI}, and the \emph{doubly minimized RMI}. 
The minimizations are over PMFs $Q_X$ and $R_Y$. 
If $\alpha=1$, then each of the above measures coincides with the classical mutual information $I(X:Y)_P$~\cite{lapidoth2019two}. 
Thus, the listed measures may be viewed as one-parameter generalizations of the classical mutual information. 

Each of the three measures can be related to the direct exponent and the strong converse exponent of a certain binary discrimination problem, as outlined in Table~\ref{tab:classical}. 
The problem in the third row of this table corresponds to the task of i.i.d. correlation detection in the classical setting. 
The listed results for the direct exponents provide an operational interpretation of the family of the non-minimized RMIs and the singly minimized RMIs of order $\alpha\in (0,1)$, and the doubly minimized RMIs of order $\alpha\in (\frac{1}{2},1)$, respectively. 
On the other hand, the listed results for the strong converse exponents provide an operational interpretation of the family of the non-minimized, the singly minimized, and the doubly minimized RMIs of order $\alpha\in (1,\infty)$, respectively.

\begin{table}
\begin{tabular}{p{0.26\textwidth}p{0.36\textwidth}p{0.36\textwidth}} \toprule
\textbf{Null hypothesis,\newline alternative hypothesis} & \textbf{Direct exponent}&\textbf{Strong converse exponent}\\ \midrule
$H_0^n=\{P_{XY}^{\times n}\}$ \newline
$H_1^n=\{P_{X}^{\times n} P_{Y}^{\times n}\}$&
For any $R\in (0,\infty)$~\cite{hoeffding1965asymptotically,hoeffding1965probabilities,csiszar1971error,blahut1974hypothesis,audenaert2008asymptotic} \newline
$\lim\limits_{n\rightarrow\infty}-\frac{1}{n}\log \hat{\alpha}_n(e^{-nR})$\newline
$= \sup\limits_{s\in (0,1)}\frac{1-s}{s}(I_s^{\uparrow\uparrow}(X:Y)_P-R)$.
&For any $R\in [0,\infty)$~\cite{blahut1974hypothesis,han1989strong,nakagawa1993converse} \newline
$\lim\limits_{n\rightarrow\infty}-\frac{1}{n}\log (1-\hat{\alpha}_n(e^{-nR}))$\newline
$= \sup\limits_{s\in (1,\infty)}\frac{s-1}{s}(R-I_s^{\uparrow\uparrow}(X:Y)_P)$.
\\ \hline
$H_0^n=\{P_{XY}^{\times n}\}$ \newline
$H_1^n=\{P_{X}^{\times n} R_{Y}^{\times n}\}_{R_{Y}}$&
For any $R\in (I_0^{\uparrow\downarrow}(X:Y)_P,\infty)$~\cite{tomamichel2018operational} \newline
$\lim\limits_{n\rightarrow\infty}-\frac{1}{n}\log \hat{\alpha}_n(e^{-nR})$\newline
$= \sup\limits_{s\in (0,1)}\frac{1-s}{s}(I_s^{\uparrow\downarrow}(X:Y)_P-R)$.
&For any $R\in [0,R_\infty)$~\cite{tomamichel2018operational}  \newline
$\lim\limits_{n\rightarrow\infty}-\frac{1}{n}\log (1-\hat{\alpha}_n(e^{-nR}))$\newline
$= \sup\limits_{s\in (1,\infty)}\frac{s-1}{s}(R-I_s^{\uparrow\downarrow}(X:Y)_P)$.
\\ \hline
$H_0^n=\{P_{XY}^{\times n}\}$ \newline
$H_1^n=\{Q_{X}^{\times n} R_{Y}^{\times n}\}_{Q_{X},R_{Y}}$&
For any 
$R\in (R_{1/2},\infty)$~\cite{tomamichel2018operational} \newline
$\lim\limits_{n\rightarrow\infty}-\frac{1}{n}\log \hat{\alpha}_n(e^{-nR})$\newline
$= \sup\limits_{s\in (\frac{1}{2},1)}\frac{1-s}{s}(I_s^{\downarrow\downarrow}(X:Y)_P-R)$.
&For any $R\in [0,R_\infty )$~\cite{tomamichel2018operational} \newline
$\lim\limits_{n\rightarrow\infty}-\frac{1}{n}\log (1-\hat{\alpha}_n(e^{-nR}))$\newline
$= \sup\limits_{s\in (1,\infty)}\frac{s-1}{s}(R-I_s^{\downarrow\downarrow}(X:Y)_P)$.
\\ \bottomrule
\end{tabular}
\caption{Overview of certain binary discrimination problems. 
Let $P_{XY}$ be a PMF. 
Each row pertains to a sequence of binary discrimination problems with null hypothesis $H_0^n$ and alternative hypothesis $H_1^n$ for $n\in \mathbb{N}_{>0}$. 
In the second row, the $n$th alternative hypothesis is given by 
$P_X^{\times n} R_Y^{\times n}$ for PMFs $R_Y$.
In the third row, the $n$th alternative hypothesis is given by 
$Q_X^{\times n} R_Y^{\times n}$ for PMFs $Q_X,R_Y$. 
The works cited in the second and third column derive single-letter formulas (i.e., formulas in which $P_{XY}$ occurs only once and not $n\rightarrow\infty$ many times) 
for the corresponding direct and strong converse exponents. 
The formulas for the strong converse exponent are valid if 
$I(X:Y)_P\neq I_\infty^{\uparrow\uparrow}(X:Y)_P$, 
$I(X:Y)_P\neq I_\infty^{\uparrow\downarrow}(X:Y)_P$, 
and $I(X:Y)_P\neq  I_\infty^{\downarrow\downarrow}(X:Y)_P$, respectively. 
The function $\hat{\alpha}_n(\mu)$ is linked to the $n$th hypothesis testing problem and is defined as the minimum type-I error when the type-II error is upper bounded by $\mu\in [0,\infty)$, see Section~\ref{ssec:hypothesis}. 
In that section, we will also explain that the formula for the direct exponent in the second row can be extended to $R\in (0,\infty)$, and the formulas for the strong converse exponent in the second and third row can be extended to $R\in [0,\infty)$, 
because the binary discrimination problems in this table are special cases of the ones that will be outlined in Table~\ref{tab:quantum} (see Remark~\ref{rem:extension}). 
The lower bound in the third row is defined as 
$R_{1/2}\coloneqq I_{1/2}^{\downarrow\downarrow}(X:Y)_P-\frac{1}{4}\frac{\partial}{\partial s^+} I_{s}^{\downarrow\downarrow}(X:Y)_P|_{s=1/2}$~\cite{tomamichel2018operational}. 
For a negative answer to the question of potential extensions of the formula for the direct exponent in the third row, see~\cite{lapidoth2018testing}. 
The upper bound in the second row is defined as 
$R_\infty\coloneqq \lim_{s\rightarrow\infty}(I_s^{\uparrow\downarrow}(X:Y)_P +s(s-1)\frac{\mathrm{d}}{\mathrm{d}s}I_s^{\uparrow\downarrow}(X:Y)_P)$~\cite{tomamichel2018operational}. 
Similarly, the upper bound in the third row is defined as 
$R_\infty\coloneqq \lim_{s\rightarrow\infty}(I_s^{\downarrow\downarrow}(X:Y)_P +s(s-1)\frac{\mathrm{d}}{\mathrm{d}s}I_s^{\downarrow\downarrow}(X:Y)_P)$~\cite{tomamichel2018operational}.}
\label{tab:classical}
\end{table}

\paragraph*{Quantum setting.} 
Let $\rho_{AB}$ be a bipartite quantum state on $A$ and $B$, and let $\rho_A$ and $\rho_B$ be the marginal states on $A$ and $B$, respectively. 
Based on the Petz divergence $D_\alpha$, we consider the following types of Petz R\'enyi mutual information (PRMI) for $\alpha\in [0,\infty)$.
\begin{align}
I_\alpha^{\uparrow\uparrow}(A:B)_\rho &\coloneqq D_\alpha (\rho_{AB}\| \rho_A\otimes \rho_B)
\label{eq:i-quantum-0}
\\
I_\alpha^{\uparrow\downarrow}(A:B)_\rho &\coloneqq \inf_{\tau_B} D_\alpha (\rho_{AB}\| \rho_A\otimes \tau_B)
\label{eq:i-quantum-1}
\\
I_\alpha^{\downarrow\downarrow}(A:B)_\rho &\coloneqq \inf_{\sigma_A,\tau_B} D_\alpha (\rho_{AB}\| \sigma_A\otimes \tau_B)
\label{eq:i-quantum-2}
\end{align}
We call them the \emph{non-minimized PRMI}, the \emph{singly minimized PRMI}, and the \emph{doubly minimized PRMI}. 
The minimizations are over quantum states $\sigma_A$ and $\tau_B$. 
If $\alpha = 1$, then each of the above measures coincides with the mutual information $I(A:B)_\rho$~\cite{wilde2014strong,hayashi2016correlation}. 
Consequently, the three listed measures can be regarded as one-parameter generalizations of the mutual information.

Analogously, based on the sandwiched R\'enyi divergence $\widetilde{D}_\alpha$, we consider the following types of sandwiched R\'enyi mutual information (SRMI) for $\alpha\in (0,\infty]$.
\begin{align}
\widetilde{I}_\alpha^{\uparrow\uparrow}(A:B)_\rho &\coloneqq \widetilde{D}_\alpha (\rho_{AB}\| \rho_A\otimes \rho_B)
\label{eq:i-quantum0}
\\
\widetilde{I}_\alpha^{\uparrow\downarrow}(A:B)_\rho &\coloneqq \inf_{\tau_B} \widetilde{D}_\alpha (\rho_{AB}\| \rho_A\otimes \tau_B)
\label{eq:i-quantum1}
\\
\widetilde{I}_\alpha^{\downarrow\downarrow}(A:B)_\rho &\coloneqq \inf_{\sigma_A,\tau_B} \widetilde{D}_\alpha (\rho_{AB}\| \sigma_A\otimes \tau_B)
\label{eq:i-quantum2}
\end{align}
We call them the \emph{non-minimized SRMI}, the \emph{singly minimized SRMI}, and the \emph{doubly minimized SRMI}. 
Again, the minimizations are over quantum states $\sigma_A$ and $\tau_B$, and if $\alpha=1$, then all three measures coincide with the mutual information~\cite{wilde2014strong,hayashi2016correlation}. 
Thus, the three listed measures represent further one-parameter generalizations of the mutual information.

Given the results in Table~\ref{tab:classical}, it is natural to inquire whether these results can be generalized from the classical to the quantum setting. 
For the discrimination problems in the first and second row of Table~\ref{tab:classical}, it has been shown that this generalization is indeed possible, as outlined in the corresponding rows of Table~\ref{tab:quantum}. 
We show that the same generalization is possible for the discrimination problem in the third row of Table~\ref{tab:classical} -- which corresponds to i.i.d. correlation detection --, as outlined in the third row of Table~\ref{tab:quantum}. 
These single-letter formulas for the direct exponent (Theorem~\ref{thm:direct}) and the strong converse exponent (Theorem~\ref{thm:strong-converse}) are our main results. 
They yield an operational interpretation of the family of the doubly minimized PRMIs of order $\alpha\in (\frac{1}{2},1)$, and the family of the doubly minimized SRMIs of order $\alpha\in (1,\infty)$, respectively. 
For completeness, we also study the moderate deviation regime both below and above the threshold $I(A:B)_\rho$ (Theorems~\ref{thm:moderate_direct} and~\ref{thm:moderate_strongconverse}), and determine the Stein exponent (Corollary~\ref{cor:stein}) and the second-order asymptotics (Theorem~\ref{thm:second}) of the corresponding binary quantum state discrimination problems. 

\begin{table}
\begin{tabular}{p{0.26\textwidth}p{0.36\textwidth}p{0.36\textwidth}} \toprule
\textbf{Null hypothesis,\newline alternative hypothesis} & \textbf{Direct exponent}&\textbf{Strong converse exponent}\\ \toprule
$H_0^n=\{\rho_{AB}^{\otimes n}\}$ \newline
$H_1^n=\{\rho_{A}^{\otimes n}\otimes \rho_{B}^{\otimes n}\}$&
For any $R\in (0,\infty)$~\cite{hayashi2007error,nagaoka2006converse,audenaert2008asymptotic} \newline
$\lim\limits_{n\rightarrow\infty}-\frac{1}{n}\log \hat{\alpha}_n(e^{-nR})$\newline
$= \sup\limits_{s\in (0,1)}\frac{1-s}{s}(I_s^{\uparrow\uparrow}(A:B)_\rho -R)$.
&For any $R\in [0,\infty)$~\cite{mosonyi2014quantum,mosonyi2015two} \newline
$\lim\limits_{n\rightarrow\infty}-\frac{1}{n}\log (1-\hat{\alpha}_n(e^{-nR}))$\newline
$= \sup\limits_{s\in (1,\infty)}\frac{s-1}{s}(R-\widetilde{I}_s^{\uparrow\uparrow}(A:B)_\rho)$.
\\ \hline
$H_0^n=\{\rho_{AB}^{\otimes n}\}$ \newline
$H_1^n=\{\rho_{A}^{\otimes n}\otimes \tau_{B}^{\otimes n}\}_{\tau_{B}}$&
For any $R\in (0,\infty)$~\cite{hayashi2016correlation} \newline
$\lim\limits_{n\rightarrow\infty}-\frac{1}{n}\log \hat{\alpha}_n(e^{-nR})$\newline
$= \sup\limits_{s\in (0,1)}\frac{1-s}{s}(I_s^{\uparrow\downarrow}(A:B)_\rho -R)$.
&For any $R\in [0,\infty)$~\cite{hayashi2016correlation} \newline
$\lim\limits_{n\rightarrow\infty}-\frac{1}{n}\log (1- \hat{\alpha}_n(e^{-nR}))$\newline
$= \sup\limits_{s\in (1,\infty)}\frac{s-1}{s}(R-\widetilde{I}_s^{\uparrow\downarrow}(A:B)_\rho )$.
\\ \hline
$H_0^n=\{\rho_{AB}^{\otimes n}\}$ \newline
$H_1^n=\{\sigma_{A}^{\otimes n}\otimes \tau_{B}^{\otimes n}\}_{\sigma_{A},\tau_{B}}$ &
For any $R\in (R_{1/2},\infty)$ [Theorem~\ref{thm:direct}] \newline
$\lim\limits_{n\rightarrow\infty}-\frac{1}{n}\log \hat{\alpha}_n(e^{-nR})$\newline
$= \sup\limits_{s\in (\frac{1}{2},1)}\frac{1-s}{s}(I_s^{\downarrow\downarrow}(A:B)_\rho -R)$.
&For any $R\in [0,\infty)$ [Theorem~\ref{thm:strong-converse}] \newline
$\lim\limits_{n\rightarrow\infty}-\frac{1}{n}\log (1-\hat{\alpha}_n(e^{-nR}))$\newline
$= \sup\limits_{s\in (1,\infty)}\frac{s-1}{s}(R-\widetilde{I}_s^{\downarrow\downarrow}(A:B)_\rho )$.
\\ \bottomrule
\end{tabular}
\caption{Overview of certain binary quantum state discrimination problems. 
Let $\rho_{AB}$ be a quantum state.
Each row pertains to a sequence of binary quantum state discrimination problems with null hypothesis $H_0^n$ and alternative hypothesis $H_1^n$ for $n\in \mathbb{N}_{>0}$. 
In the second row, the $n$th alternative hypothesis is given by 
$\rho_A^{\otimes n}\otimes\tau_B^{\otimes n}$ 
for quantum states $\tau_B$.
In the third row, the $n$th alternative hypothesis is given by 
$\sigma_A^{\otimes n}\otimes\tau_B^{\otimes n}$ 
for quantum states $\sigma_A,\tau_B$. 
The works cited in the second and third column derive single-letter formulas for the corresponding direct and strong converse exponents. 
The formulas for the strong converse exponent are valid if 
$I(A:B)_\rho\neq \widetilde{I}_\infty^{\uparrow\uparrow}(A:B)_\rho$, 
$I(A:B)_\rho\neq \widetilde{I}_\infty^{\uparrow\downarrow}(A:B)_\rho$, 
and $I(A:B)_\rho\neq \widetilde{I}_\infty^{\downarrow\downarrow}(A:B)_\rho$, respectively. 
The lower bound in the third row is defined as 
$R_{1/2}\coloneqq I_{1/2}^{\downarrow\downarrow}(A:B)_\rho-\frac{1}{4}\frac{\partial}{\partial s^+} I_{s}^{\downarrow\downarrow}(A:B)_\rho\big|_{s=1/2}$.}
\label{tab:quantum}
\end{table}

\paragraph*{Related work.} 
The hypothesis testing problem considered here previously appeared in~\cite{berta2021composite}. 
Let us denote the minimum type-II error when the type-I error is upper bounded by $\mu\in [0,\infty)$ by $\hat{\beta}_{n}(\mu)$. 
Using this notation, it was shown that the Stein exponent satisfies 
$\lim_{n\rightarrow\infty} -\frac{1}{n}\log \hat{\beta}_{n}(\mu) = I(A:B)_\rho$ 
for all $\mu \in (0,1)$~\cite[Proposition~3.4]{berta2021composite}. 
This result is recovered immediately from our findings (see also~\cite{schmitt2026tumulainformationdoublyminimized} for a similar argument). 
However,~\cite{berta2021composite} does not address other exponents, such as the direct or strong converse exponent. 
This distinguishes our results from those in~\cite{berta2021composite}.

While our main results on the direct and strong converse exponent generalize earlier results for the classical setting~\cite{tomamichel2018operational} (see Tables~\ref{tab:classical} and~\ref{tab:quantum}), 
they do not directly follow from recent work on quantum state discrimination with composite correlated hypotheses~\cite{fang2025errorexponentsquantumstate}, as several results in~\cite{fang2025errorexponentsquantumstate} assume that the alternative hypothesis $H_1^n$ is convex. 
While the assumption of convexity is common in the literature~\cite{fang2025errorexponentsquantumstate,Brand_o_2010,Hayashi2025GeneralizedQSL,Lami_2025}, it is not satisfied in our setting, since $H_1^n$ in the last row of Table~\ref{tab:quantum} is non-convex.

\paragraph*{Outline.} 
In Section~\ref{sec:preliminaries}, we summarize the necessary mathematical preliminaries.
We begin by introducing our general notation~(\ref{ssec:notation}).
Next, we present definitions and properties related to permutation invariance~(\ref{ssec:permutation}), entropies, divergences, and R\'enyi mutual information~(\ref{ssec:divergence}), and binary quantum state discrimination~(\ref{ssec:hypothesis}). 
In Section~\ref{sec:main}, we present our main results (Theorems~\ref{thm:direct} and~\ref{thm:strong-converse}), along with results on the moderate deviation regime (Theorems~\ref{thm:moderate_direct} and~\ref{thm:moderate_strongconverse}), the Stein exponent (Corollary~\ref{cor:stein}) and the second-order asymptotics (Theorem~\ref{thm:second}).

\section{Preliminaries}\label{sec:preliminaries}

\subsection{Notation}\label{ssec:notation}
``$\log$'' is taken to refer to the natural logarithm. 
The set of natural numbers that are strictly smaller than $n\in \mathbb{N}$ is denoted by $[n]\coloneqq \{0,1,\dots, n-1\}$. 

In this paper, we work exclusively with finite-dimensional Hilbert spaces for simplicity. 
The dimension of a Hilbert space $A$ is denoted as $d_A\coloneqq \dim(A)\in \mathbb{N}_{>0}$. 
The tensor product of two Hilbert spaces $A$ and $B$ is sometimes denoted by $AB$ instead of $A\otimes B$, and $A^{n}\coloneqq A^{\otimes n}$ for any $n\in \mathbb{N}_{>0}$. 
The set of linear maps from $A$ to $A$ is denoted by $\mathcal{L}(A)$. 
Identities are sometimes left implicit; for instance, for $X_A\in\mathcal{L}(A)$, the symbol ``$X_A$'' may denote $X_A\otimes 1_B\in \mathcal{L}(A\otimes B)$. 
The rank and spectrum of $X\in \mathcal{L}(A)$ are denoted by $\rank(X)$ and $\spec(X)$, respectively. 
The support of $X\in \mathcal{L}(A)$ is denoted by $\supp(X)$ and is defined as the orthogonal complement of the kernel of $X$. 
For $X,Y\in \mathcal{L}(A)$, $X\ll Y$ is true iff the kernel of $Y$ is contained in the kernel of $X$. 
For $X,Y\in \mathcal{L}(A)$, $X\perp Y$ is true iff $XY=0=YX$.

The adjoint of $X\in \mathcal{L}(A)$ is denoted by $X^\dagger$. 
For $X\in \mathcal{L}(A)$, $X\geq 0$ is true iff $X$ is positive semidefinite, and 
$X>0$ is true iff $X$ is positive definite. 
If $X,Y\in \mathcal{L}(A)$ are self-adjoint, then $X\geq Y$ is true iff $X-Y\geq 0$. 
If $X\in \mathcal{L}(A)$ is positive semidefinite, then  $X^p$ is defined for $p\in \mathbb{R}$ by taking the power on the support of $X$. 
The operator absolute value of $X\in \mathcal{L}(A)$ is denoted by $\lvert X\rvert \coloneqq (X^\dagger X)^{1/2}$. 
The trace of $X\in \mathcal{L}(A)$ is denoted as $\tr[X]$, and the partial trace over $A$ is denoted as $\tr_A$. 
For $X\in \mathcal{L}(A)$, the Schatten $p$-norm is defined as 
$\|X \|_p\coloneqq \tr[\lvert X\rvert^p]^{1/p}$ for $p\in [1,\infty)$, and as 
$\lVert X\rVert_{\infty}\coloneqq \sqrt{\max (\spec(X^\dagger X))}$ for $p=\infty$. 
The Schatten $p$-quasi-norm is defined as $\|X \|_p\coloneqq \tr[\lvert X\rvert^p]^{1/p}$ for $p\in (0,1)$.

If $X,Y\in \mathcal{L}(A)$ are self-adjoint, then $\{X\geq Y\}$ denotes the orthogonal projection onto the subspace corresponding to the non-negative eigenvalues of $X-Y$, and $\{X<Y\}\coloneqq 1-\{X\geq Y\}$ denotes the orthogonal projection onto the subspace corresponding to the strictly negative eigenvalues of $X-Y$. 

The set of quantum states on $A$ is defined as  
$\mathcal{S}(A)\coloneqq\{\rho\in \mathcal{L}(A):\rho\geq 0,\tr[\rho]=1\}$.

\subsection{Permutation invariance}\label{ssec:permutation}
We denote the \emph{symmetric group of degree $n\in \mathbb{N}_{>0}$} by $S_n$.
The unitary operator $U(\pi)_{A^n}\in \mathcal{L}(A^{\otimes n})$ associated with a permutation $\pi \in S_n$ is defined by
\begin{equation}
U(\pi)_{A^n}\ket{\phi_1}\otimes\dots\otimes\ket{\phi_n}
= \ket{\phi_{\pi^{-1}(1)}}\otimes\dots\otimes\ket{\phi_{\pi^{-1}(n)}}
\qquad\forall \ket{\phi_1},\dots,\ket{\phi_n}\in A.
\end{equation}

The set of \emph{permutation invariant operators} is defined as 
\begin{align}
\mathcal{L}_{\sym}(A^{\otimes n})
\coloneqq \{X_{A^n}\in \mathcal{L}(A^{\otimes n}):
U(\pi)_{A^n} X_{A^n} U(\pi)_{A^n}^\dagger = X_{A^n}\,\forall \pi \in S_n\},
\end{align}
and the set of \emph{permutation invariant states} as 
$\mathcal{S}_{\sym}(A^{\otimes n})\coloneqq \mathcal{S}(A^{\otimes n})\cap\mathcal{L}_{\sym}(A^{\otimes n})$.

We denote the \emph{universal permutation invariant state}~\cite{renner2006security,christandl2009postselection,hayashi2016correlation} for $n\in \mathbb{N}_{>0}$ as
\begin{equation}
\omega_{A^n}^n \coloneqq \frac{1}{g_{n,d_A}}\tr_{{A'}^n}[(P^{n}_{\sym})_{A^n{A'}^n}],
\qquad\text{where}\qquad
g_{n,d_A}\coloneqq \binom{n+d_A^2-1}{n},
\end{equation}
$A'$ is a Hilbert space isomorphic to $A$, 
and $(P^{n}_{\sym})_{A^n{A'}^n}$ is the orthogonal projection onto the symmetric subspace of $(AA')^{\otimes n}$.

\subsection{Entropies, divergences, and R\'enyi mutual information}\label{ssec:divergence}
For $\rho\in \mathcal{S}(A)$, the \emph{von Neumann entropy} is defined as $H(A)_\rho\coloneqq -\tr[\rho\log \rho]$.
The \emph{R\'enyi entropy (of order $\alpha$)} is defined as
$H_\alpha (A)_\rho
\coloneqq\frac{1}{1-\alpha}\log \tr[\rho^\alpha]$
for $\alpha\in (-\infty,1)\cup (1,\infty)$, 
and for $\alpha \in \{1,\infty\}$ as the corresponding limits. 
For $\rho\in \mathcal{S}(AB)$, 
the \emph{mutual information} between $A$ and $B$ is $I(A:B)_\rho \coloneqq H(A)_\rho+H(B)_\rho-H(AB)_\rho$. 
The \emph{mutual information variance} is defined as~\cite{tomamichel2013hierarchy,li2014second,hayashi2016correlation}
\begin{align}\label{eq:def-variance}
V(A:B)_\rho\coloneqq \tr[\rho_{AB}(\log \rho_{AB}-\log (\rho_A\otimes\rho_B)-I(A:B)_\rho)^2].
\end{align}

The \emph{Petz (quantum R\'enyi) divergence (of order $\alpha$)} is defined for $\alpha\in (0,1)\cup (1,\infty),\rho\in \mathcal{S}(A)$, and any positive semidefinite $\sigma\in \mathcal{L}(A)$ as~\cite{petz1986quasi}
\begin{equation}
D_\alpha (\rho\| \sigma)\coloneqq\frac{1}{\alpha -1} \log \tr [\rho^\alpha \sigma^{1-\alpha}]
\end{equation}
if $(\alpha <1\land \rho\not\perp\sigma)\lor \rho\ll \sigma$ and 
$D_\alpha (\rho\| \sigma)\coloneqq \infty$ else. 
In addition, $D_0$ and $D_1$ are defined as the respective limits of $D_\alpha$ as $\alpha\rightarrow 0$ and $\alpha\rightarrow 1$.

The \emph{sandwiched (quantum R\'enyi) divergence (of order $\alpha$)} is defined for $\alpha\in (0,1)\cup (1,\infty),\rho\in \mathcal{S}(A)$, and any positive semidefinite $\sigma\in \mathcal{L}(A)$ as~\cite{mueller2013quantum,wilde2014strong}
\begin{equation}
\widetilde{D}_\alpha (\rho\| \sigma)\coloneqq
\frac{1}{\alpha -1} \log \tr [(\sigma^{\frac{1-\alpha}{2\alpha}} \rho \sigma^{\frac{1-\alpha}{2\alpha}})^\alpha]
\end{equation}
if $(\alpha <1\land \rho\not\perp\sigma)\lor \rho\ll \sigma$ and 
$\widetilde{D}_\alpha (\rho\| \sigma)\coloneqq\infty$ else. 
In addition, $\widetilde{D}_1$ and $\widetilde{D}_\infty$ are defined as the respective limits of $\widetilde{D}_\alpha$ as $\alpha\rightarrow 1$ and $\alpha\rightarrow \infty$.

For $\rho_{AB}\in \mathcal{S}(AB)$, the \emph{doubly minimized PRMI (of order $\alpha$)} is defined for any $\alpha\in [0,\infty)$ as
\begin{align}
I_\alpha^{\downarrow\downarrow} (A:B)_\rho 
&\coloneqq \inf_{\substack{\sigma_A\in \mathcal{S}(A),\\ \tau_B\in \mathcal{S}(B)}} D_\alpha (\rho_{AB}\| \sigma_A\otimes \tau_B ),
\label{eq:prmi2}
\end{align}
and the \emph{doubly minimized SRMI (of order $\alpha$)} is defined for any $\alpha\in (0,\infty]$ as 
\begin{align}
\widetilde{I}^{\downarrow\downarrow}_\alpha (A:B)_\rho 
&\coloneqq \inf_{\substack{\sigma_A\in \mathcal{S}(A),\\ \tau_B\in \mathcal{S}(B)}}\widetilde{D}_\alpha (\rho_{AB}\| \sigma_A\otimes \tau_B ).
\label{eq:srmi2}
\end{align}

\subsection{Binary quantum state discrimination}\label{ssec:hypothesis}
In this section, we will explain the notation related to binary discrimination. 
First, we will describe the classical setting, then the quantum setting, and then we will explain how the classical setting can be regarded as a special case of the quantum setting under suitable conditions. 

\paragraph*{Classical setting.} Classical binary discrimination pertains to a scenario where one is given a PMF $S_X$ over $\mathcal{X}$ that is an element of $H_0$ or $H_1$, both of which are non-empty subsets of the set of PMFs over $\mathcal{X}$. 
The task is to decide which is true: $S_X\in H_0$ (the \emph{null hypothesis}) or $S_X\in H_1$ (the \emph{alternative hypothesis}). 
The \emph{test} is a function $T:\mathcal{X}\rightarrow [0,1]$, and the decision is determined by the binary test $(T,1-T)$ applied to $S_X$. 
If the event corresponding to $T$ occurs, then the decision is made that the null hypothesis is true.
Conversely, if the event corresponding to $1-T$ occurs, then the decision is made that the alternative hypothesis is true.
This paper solely addresses discrimination problems with a \emph{simple} null hypothesis, i.e., $H_0=\{P_X\}$ for some PMF $P_X$.

We are mainly concerned with \emph{sequences} of binary discrimination problems of the following form, for a fixed PMF $P_{XY}$ over $\mathcal{X}\times\mathcal{Y}$. 
For each $n\in \mathbb{N}_{>0}$, the null hypothesis is $H_0^n=\{P_{XY}^{\times n} \}$, the alternative hypothesis $H_1^n$ is a non-empty subset of the set of PMFs over $(\mathcal{X}\times\mathcal{Y})^{\times n}$, and the test is a function 
$T^n:(\mathcal{X}\times \mathcal{Y})^{\times n}\rightarrow [0,1]$. 
The \emph{type-I error} and the \emph{(worst case) type-II error} are then, respectively,
\begin{align}
\alpha_n(T^n)&\coloneqq \sum_{\substack{x_1,\dots, x_n\in \mathcal{X}, \\ y_1,\dots, y_n\in \mathcal{Y}}} 
P_{XY}(x_1,y_1) \cdot \ldots \cdot P_{XY}(x_n,y_n)(1-T^n(x_1,y_1,\dots, x_n,y_n)) ,
\label{eq:classical-type-I}\\
\beta_n(T^n)&\coloneqq \sup_{Q_{X^nY^n}\in H_1^n}\sum_{\substack{x_1,\dots, x_n\in \mathcal{X}, \\ y_1,\dots, y_n\in \mathcal{Y}}} 
Q_{X^nY^n}(x_1,y_1,\dots, x_n,y_n)T^n(x_1,y_1,\dots, x_n,y_n).
\label{eq:classical-type-II}
\end{align}
For the $n$th hypothesis testing problem, the \emph{minimum} type-I error when the type-II error is upper bounded by $\mu \in [0,\infty)$ is denoted by
\begin{equation}\label{eq:alpha-classical}
\hat{\alpha}_n(\mu)\coloneqq \inf_{T^n}\{\alpha_n(T^n): \beta_n(T^n)\leq \mu \},
\end{equation}
where the minimization is over all functions 
$T^n:(\mathcal{X}\times \mathcal{Y})^{\times n}\rightarrow [0,1]$. 

\paragraph*{Quantum setting.} 
Binary quantum state discrimination pertains to a scenario where one is given a quantum state $\xi\in \mathcal{S}(A)$ that is an element of $H_0$ or $H_1$, both of which are non-empty subsets of $\mathcal{S}(A)$. 
The task is to decide which is true: $\xi\in H_0$ (the \emph{null hypothesis}) or $\xi\in H_1$ (the \emph{alternative hypothesis}). 
The \emph{test} is some $T\in \mathcal{L}(A)$ that satisfies $0\leq T\leq 1$. 
The decision is determined by the binary measurement $(T,1-T)$ applied to $\xi$. 
If the measurement outcome is the one associated with $T$, then the decision is made that the null hypothesis is true. 
Conversely, if the measurement outcome is the one associated with $1-T$, then the decision is made that the alternative hypothesis is true.
This paper solely addresses quantum state discrimination problems with a \emph{simple} null hypothesis, i.e., $H_0=\{\rho\}$ for some $\rho\in \mathcal{S}(A)$. 

We are mainly concerned with \emph{sequences} of binary quantum state discrimination problems of the following form, for a fixed $\rho_{AB}\in \mathcal{S}(AB)$. 
For each $n\in \mathbb{N}_{>0}$, the null hypothesis is $H_0^n=\{\rho_{AB}^{\otimes n} \}$, the alternative hypothesis $H_1^n$ is a non-empty subset of $\mathcal{S}(A^nB^n)$, and the test is some $T^n_{A^nB^n}\in \mathcal{L}(A^nB^n)$ that satisfies $0\leq T^n_{A^nB^n}\leq 1$.
The \emph{type-I error} and the \emph{(worst case) type-II error} are then, respectively,
\begin{align}
\alpha_n(T^n_{A^nB^n})&\coloneqq \tr[\rho_{AB}^{\otimes n} (1-T^n_{A^nB^n})],
\\
\beta_n(T^n_{A^nB^n})&\coloneqq \sup_{\sigma_{A^nB^n}\in H_1^n} \tr[\sigma_{A^nB^n} T^n_{A^nB^n}].
\end{align}
For the $n$th hypothesis testing problem, the \emph{minimum} type-I error when the type-II error is upper bounded by $\mu \in [0,\infty)$ is denoted by
\begin{equation}\label{eq:alpha-quantum}
\hat{\alpha}_n(\mu)\coloneqq \inf_{\substack{T^n_{A^nB^n}\in \mathcal{L} (A^nB^n): \\ 0\leq T^n_{A^nB^n}\leq 1}}\{\alpha_n(T^n_{A^nB^n}): \beta_n(T^n_{A^nB^n})\leq \mu \}.
\end{equation}

In order to quantify the trade-off between type-I and type-II errors in the asymptotic limit where $n\rightarrow\infty$, we define the following error exponents.~\cite{tomamichel2018operational,mosonyi2022error}
\begin{itemize}
\item The \emph{direct exponent} with respect to $R\in [0,\infty)$ is
$\liminf_{n\rightarrow\infty } -\frac{1}{n}\log \hat{\alpha}_n(e^{-nR})$
if this limit exists, 
and $+\infty$ else. 
($R$ is referred to as the \emph{type-II rate}.)
\item The \emph{strong converse exponent} with respect to $R\in [0,\infty)$ is 
$\limsup_{n\rightarrow\infty} -\frac{1}{n}\log (1- \hat{\alpha}_n (e^{-nR}))$
if this limit exists, 
and $+\infty$ else. 
($R$ is referred to as the \emph{type-II rate}.)
\item The \emph{threshold rate} (or: \emph{Stein exponent}) is 
$\sup\{R\in \mathbb{R}:\limsup_{n\rightarrow\infty}\hat{\alpha}_n(e^{-nR})=0\}$.
\item The \emph{strong converse threshold rate} is 
$\inf\{R\in \mathbb{R}:\liminf_{n\rightarrow\infty}\hat{\alpha}_n(e^{-nR})=1\}$
if this infimum exists, 
and $+\infty$ else.
\end{itemize}

\paragraph*{CC states.} 
Binary \emph{quantum} state discrimination can be compared to \emph{classical} binary discrimination by restricting the former to the special case of classical-classical (CC) states, as we will elaborate in the following. 
Let $\mathcal{X}\coloneqq [d_A],\mathcal{Y}\coloneqq [d_B]$,
let $P_{XY}$ be a PMF over $\mathcal{X}\times\mathcal{Y}$, and let 
$\rho_{AB}\coloneqq\sum_{x\in \mathcal{X},y\in \mathcal{Y}}P_{XY}(x,y)\proj{a_x,b_y}_{AB}$, 
where $\{\ket{a_x}_A\}_{x\in [d_A]}$ and $\{\ket{b_y}_B\}_{y\in [d_B]}$ are orthonormal bases for $A$ and $B$, respectively. 

Consider a sequence of binary \emph{quantum} state discrimination problems as above with 
null hypothesis $H_0^{\mathrm{q},n}=\{\rho_{AB}^{\otimes n}\}$ and alternative hypothesis $H_1^{\mathrm{q},n}$, 
and let the function in~\eqref{eq:alpha-quantum} be denoted by $\hat{\alpha}_n^{\mathrm{q}}$. 
For notational convenience, let us define for any $\sigma_{A^nB^n}\in H_1^{\mathrm{q},n}$
\begin{align}
Q^\sigma_{X^nY^n}(x_1,y_1,\dots, x_n,y_n)\coloneqq 
\bra{a_{x_1},b_{y_1},\dots, a_{x_n},b_{y_n}} \sigma_{A^nB^n} \ket{a_{x_1},b_{y_1},\dots, a_{x_n},b_{y_n}}
\end{align}
for all $x_1,\dots,x_n\in \mathcal{X},y_1,\dots, y_n\in \mathcal{Y}$. 
To facilitate subsequent comparison with the classical setting, suppose that the alternative hypothesis is closed under pinching with respect to the eigenbasis of the CC state $\rho_{AB}$, i.e., 
for all $\sigma_{A^nB^n}\in H_1^{\mathrm{q},n}:$
\begin{align}\label{eq:analogy1}
\sum_{\substack{x_1,\dots, x_n\in \mathcal{X},\\y_1,\dots, y_n\in \mathcal{Y}}}
Q^\sigma_{X^nY^n}(x_1,y_1,\dots, x_n,y_n)
\proj{a_{x_1},b_{y_1},\dots, a_{x_n},b_{y_n}}
\in H_1^{\mathrm{q},n}.
\end{align}
Consider now a sequence of \emph{classical} binary discrimination problems as above with null hypothesis $H_0^{\mathrm{c},n}=\{P_{XY}^{\times n}\}$ 
and alternative hypothesis $H_1^{\mathrm{c},n}$, 
and let the function in~\eqref{eq:alpha-classical} be denoted by $\hat{\alpha}_n^{\mathrm{c}}$. 
Suppose that the alternative hypotheses for the quantum and the classical setting are \emph{analogous} in the sense that
\begin{align}\label{eq:analogy2}
\bigcup_{\sigma_{A^nB^n}\in H_1^{\mathrm{q},n}} \{Q^\sigma_{X^nY^n}\}
=H_1^{\mathrm{c},n}.
\end{align}
Assuming that the conditions in~\eqref{eq:analogy1} and~\eqref{eq:analogy2} are satisfied, 
it is straightforward to derive that 
$\hat{\alpha}_n^{\mathrm{q}}(\mu)=\hat{\alpha}_n^{\mathrm{c}}(\mu)$ for all $\mu\in [0,\infty)$, 
as shown in Appendix~\ref{app:classical}. 
Thus, classical binary discrimination problems can be regarded as special cases of binary quantum state discrimination problems. 

\begin{rem}[Extension of statements in Table~\ref{tab:classical}]\label{rem:extension}
The conditions in~\eqref{eq:analogy1} and~\eqref{eq:analogy2} are satisfied for the examples in Table~\ref{tab:classical} and Table~\ref{tab:quantum}. 
Therefore, the classical discrimination problems in Table~\ref{tab:classical} can be regarded as special cases of the corresponding binary quantum state discrimination problems in Table~\ref{tab:quantum}, so propositions in Table~\ref{tab:quantum} imply corresponding propositions in Table~\ref{tab:classical}. 
This implies that the equality for the direct exponent in the second row of Table~\ref{tab:classical} can be extended to $R\in (0,\infty)$, 
and that the equalities for the strong converse exponent in the second and third row of Table~\ref{tab:classical} can be extended to $R\in [0,\infty)$. 
\end{rem}

\section{Main results}\label{sec:main}

\paragraph*{Correlation detection.} Let $\rho_{AB}$ be a bipartite quantum state that is correlated, i.e., $\rho_{AB}\neq \rho_A\otimes \rho_B$. 
How well can this quantum state be distinguished from any uncorrelated quantum state $\sigma_A\otimes \tau_B$? 
To make this question more precise, we will use the terminology for binary quantum state discrimination, as outlined above in Section~\ref{ssec:hypothesis}. 
Consider the null hypothesis $H_0\coloneqq \{\rho_{AB}\}$ and the alternative hypothesis 
$H_1\coloneqq \{\sigma_A\otimes \tau_B:\sigma_A\in \mathcal{S}(A),\tau_B\in \mathcal{S}(B)\}$. 
Then the minimum type-I error when the type-II error is upper bounded by $\mu\in [0,\infty)$ is
\begin{align}\label{eq:corr-det}
\min_{\substack{T_{AB}\in \mathcal{L} (AB): \\ 0\leq T_{AB}\leq 1}}
\{\tr[\rho_{AB}(1-T_{AB})]: 
\max_{\substack{\sigma_{A} \in \mathcal{S}(A),\\ \tau_B \in \mathcal{S}(B)}}
\tr[\sigma_{A}\otimes \tau_{B}\, T_{AB}]
\leq \mu\}.
\end{align}
\eqref{eq:corr-det} is the minimum probability with which one erroneously decides that the given quantum state is uncorrelated, under the constraint that the probability with which one erroneously decides that the given quantum state is $\rho_{AB}$ is upper bounded by $\mu$. 

The goal of this section is to show that the single-letter formulas for the direct and the strong converse exponent in the third row of Table~\ref{tab:classical}, where the doubly minimized RMI occurs, can be generalized from the classical to the quantum setting by means of the doubly minimized PRMI and the doubly minimized SRMI, respectively. 
Accordingly, we are interested in any one of the following variants of correlation detection as described in~\eqref{eq:corr-det}. 
It should be noted that several variants of the problem are introduced, as all of them will be found to have the same direct and strong converse exponents.

First, one may consider the i.i.d. version of~\eqref{eq:corr-det}, where 
$H_0^n\coloneqq \{\rho_{AB}^{\otimes n}\}$ and 
$H_1^n\coloneqq \{\sigma_A^{\otimes n}\otimes \tau_B^{\otimes n}:\sigma_A\in \mathcal{S}(A),\tau_B\in \mathcal{S}(B)\}$. 
The minimum type-I error is then 
\begin{align}
\hat{\alpha}_{n,\rho}^{\mathrm{iid}}(\mu )
&\coloneqq \min_{\substack{T^n_{A^nB^n}\in \mathcal{L} (A^nB^n): \\ 0\leq T^n_{A^nB^n}\leq 1}}
\{\tr[\rho_{AB}^{\otimes n}(1-T^n_{A^nB^n})]:
\max_{\substack{\sigma_{A} \in \mathcal{S}(A),\\ \tau_B \in \mathcal{S}(B)}}
\tr[\sigma_{A}^{\otimes n}\otimes \tau_{B}^{\otimes n}\, T^n_{A^nB^n}]
\leq \mu\},
\label{eq:def-alpha-otimes}
\end{align}
which shall be defined for any $n\in \mathbb{N}_{>0},\rho_{AB}\in \mathcal{S}(AB),\mu\in [0,\infty)$.

Second, one may impose the i.i.d. assumption on the null hypothesis only, and take the alternative hypothesis to be given by all states that are uncorrelated between $A^n$ and $B^n$, and permutation invariant on both $A^n$ and $B^n$, i.e., 
$H_0^n\coloneqq \{\rho_{AB}^{\otimes n}\}$ and 
$H_1^n\coloneqq \{\sigma_{A^n}\otimes \tau_{B^n}:\sigma_{A^n}\in \mathcal{S}_{\sym}(A^{\otimes n}),\tau_{B^n}\in \mathcal{S}_{\sym}(B^{\otimes n})\}$. 
The minimum type-I error is then 
\begin{align}
\hat{\alpha}_{n,\rho}(\mu)
&\coloneqq \min_{\substack{T^n_{A^nB^n}\in \mathcal{L} (A^nB^n): \\ 0\leq T^n_{A^nB^n}\leq 1}}
\{\tr[\rho_{AB}^{\otimes n}(1-T^n_{A^nB^n})]:
\max_{\substack{\sigma_{A^n} \in \mathcal{S}_{\sym}(A^{\otimes n}),\\ \tau_{B^n} \in \mathcal{S}_{\sym}(B^{\otimes n})}}
\tr[\sigma_{A^n}\otimes \tau_{B^n} T^n_{A^nB^n}]\leq \mu\}.
\label{eq:def-alpha}
\end{align}

Third, the second option may be modified by imposing the permutation invariance constraint on $A^n$ only, i.e., 
$H_0^n\coloneqq \{\rho_{AB}^{\otimes n}\}$ and 
$H_1^n\coloneqq \{\sigma_{A^n}\otimes \tau_{B^n}:\sigma_{A^n}\in \mathcal{S}_{\sym}(A^{\otimes n}),\tau_{B^n}\in \mathcal{S}(B^n)\}$. 
The minimum type-I error then matches that of the second option,
\begin{align}\label{eq:alpha-sym-an}
\hat{\alpha}_{n,\rho}(\mu )
&= \min_{\substack{T^n_{A^nB^n}\in \mathcal{L} (A^nB^n): \\ 0\leq T^n_{A^nB^n}\leq 1}}
\{\tr[\rho_{AB}^{\otimes n}(1-T^n_{A^nB^n})]:
\max_{\substack{\sigma_{A^n} \in \mathcal{S}_{\sym}(A^{\otimes n}),\\ \tau_{B^n} \in \mathcal{S}(B^n)}}
\tr[\sigma_{A^n}\otimes \tau_{B^n} T^n_{A^nB^n}]
\leq \mu\}.
\end{align}
The proof of the equality in~\eqref{eq:alpha-sym-an} is given in Appendix~\ref{app:lemmas}, see Lemma~\ref{lem:order}~(c).

\subsection{Operational interpretation of doubly minimized PRMI from direct exponent}\label{ssec:direct}
\paragraph*{Problem formulation.} The problem we are interested in is to find a single-letter formula for the direct exponent of the quantum state discrimination problems associated with $\hat{\alpha}_{n,\rho}^{\mathrm{iid}}$ and $\hat{\alpha}_{n,\rho}$ as defined in~\eqref{eq:def-alpha-otimes} and~\eqref{eq:def-alpha}.
This is accomplished in the following theorem.
It shows that if the type-II rate $R$ is sufficiently large but below the threshold given by $I(A:B)_\rho$, then the minimum type-I error decreases to $0$ exponentially fast in $n$, and the corresponding exponent is determined by the family of the doubly minimized PRMIs of order $s\in (\frac{1}{2},1)$. 

\begin{thm}[Direct exponent]\label{thm:direct}
Let $\rho_{AB}\in \mathcal{S}(AB)$ and let
\begin{equation}\label{eq:def-r12}
R_{1/2}\coloneqq I_{1/2}^{\downarrow\downarrow}(A:B)_\rho-\frac{1}{4}\frac{\partial}{\partial s^+} I_{s}^{\downarrow\downarrow}(A:B)_\rho\big|_{s=1/2}
\in [I_{0}^{\downarrow\downarrow}(A:B)_\rho,I_{1/2}^{\downarrow\downarrow}(A:B)_\rho].
\end{equation}
Then, for any $R\in (R_{1/2},\infty)$
\begin{equation}\label{eq:direct}
\lim_{n\rightarrow\infty}-\frac{1}{n}\log \hat{\alpha}_{n,\rho}(e^{-nR})
=\sup_{s\in (\frac{1}{2},1)}\frac{1-s}{s}(I_s^{\downarrow\downarrow}(A:B)_\rho - R),
\end{equation}
and the same is true if $\hat{\alpha}_{n,\rho}$ in~\eqref{eq:direct} is replaced by 
$\hat{\alpha}_{n,\rho}^{\mathrm{iid}}$.

Furthermore, for any $R\in (0,\infty)$, the right-hand side of~\eqref{eq:direct} lies in $[0,\max(0,I(A:B)_\rho -R)]$, 
and it is strictly positive iff $R<I(A:B)_\rho$.
\end{thm}
\begin{proof}
See Appendix~\ref{app:lemmas} and Appendix~\ref{app:hoeffding}.
\end{proof}

The proof of Theorem~\ref{thm:direct} is divided into two parts: a proof of achievability and a proof of optimality. 
The proof of achievability uses a quantum Neyman-Pearson test that compares 
$\rho_{AB}^{\otimes n}$ with $\omega_{A^n}^n\otimes \omega_{B^n}^n$, 
and leverages the asymptotic optimality of the universal permutation invariant state, which has been established in \cite[Theorem~3~(l)]{burri2025prmisrmi1}. 
This proof method is an adapted version of an analogous proof of achievability for the singly minimized PRMI~\cite[Section~V.A]{hayashi2016correlation}. 
The proof of optimality employs techniques for classical binary hypothesis testing from~\cite{tomamichel2018operational}, and makes use of several properties of the doubly minimized PRMI, including \cite[Theorem~3~(j), (p), (q)]{burri2025prmisrmi1}.

\begin{rem}[Necessity of permutation invariance of alternative hypothesis]\label{rem:permutation_prmi} 
Consider the following variant of $\hat{\alpha}_{n,\rho}$ where the alternative hypothesis is constrained only by the independence of $A^n$ and $B^n$.
\begin{align}\label{eq:def-alpha-ind}
\hat{\alpha}_{n,\rho}^{\mathrm{ind}}(\mu)
&\coloneqq \min_{\substack{T^n_{A^nB^n}\in \mathcal{L} (A^nB^n): \\ 0\leq T^n_{A^nB^n}\leq 1}}
\{\tr[\rho_{AB}^{\otimes n}(1-T^n_{A^nB^n})]: 
\max_{\substack{\sigma_{A^n} \in \mathcal{S}(A^n),\\ \tau_{B^n} \in \mathcal{S}(B^n)}}
\tr[\sigma_{A^n}\otimes \tau_{B^n} T^n_{A^nB^n}]\leq \mu\}
\end{align}
In light of the equality in~\eqref{eq:alpha-sym-an}, it is natural to inquire whether Theorem~\ref{thm:direct} remains valid when $\hat{\alpha}_{n,\rho}$ is replaced by $\hat{\alpha}^{\mathrm{ind}}_{n,\rho}$. 
This is not the case; an explicit counterexample where $\rho_{AB}$ is a correlated CC state is provided in Appendix~\ref{app:example_prmi}. 
\end{rem}

The proof of achievability for Theorem~\ref{thm:direct} yields the following corollary.

\begin{cor}[Asymptotic minimum type-I error]\label{cor:limit-type1_direct}
Let $\rho_{AB}\in \mathcal{S}(AB)$ and let $R\in (-\infty,I(A:B)_\rho)$. Then 
$\lim_{n\rightarrow\infty}\hat{\alpha}_{n,\rho}(e^{-nR})=0$.
Moreover, the same is true if $\hat{\alpha}_{n,\rho}$ is replaced by 
$\hat{\alpha}_{n,\rho}^{\mathrm{iid}}$.
\end{cor}
\begin{proof}
See Appendix~\ref{app:cor_direct}.
\end{proof}

The proof of Theorem~\ref{thm:direct} implies the following corollary, which can be seen as an alternative formulation of Theorem~\ref{thm:direct}. 
\begin{cor}[Direct exponent]\label{cor:direct}
Let $\rho_{AB}\in \mathcal{S}(AB)$ and let
\begin{align}\label{eq:def-R(s)-direct}
R:(1/2,1]\rightarrow [0,I(A:B)_\rho],\quad 
s\mapsto 
I_s^{\downarrow\downarrow}(A:B)_\rho-s(1-s)\frac{\mathrm{d}}{\mathrm{d} s}I_s^{\downarrow\downarrow}(A:B)_\rho.
\end{align}
Then $R$ is continuous and monotonically increasing. 
Let $R_{1/2}\coloneqq \lim_{s\rightarrow 1/2^+}R(s)$ and 
$R(1/2)\coloneqq R_{1/2}$. 
Let $s_{1/2}\coloneqq \max\{s\in [\frac{1}{2},1]:R(s)=R_{1/2}\}$ and 
$s_1\coloneqq\min\{s\in [\frac{1}{2},1]:R(s)=I(A:B)_\rho\}$.

Suppose $I_{1/2}^{\downarrow\downarrow}(A:B)_\rho\neq I(A:B)_\rho$. 
Then, $\frac{1}{2}= s_{1/2}<s_1\leq 1$ and for any $s\in (s_{1/2},s_1)$
\begin{equation}\label{eq:direct-cor}
\lim_{n\rightarrow\infty}-\frac{1}{n}\log \hat{\alpha}_{n,\rho}(e^{-nR(s)})
=\frac{1-s}{s}(I_s^{\downarrow\downarrow}(A:B)_\rho - R(s))
=(1-s)^2\frac{\mathrm{d}}{\mathrm{d}s} I_s^{\downarrow\downarrow}(A:B)_\rho.
\end{equation}
Moreover, the same is true if $\hat{\alpha}_{n,\rho}$ in~\eqref{eq:direct-cor} is replaced by 
$\hat{\alpha}_{n,\rho}^{\mathrm{iid}}$.
\end{cor}
\begin{proof}
See Appendix~\ref{app:proof-cor_direct}.
\end{proof}

\begin{rem}[Direct exponent, assuming non-zero mutual information variance]\label{rem:direct} 
The formulation of Corollary~\ref{cor:direct} simplifies if $\rho_{AB}$ is assumed to have non-zero mutual information variance. 
More specifically, as proven in Appendix~\ref{app:rem_direct}: 
\emph{For any $\rho_{AB}\in \mathcal{S}(AB)$ such that $V(A:B)_\rho\neq 0$ holds~\eqref{eq:direct-cor} for all $s\in (\frac{1}{2},1)$, 
and the same is true if $\hat{\alpha}_{n,\rho}$ in~\eqref{eq:direct-cor} is replaced by 
$\hat{\alpha}_{n,\rho}^{\mathrm{iid}}$.} 
Note that there exist states whose mutual information variance is zero. 
For instance, $V(A:B)_\rho=0$ 
if $\rho_{AB}$ is a product state, 
if $\rho_{AB}$ is a pure state whose non-zero Schmidt coefficients are all equal to each other, 
or if $\rho_{AB}$ is a copy-CC state with a flat probability distribution (i.e., $\exists x\in \mathcal{X}: \forall y\in \mathcal{X}: P_X(y)=P_X(x) $ or $P_X(y)=0$). 
Characterizing the class of states whose mutual information variance is zero is left as an open problem. 
In particular, we leave open the question of whether any $\rho_{AB}\in \mathcal{S}(AB)$ such that 
$I_{1/2}^{\downarrow\downarrow}(A:B)_\rho\neq I(A:B)_\rho$ has non-zero mutual information variance. 
\end{rem}

The following result is a corollary of Theorem~\ref{thm:direct} and Corollary~\ref{cor:direct}. It provides a single-letter formula for the forward $\beta$-cutoff rate. The concept of cutoff rates in the context of hypothesis testing has been originally introduced in~\cite{csiszar1995generalized} (see also~\cite{alajaji2004csiszars}). 
To define the forward $\beta$-cutoff rate, one may view the direct exponent as a function of $R$ (which is monotonically decreasing) and consider linear functions of $R$ that are lower bounds on this curve. 
The forward $\beta$-cutoff rate then characterizes the best such linear lower bound. 
More explicitly, it identifies the largest $R_0$ such that the direct exponent is bounded below by a line of slope $\beta$ intersecting the $x$-axis at $R=R_0$, i.e., 
the direct exponent behaves no worse than $\beta (R-R_0)$ for all relevant rates $R$. 
The parameter $\beta\in (-1,0)$ controls the steepness of this bound: 
values closer to $0$ yield flatter bounds, while values closer to $-1$ correspond to steeper ones. 
Thus, the forward $\beta$-cutoff rate is the $x$-axis intercept of the best supporting line of slope $\beta$ to the direct exponent.

\begin{cor}[Forward $\beta$-cutoff rate]\label{cor:forward-cutoff}
Let $\rho_{AB}\in \mathcal{S}(AB)$. 
Let $R_{1/2},s_{1/2},s_1$ be defined as in Corollary~\ref{cor:direct}. 
For all $\alpha\in \{1/2,1\}$, let $\beta_\alpha\coloneqq 1-\frac{1}{s_\alpha}$. 
For all $\beta\in (-1,0)$, let
\begin{align}\label{eq:r0-f}
R_0^{(f)}(\beta)_\rho
\coloneqq \sup \{R_0\in [0,\infty):\liminf_{n\rightarrow\infty} -\frac{1}{n}\log \hat{\alpha}_{n,\rho}(e^{-nR}) \geq \beta (R-R_0) \quad \forall R\in (R_{1/2},\infty) \}.
\end{align}

Suppose $I_{1/2}^{\downarrow\downarrow}(A:B)_\rho\neq I(A:B)_\rho$. 
Then, $-1=\beta_{1/2}<\beta_1\leq 0$ and for all $\beta\in (\beta_{1/2},\beta_1)$
\begin{align}\label{eq:r0f}
R_0^{(f)}(\beta)_\rho=I_{\frac{1}{1-\beta}}^{\downarrow\downarrow}(A:B)_\rho .
\end{align}
Moreover, the same is true if $\hat{\alpha}_{n,\rho}$ in~\eqref{eq:r0-f} is replaced by 
$\hat{\alpha}_{n,\rho}^{\mathrm{iid}}$.
\end{cor}
\begin{proof}
See Appendix~\ref{ssec:forward-cutoff}.
\end{proof}

\begin{rem}[Forward $\beta$-cutoff rate, assuming non-zero mutual information variance]\label{rem:cutoff_f} 
The formulation of Corollary~\ref{cor:forward-cutoff} simplifies if $\rho_{AB}$ is assumed to have non-zero mutual information variance. 
For such states, $s_1=1$ (see proof for Remark~\ref{rem:direct} in Appendix~\ref{app:rem_direct}), so $\beta_1=0$. 
Corollary~\ref{cor:forward-cutoff} is therefore simplified to: 
\emph{For any $\rho_{AB}\in \mathcal{S}(AB)$ such that $V(A:B)_\rho\neq 0$ holds~\eqref{eq:r0f} for all $\beta\in (-1,0)$, 
and the same is true if $\hat{\alpha}_{n,\rho}$ in~\eqref{eq:r0-f} is replaced by 
$\hat{\alpha}_{n,\rho}^{\mathrm{iid}}$.} 
\end{rem}

The significance of Corollary~\ref{cor:forward-cutoff} (and Remark~\ref{rem:cutoff_f}) lies in the fact that it provides an operational interpretation for individual members of the family of the doubly minimized PRMI of order $\alpha\in (\frac{1}{2},1)$. 
In contrast, Theorem~\ref{thm:direct} only provides an operational interpretation for the entire family of the doubly minimized PRMI of order $\alpha\in (\frac{1}{2},1)$ due to the presence of a supremum on the right-hand side of~\eqref{eq:direct}.

For completeness, we also study the moderate deviation regime below the threshold $I(A:B)_\rho$. 
Namely, we investigate the minimum type-I error when the type-II error decays as $e^{-nR_n}$, 
where $R_n\coloneqq I(A:B)_\rho-a_n$ and $(a_n)_{n\in \mathbb{N}}$ is a moderate sequence, i.e., a strictly positive sequence satisfying $a_n\rightarrow 0$ and $\sqrt{n}a_n\rightarrow\infty$ as $n\rightarrow\infty$. 
Paralleling the results of~\cite{Cheng2017ModerateDA,Chubb_2017} for binary quantum state discrimination in the i.i.d. setting, we show that the minimum type-I error behaves as $\exp(-\frac{na_n^2}{2V(A:B)_\rho}+o(na_n^2))$, as formalized in the following theorem. 

\begin{thm}[Moderate deviation analysis below threshold]\label{thm:moderate_direct}
Let $(a_n)_{n\in \mathbb{N}}$ be a strictly positive real sequence such that 
$\lim_{n\rightarrow\infty}a_n=0$ and $\lim_{n\rightarrow\infty}\sqrt{n}a_n=\infty$. 
Let $\rho_{AB}\in \mathcal{S}(AB)$ be such that $V(A:B)_\rho>0$. Then,
\begin{align}\label{eq:mod-thm-d}
\lim_{n\rightarrow\infty} -\frac{1}{na_n^2}\log \hat{\alpha}_{n,\rho}(e^{-nR_n})
=\frac{1}{2V(A:B)_\rho},
\end{align}
where $R_n\coloneqq I(A:B)_\rho-a_n$ for all $n\in \mathbb{N}$. 
Moreover, the same is true if $\hat{\alpha}_{n,\rho}$ in~\eqref{eq:mod-thm-d} is replaced by $\hat{\alpha}_{n,\rho}^{\mathrm{iid}}$.
\end{thm}
\begin{proof}
See Appendix~\ref{proof:moderate_direct}.
\end{proof}

The proof of Theorem~\ref{thm:moderate_direct} consists of two parts: a proof of achievability and a proof of optimality. 
The proof of achievability follows by refining the direct exponent analysis in the moderate deviation regime, using a local expansion~\cite{burri2025prmisrmi1} of the doubly minimized PRMI of order $\alpha$ around $\alpha=1$. 
The proof of optimality is a direct consequence of the moderate deviation analysis in the i.i.d. setting~\cite{Cheng2017ModerateDA,Chubb_2017}.

\subsection{Operational interpretation of doubly minimized SRMI from strong converse exponent}\label{ssec:strongconverse}
\paragraph*{Problem formulation.} 
The problem we are interested in is to find a single-letter formula for the strong converse exponent of the binary quantum state discrimination problems associated with $\hat{\alpha}_{n,\rho}^{\mathrm{iid}}$ and $\hat{\alpha}_{n,\rho}$ as defined in~\eqref{eq:def-alpha-otimes} and~\eqref{eq:def-alpha}. 
This is accomplished in the following theorem. 
It shows that if the type-II rate $R$ exceeds the threshold given by $I(A:B)_\rho$, then the minimum type-I error goes to $1$ exponentially fast in $n$, and the corresponding exponent is determined by the family of the doubly minimized SRMIs of order $s\in (1,\infty)$. 

\begin{thm}[Strong converse exponent]\label{thm:strong-converse}
Let $\rho_{AB}\in \mathcal{S}(AB)$ be such that 
$I(A:B)_\rho\neq \widetilde{I}_\infty^{\downarrow\downarrow}(A:B)_\rho$. 
For any $R\in [0,\infty)$ 
\begin{equation}\label{eq:strong-converse}
\lim_{n\rightarrow\infty} -\frac{1}{n}\log (1-\hat{\alpha}_{n,\rho}(e^{-nR}))
=\sup_{s\in (1,\infty)}\frac{s-1}{s}(R-\widetilde{I}_s^{\downarrow\downarrow}(A:B)_\rho ),
\end{equation}
and the same is true if $\hat{\alpha}_{n,\rho}$ in~\eqref{eq:strong-converse} is replaced by 
$\hat{\alpha}_{n,\rho}^{\mathrm{iid}}$.

Furthermore, for any $\rho_{AB}\in \mathcal{S}(AB),R\in [0,\infty)$, the right-hand side of~\eqref{eq:strong-converse} lies in $[0,\max(0,R-I(A:B)_\rho)]$, 
and it is strictly positive iff $R>I(A:B)_\rho$.
\end{thm}
\begin{proof}
See Appendix~\ref{app:strong-converse}.
\end{proof}

The proof of Theorem~\ref{thm:strong-converse} consists of two parts: a proof of achievability and a proof of optimality. 
The proof of optimality is a direct consequence of a strong converse bound in~\cite{mosonyi2014quantum}. 
The proof of achievability is technically more involved and proceeds via case distinction into two cases, depending on the type-II rate $R$. 
In the case where $R$ is less than a certain threshold $R_\infty$, 
the proof utilizes a quantum Neyman-Pearson test that compares 
$\mathcal{P}_{\omega_{A^n}^n\otimes \omega_{B^n}^n}(\rho_{AB}^{\otimes n})$ with $\omega_{A^n}^n\otimes \omega_{B^n}^n$, 
and uses the asymptotic attainability by pinching of the doubly minimized SRMI~\cite[Theorem~4~(j)]{burri2025prmisrmi1} and further properties of the doubly minimized SRMI~\cite[Theorem~4~(k), (l), (n), (o)]{burri2025prmisrmi1}. 
This part of the proof of achievability is an adapted version of an analogous proof of achievability for the singly minimized SRMI~\cite[Section~VI.C]{hayashi2016correlation}, and employs techniques for classical binary hypothesis testing from~\cite{tomamichel2018operational}. 
In the case where $R$ is greater than $R_\infty$, randomized tests are employed. 
The idea of proving strong converse theorems via case distinction into two regions of rates, along with the realization that randomized tests are necessary in the region of large rates, originates from work on strong converse theorems in classical binary hypothesis testing~\cite{nakagawa1993converse} and has been transferred to the quantum setting in~\cite{mosonyi2015two}. 

\begin{rem}[Necessity of permutation invariance of alternative hypothesis]\label{rem:permutation_srmi}
Consider again $\hat{\alpha}^{\mathrm{ind}}_{n,\rho}$, i.e., the variant of $\hat{\alpha}_{n,\rho}$ defined as in~\eqref{eq:def-alpha-ind}. 
Given the equality in~\eqref{eq:alpha-sym-an}, it is natural to ask: 
Does Theorem~\ref{thm:strong-converse} retain its validity if $\hat{\alpha}_{n,\rho}$ is replaced by $\hat{\alpha}^{\mathrm{ind}}_{n,\rho}$?
This is not the case. 
As further elaborated in Appendix~\ref{app:example_srmi}, explicit counterexamples are given by separable but not independent states $\rho_{AB}$.
\end{rem}

The proof of optimality for Theorem~\ref{thm:strong-converse} yields the following corollary.

\begin{cor}[Asymptotic minimum type-I error]\label{cor:limit-type1_strongconverse}
Let $\rho_{AB}\in \mathcal{S}(AB)$ and let $R\in (I(A:B)_\rho,\infty)$. Then  
$\lim_{n\rightarrow\infty}\hat{\alpha}_{n,\rho}(e^{-nR})=1$.
Moreover, the same is true if $\hat{\alpha}_{n,\rho}$ is replaced by 
$\hat{\alpha}_{n,\rho}^{\mathrm{iid}}$.
\end{cor}
\begin{proof}
See Appendix~\ref{app:cor_strongconverse}.
\end{proof}

The proof of Theorem~\ref{thm:strong-converse} implies the following corollary, which can be seen as another formulation of Theorem~\ref{thm:strong-converse}. 

\begin{cor}[Strong converse exponent]\label{cor:strong-converse}
Let $\rho_{AB}\in \mathcal{S}(AB)$ and let
\begin{align}\label{eq:def-R(s)-strongconverse}
R:(1,\infty)\rightarrow 
[I(A:B)_\rho,\infty),\quad  
s\mapsto 
\widetilde{I}_s^{\downarrow\downarrow}(A:B)_\rho+s(s-1)\frac{\mathrm{d}}{\mathrm{d} s}\widetilde{I}_s^{\downarrow\downarrow}(A:B)_\rho.
\end{align}
Then $R$ is continuous and monotonically increasing. 
Let $R(1)\coloneqq \lim_{s\rightarrow 1^+}R(s)=I(A:B)_\rho$, 
$R_\infty\coloneqq \lim_{s\rightarrow\infty}R(s)\in [\widetilde{I}_\infty^{\downarrow\downarrow}(A:B)_\rho,\infty]$, and 
$R(\infty)\coloneqq R_\infty$. 
Let $s_1\coloneqq \max\{s\in [1,\infty]:R(s)=I(A:B)_\rho\}$ 
and
$s_\infty \coloneqq\min\{s\in [1,\infty]:R(s)=R_\infty\}$. 

Suppose $I(A:B)_\rho\neq \widetilde{I}_\infty^{\downarrow\downarrow}(A:B)_\rho$. 
Then, $1\leq s_1<s_\infty = \infty$ and for any $s\in (s_1,s_\infty)$
\begin{equation}\label{eq:sc1}
\lim_{n\rightarrow\infty} -\frac{1}{n}\log (1-\hat{\alpha}_{n,\rho}(e^{-nR(s)}))
=\frac{s-1}{s}(R(s)-\widetilde{I}_s^{\downarrow\downarrow}(A:B)_\rho )
=(s-1)^2\frac{\mathrm{d}}{\mathrm{d}s} \widetilde{I}_s^{\downarrow\downarrow}(A:B)_\rho,
\end{equation}
and if $R_\infty<\infty$, then for any $R'\in [R_\infty,\infty)$
\begin{align}\label{eq:sc2}
\lim_{n\rightarrow\infty} -\frac{1}{n}\log (1-\hat{\alpha}_{n,\rho}(e^{-nR'}))
=R'-\widetilde{I}_\infty^{\downarrow\downarrow}(A:B)_\rho.
\end{align}
Moreover, the same is true if $\hat{\alpha}_{n,\rho}$ in~\eqref{eq:sc1} and~\eqref{eq:sc2} is replaced by 
$\hat{\alpha}_{n,\rho}^{\mathrm{iid}}$.
\end{cor}
\begin{proof}
See Appendix~\ref{app:proof-cor_strong-converse}.
\end{proof}

\begin{rem}[Strong converse exponent, assuming non-zero mutual information variance]\label{rem:strong-converse} 
The formulation of the main assertion in Corollary~\ref{cor:strong-converse} simplifies if $\rho_{AB}$ is assumed to have non-zero mutual information variance. 
More specifically, as proven in Appendix~\ref{app:rem_strong-converse}: 
\emph{For any $\rho_{AB}\in \mathcal{S}(AB)$ such that $V(A:B)_\rho\neq 0$ holds~\eqref{eq:sc1} for all $s\in (1,\infty)$, 
and the same is true if $\hat{\alpha}_{n,\rho}$ in~\eqref{eq:sc1} is replaced by 
$\hat{\alpha}_{n,\rho}^{\mathrm{iid}}$.} 
Note that there exist states whose mutual information variance is zero, but this class of states is not yet understood very well, cf. Remark~\ref{rem:direct}. 
We leave open the question of whether any $\rho_{AB}\in \mathcal{S}(AB)$ such that 
$I(A:B)_\rho\neq \widetilde{I}_{\infty}^{\downarrow\downarrow}(A:B)_\rho$ has non-zero mutual information variance. 
\end{rem}

The following result is a corollary of Theorem~\ref{thm:strong-converse} and Corollary~\ref{cor:strong-converse}. It provides a single-letter formula for the reverse $\beta$-cutoff rate. 
To define the reverse $\beta$-cutoff rate, one may view the strong converse exponent as a function of $R$ (which is monotonically increasing) and consider linear functions of $R$ that are lower bounds on this curve. 
The reverse $\beta$-cutoff rate then characterizes the best such linear lower bound. 
More explicitly, it identifies the smallest $R_0$ such that the strong converse exponent is bounded below by a line of slope $\beta$ intersecting the $x$-axis at $R=R_0$, i.e., 
the strong converse exponent behaves no worse than $\beta (R-R_0)$ for all rates $R\in (0,\infty)$. 
The parameter $\beta\in (0,1)$ controls the steepness of this bound: 
values closer to $0$ yield flatter bounds, while values closer to $1$ correspond to steeper ones. 
Thus, the reverse $\beta$-cutoff rate is the $x$-axis intercept of the best supporting line of slope $\beta$ to the strong converse exponent.
\begin{cor}[Reverse $\beta$-cutoff rate]\label{cor:reverse-cutoff}
Let $\rho_{AB}\in \mathcal{S}(AB)$. 
Let $s_{1},s_\infty$ be defined as in Corollary~\ref{cor:strong-converse}. 
For all $\alpha\in \{1,\infty\}$, let $\beta_\alpha\coloneqq 1-\frac{1}{s_\alpha}$. 
For all $\beta\in (0,1)$, let
\begin{align}\label{eq:r0-r}
R_0^{(r)}(\beta)_\rho
\coloneqq \inf \{R_0\in [0,\infty):\liminf_{n\rightarrow\infty} -\frac{1}{n}\log (1-\hat{\alpha}_{n,\rho}(e^{-nR})) \geq \beta (R-R_0) \quad \forall R\in (0,\infty) \}.
\end{align}

Suppose $I(A:B)_\rho\neq \widetilde{I}_\infty^{\downarrow\downarrow}(A:B)_\rho$. 
Then, $0\leq \beta_1< \beta_\infty = 1$ and for all $\beta\in (\beta_{1},\beta_\infty )$
\begin{align}\label{eq:r0r}
R_0^{(r)}(\beta)_\rho=\widetilde{I}_{\frac{1}{1-\beta}}^{\downarrow\downarrow}(A:B)_\rho.
\end{align}
Moreover, the same is true if $\hat{\alpha}_{n,\rho}$ in~\eqref{eq:r0-r} is replaced by 
$\hat{\alpha}_{n,\rho}^{\mathrm{iid}}$.
\end{cor}
\begin{proof}
See Appendix~\ref{ssec:reverse-cutoff}.
\end{proof}

\begin{rem}[Reverse $\beta$-cutoff rate, assuming non-zero mutual information variance]\label{rem:cutoff_r} 
The formulation of Corollary~\ref{cor:reverse-cutoff} simplifies if $\rho_{AB}$ is assumed to have non-zero mutual information variance. 
For such states, $s_1=1$ (see proof for Remark~\ref{rem:strong-converse} in Appendix~\ref{app:rem_strong-converse}), so $\beta_1=0$. 
Corollary~\ref{cor:reverse-cutoff} is therefore simplified to: 
\emph{For any $\rho_{AB}\in \mathcal{S}(AB)$ such that $V(A:B)_\rho\neq 0$ holds~\eqref{eq:r0r} for all $\beta\in (0,1)$. 
Moreover, the same is true if $\hat{\alpha}_{n,\rho}$ in~\eqref{eq:r0-r} is replaced by 
$\hat{\alpha}_{n,\rho}^{\mathrm{iid}}$.} 
\end{rem}

The significance of Corollary~\ref{cor:reverse-cutoff} (and Remark~\ref{rem:cutoff_r}) lies in the fact that it provides an operational interpretation for individual members of the family of the doubly minimized SRMI of order $\alpha\in (1,\infty)$. 
In contrast, Theorem~\ref{thm:strong-converse} only provides an operational interpretation for the entire family of the doubly minimized SRMI of order $\alpha\in (1,\infty)$ due to the presence of a supremum over $\alpha$ on the right-hand side of~\eqref{eq:strong-converse}.

For completeness, we also study the moderate deviation regime above the threshold $I(A:B)_\rho$. 
Namely, we investigate the minimum type-I error when the type-II error decays as $e^{-nR_n}$, 
where $R_n\coloneqq I(A:B)_\rho+a_n$ and $(a_n)_{n\in \mathbb{N}}$ is a moderate sequence. 
Paralleling the results of~\cite{Chubb_2017} for binary quantum state discrimination in the i.i.d. setting, we show that the minimum type-I error behaves as $1-\exp(-\frac{na_n^2}{2V(A:B)_\rho}+o(na_n^2))$, as formalized in the following theorem. 

\begin{thm}[Moderate deviation analysis above threshold]\label{thm:moderate_strongconverse}
Let $(a_n)_{n\in \mathbb{N}}$ be a strictly positive real sequence such that 
$\lim_{n\rightarrow\infty}a_n=0$ and $\lim_{n\rightarrow\infty}\sqrt{n}a_n=\infty$. 
Let $\rho_{AB}\in \mathcal{S}(AB)$ be such that $V(A:B)_\rho>0$. Then, 
\begin{align}\label{eq:mod-thm-sc}
\lim_{n\rightarrow\infty} -\frac{1}{na_n^2}\log (1-\hat{\alpha}_{n,\rho}(e^{-nR_n}))
=\frac{1}{2V(A:B)_\rho},
\end{align}
where $R_n\coloneqq I(A:B)_\rho+a_n$ for all $n\in \mathbb{N}$. 
Moreover, the same is true if $\hat{\alpha}_{n,\rho}$ in~\eqref{eq:mod-thm-sc} is replaced by $\hat{\alpha}_{n,\rho}^{\mathrm{iid}}$.
\end{thm}
\begin{proof}
See Appendix~\ref{proof:moderate_strongconverse}.
\end{proof}

The proof of Theorem~\ref{thm:moderate_strongconverse} consists of two parts: a proof of achievability and a proof of optimality. 
The proof of achievability follows by refining the strong converse exponent analysis and employs techniques from large deviation theory adapted to the moderate deviation regime, in particular exponentially tilted distributions and the G\"{a}rtner-Ellis lower bound~\cite{Ellis1975entropy}. 
The proof of optimality is a direct consequence of the moderate deviation analysis in the i.i.d. setting~\cite{Chubb_2017}.

\subsection{Stein exponent and second-order asymptotics}\label{ssec:second}

In Section~\ref{ssec:direct}, we showed that 
$\hat{\alpha}_{n,\rho}(e^{-nR})$ converges to $0$ as $n\rightarrow\infty$ 
for any $R\in (-\infty, I(A:B)_\rho)$, see Corollary~\ref{cor:limit-type1_direct}. 
Similarly, we showed in Section~\ref{ssec:strongconverse} that 
$\hat{\alpha}_{n,\rho}(e^{-nR})$ converges to $1$ as $n\rightarrow\infty$ 
for any $R\in (I(A:B)_\rho,\infty)$, see Corollary~\ref{cor:limit-type1_strongconverse}. 
The combination of these two corollaries results in the following corollary, which is a quantum Stein's lemma. 

\begin{cor}[Stein exponent]\label{cor:stein}
Let $\rho_{AB}\in \mathcal{S}(AB)$. Then
\begin{align}\label{eq:stein}
\sup\{R\in \mathbb{R}:\lim_{n\rightarrow\infty}\hat{\alpha}_{n,\rho}(e^{-nR})=0\}
=I(A:B)_\rho
=\inf\{R\in \mathbb{R}:\lim_{n\rightarrow\infty}\hat{\alpha}_{n,\rho}(e^{-nR})=1\}.
\end{align}
Moreover, the same is true if $\hat{\alpha}_{n,\rho}$ in~\eqref{eq:stein} is replaced by 
$\hat{\alpha}_{n,\rho}^{\mathrm{iid}}$.
\end{cor}

This corollary states that the threshold rate (or: Stein exponent) and the strong converse threshold rate coincide, and that the asymptotic minimum type-I error jumps sharply from $0$ to $1$ 
when the type-II rate $R$ surpasses the threshold given by $I(A:B)_\rho$.

For completeness, we also consider the second-order asymptotics of the binary quantum state discrimination problems associated with $\hat{\alpha}_{n,\rho}$ and $\hat{\alpha}_{n,\rho}^{\mathrm{iid}}$. 
The objective is to quantify the behavior of the minimum type-I error in the limit where the type-II rate $R_n$ asymptotically approaches the threshold value $I(A:B)_\rho$. 
For simplicity, we consider the concrete case where the type-II rate depends on $n$ as $R_n\coloneqq I(A:B)_\rho+\frac{r}{\sqrt{n}}$ for some fixed parameter $r\in \mathbb{R}$. 
The resulting second-order asymptotics are as follows.

\begin{thm}[Second-order asymptotics]\label{thm:second}
Let $\rho_{AB}\in \mathcal{S}(AB)$ be such that $V(A:B)_\rho\neq 0$ and let $r\in \mathbb{R}$. 
Then
\begin{align}\label{eq:second-order}
\lim_{n\rightarrow\infty} \hat{\alpha}_{n,\rho}(e^{-nR_n })
=\Phi\left(\frac{r}{\sqrt{V(A:B)_{\rho} }}\right),
\end{align}
where $R_n\coloneqq I(A:B)_\rho +\frac{r}{\sqrt{n}}$, 
and $\Phi$ is the cumulative distribution function of the standard normal distribution, i.e., 
$\mathbb{R}\rightarrow [0,1],x\mapsto \Phi(x)\coloneqq \frac{1}{\sqrt{2\pi}}\int_{-\infty}^x e^{-t^2/2} \mathrm{d}t$.
Moreover, the same is true if $\hat{\alpha}_{n,\rho}$ in~\eqref{eq:second-order} is replaced by $\hat{\alpha}_{n,\rho}^{\mathrm{iid}}$.
\end{thm}
\begin{proof}
See Appendix~\ref{app:second}.
\end{proof}
This theorem implies that the asymptotic minimum type-I error 
increases smoothly from $0$ to $1$ as $r$ increases from $-\infty$ to $+\infty$
(rather than exhibiting a discontinuous jump from $0$ to $1$ at $r=0$). 
The proof of Theorem~\ref{thm:second} consists of an achievability and an optimality part. 
The achievability part is accomplished by adapting the techniques used in~\cite[Theorem~19]{hayashi2016correlation} to prove a similar statement related to the singly minimized PRMI/SRMI, 
and by using the asymptotic optimality of universal permutation invariant state for the doubly minimized SRMI~\cite[Theorem~4~(j)]{burri2025prmisrmi1}. 
The optimality part follows immediately from previous results~\cite{li2014second,tomamichel2013hierarchy} on the second-order asymptotics of i.i.d. quantum hypothesis testing.

\begin{acknowledgments}
The author thanks Renato Renner for valuable discussions and comments. 
This work was supported by 
the Swiss National Science Foundation via grant No.\ 200021\_188541
and the National Centre of Competence in Research SwissMAP, 
and the Quantum Center at ETH Zurich.
\end{acknowledgments}

\appendix
\section{Proof for Section~\ref{ssec:hypothesis}}\label{app:classical}
\begin{proof}
Let $n\in \mathbb{N}_{>0},\mu\in [0,\infty)$. 
Consider the optimization problem that defines $\hat{\alpha}_n^{\mathrm{q}}(\mu)$. 
Let $T^n_{A^nB^n}\in \mathcal{L}(A^nB^n)$ be in the feasible set of this optimization problem, i.e., 
$0\leq T^n_{A^nB^n}\leq 1$ and $\sup_{\sigma_{A^nB^n}\in H_1^{\mathrm{q},n}}\tr[\sigma_{A^nB^n}T^n_{A^nB^n}]\leq\mu$. Let 
\begin{align}
\hat{T}^n&(x_1,y_1,\dots, x_n,y_n)\coloneqq \bra{a_{x_1},b_{y_1},\dots, a_{x_n},b_{y_n}} T_{A^nB^n}^n \ket{a_{x_1},b_{y_1},\dots, a_{x_n},b_{y_n}}
\end{align}
for all $x_1,\dots, x_n\in \mathcal{X},y_1,\dots, y_n\in \mathcal{Y}$, and let 
\begin{align}
\hat{T}_{A^nB^n}^n&\coloneqq 
\sum_{\substack{x_1,\dots, x_n\in \mathcal{X},\\ y_1,\dots, y_n\in \mathcal{Y}}}
\hat{T}^n(x_1,y_1,\dots, x_n,y_n)\proj{a_{x_1},b_{y_1},\dots, a_{x_n},b_{y_n}}.
\label{eq:tn-tanbn}
\end{align}
Since $\rho_{AB}$ is a CC state, 
$\tr[\rho_{AB}^{\otimes n}(1-T^n_{A^nB^n})]=\tr[\rho_{AB}^{\otimes n}(1-\hat{T}^n_{A^nB^n})]$. 
By~\eqref{eq:analogy1}, we have 
\begin{align}
\sup_{\sigma_{A^nB^n}\in H_1^{\mathrm{q},n}}\tr[\sigma_{A^nB^n}\hat{T}^n_{A^nB^n}]
&\leq \sup_{\sigma_{A^nB^n}\in H_1^{\mathrm{q},n}}\tr[\sigma_{A^nB^n}T^n_{A^nB^n}]\leq \mu .
\end{align}
Therefore, 
\begin{align}\label{eq:qc0}
\hat{\alpha}_n^{\mathrm{q}}(\mu)
&=\inf_{\hat{T}^n}\{\tr[\rho_{AB}^{\otimes n}(1-\hat{T}^n_{A^nB^n})]:
\sup_{\sigma_{A^nB^n}\in H_1^{\mathrm{q},n}}\tr[\sigma_{A^nB^n}\hat{T}^n_{A^nB^n}]\leq \mu \},
\end{align}
where the minimization is over all functions $\hat{T}^n:(\mathcal{X}\times\mathcal{Y})^{\times n}\rightarrow [0,1]$, and the expressions inside the curly brackets are evaluated for $\hat{T}^n_{A^nB^n}$ as in~\eqref{eq:tn-tanbn}.
Then, 
\begin{align}
\hat{\alpha}_n^{\mathrm{q}}(\mu)
&=\inf_{\hat{T}^n}\{\tr[\rho_{AB}^{\otimes n}(1-\hat{T}^n_{A^nB^n})]:
\nonumber\\
&\qquad 
\sup_{\sigma_{A^nB^n}\in H_1^{\mathrm{q},n}}
\sum_{\substack{x_1,\dots, x_n\in \mathcal{X},\\ y_1,\dots, y_n\in \mathcal{Y}}}
Q^\sigma_{X^nY^n}(x_1,y_1,\dots, x_n,y_n)\hat{T}^n(x_1,y_1,\dots, x_n,y_n)\leq \mu \}
\label{eq:q-c0}\\
&=\hat{\alpha}_n^{\mathrm{c}}(\mu).
\label{eq:q-c}
\end{align}
\eqref{eq:q-c0} follows from~\eqref{eq:qc0}. 
\eqref{eq:q-c} follows from~\eqref{eq:analogy2}.
\end{proof}

\section{Proofs for Section~\ref{ssec:direct}}
\subsection{Lemmas for Theorem~\ref{thm:direct}}\label{app:lemmas}

Let us define the following function of $\mu\in [0,\infty)$ 
for any $\rho_{AB}\in \mathcal{S}(AB),n\in \mathbb{N}_{>0}$.
\begin{align}
\hat{\alpha}_{n,\rho}'(\mu )
&= \min_{\substack{T^n_{A^nB^n}\in \mathcal{L} (A^nB^n): \\ 0\leq T^n_{A^nB^n}\leq 1}}
\{\tr[\rho_{AB}^{\otimes n}(1-T^n_{A^nB^n})]:
\max_{\substack{\sigma_{A^n} \in \mathcal{S}_{\sym}(A^{\otimes n}),\\ \tau_{B^n} \in \mathcal{S}(B^n)}}
\tr[\sigma_{A^n}\otimes \tau_{B^n} T^n_{A^nB^n}]
\leq \mu\}
\label{eq:def-alpha-a}
\end{align} 
The following lemma describes some basic properties of the functions 
$\hat{\alpha}_{n,\rho}^{\mathrm{iid}},\hat{\alpha}_{n,\rho},\hat{\alpha}_{n,\rho}^{\mathrm{ind}},\hat{\alpha}_{n,\rho}'$, 
as defined in~\eqref{eq:def-alpha-otimes}, \eqref{eq:def-alpha},~\eqref{eq:def-alpha-ind},~\eqref{eq:def-alpha-a}.
\begin{lem}[Minimum type-I errors]\label{lem:order}
Let $\rho_{AB}\in \mathcal{S}(AB),n\in \mathbb{N}_{>0}$. Then all of the following hold.
\begin{enumerate}[label=(\alph*)]
\item The function $[0,\infty)\rightarrow \mathbb{R},\mu\mapsto \hat{\alpha}_{n,\rho}(\mu)$ is monotonically decreasing, and the same is true for 
$\hat{\alpha}_{n,\rho}^{\mathrm{iid}},\hat{\alpha}_{n,\rho}^{\mathrm{ind}},$ and $\hat{\alpha}_{n,\rho}'$.
\item 
$0\leq \hat{\alpha}_{n,\rho}^{\mathrm{iid}} (\mu)
\leq \hat{\alpha}_{n,\rho}(\mu)
\leq \hat{\alpha}_{n,\rho}'(\mu)
\leq \hat{\alpha}_{n,\rho}^{\mathrm{ind}}(\mu)
\leq \max(0,1-\mu) \leq 1$
for all $\mu\in [0,\infty)$.
\item $\hat{\alpha}_{n,\rho}(\mu)=\hat{\alpha}_{n,\rho}'(\mu)$ 
for all $\mu\in [0,\infty)$.
\end{enumerate}
\end{lem}
\begin{proof}[Proof of (a)]
This assertion follows directly from the definitions of the functions.
\end{proof}
\begin{proof}[Proof of (b)]
Let $\mu\in [0,\infty)$.  
$\mathcal{S}_{\sym}(A^{\otimes n})\otimes \mathcal{S}(B^{\otimes n})\subseteq \mathcal{S}(A^{\otimes n})\otimes \mathcal{S}(B^{\otimes n})$ 
implies that 
$\hat{\alpha}_{n,\rho}'(\mu)\leq \hat{\alpha}_{n,\rho}^{\mathrm{ind}}(\mu)$.
$\mathcal{S}_{\sym}(A^{\otimes n})\otimes \mathcal{S}_{\sym}(B^{\otimes n})\subseteq \mathcal{S}_{\sym}(A^{\otimes n})\otimes \mathcal{S}(B^n)$
implies that 
$\hat{\alpha}_{n,\rho}(\mu)\leq \hat{\alpha}_{n,\rho}'(\mu)$.
$\mathcal{S}(A)^{\otimes n}\otimes \mathcal{S}(B)^{\otimes n}\subseteq \mathcal{S}_{\sym}(A^{\otimes n})\otimes \mathcal{S}_{\sym}(B^{\otimes n})$
implies that 
$\hat{\alpha}_{n,\rho}^{\mathrm{iid}}(\mu)\leq \hat{\alpha}_{n,\rho}(\mu)$.

Let $T^n_{A^nB^n}\in \mathcal{L}(A^nB^n)$ be in the feasible set of the optimization problem that defines $\hat{\alpha}_{n,\rho}^{\mathrm{iid}}$. 
Then, $0\leq T^n_{A^nB^n}\leq 1$. Hence, $1-T^n_{A^nB^n}\geq 0$, which implies that 
$\hat{\alpha}_{n,\rho}^{\mathrm{iid}}(\mu)\geq 0$. 

If $\mu\in [0,1]$, then choosing the test $T^n_{A^nB^n}\coloneqq \mu 1$ implies that 
$\hat{\alpha}_{n,\rho}^{\mathrm{ind}}(\mu)\leq 1-\mu$.
If $\mu\in (1,\infty)$, then choosing the test $T^n_{A^nB^n}\coloneqq 1$ implies that 
$\hat{\alpha}_{n,\rho}^{\mathrm{ind}}(\mu)\leq 0$. 
Hence, $\hat{\alpha}_{n,\rho}^{\mathrm{ind}}(\mu)\leq \max(0,1-\mu)$.
\end{proof}
\begin{proof}[Proof of (c)]
Let $\mu\in [0,\infty)$. 
By~(b), it suffices to prove that $\hat{\alpha}_{n,\rho}'(\mu)\leq \hat{\alpha}_{n,\rho}(\mu)$. 

Let $T^n_{A^nB^n}\in \mathcal{L}(A^nB^n)$ be in the feasible set of the optimization problem that defines $\hat{\alpha}_{n,\rho}(\mu)$. 
Let $\hat{T}^n_{A^nB^n}\coloneqq \frac{1}{|S_n|}\sum_{\pi \in S_n}U(\pi)_{A^n}\otimes U(\pi)_{B^n} T^n_{A^nB^n}U(\pi)_{A^n}^\dagger\otimes U(\pi)_{B^n}^\dagger$. 
Then, for any $\sigma_{A^n} \in \mathcal{S}_{\sym}(A^{\otimes n}),\tau_{B^n} \in \mathcal{S}_{\sym}(B^{\otimes n})$
\begin{align}
\tr[\sigma_{A^n}\otimes \tau_{B^n}T^n_{A^nB^n}]
&=\frac{1}{|S_n|}\sum_{\pi \in S_n}\tr[(U(\pi)_{A^n}^\dagger \sigma_{A^n}U(\pi)_{A^n})\otimes ( U(\pi)_{B^n}^\dagger\tau_{B^n}U(\pi)_{B^n})T^n_{A^nB^n}]
\\
&=\frac{1}{|S_n|}\sum_{\pi \in S_n}\tr[\sigma_{A^n}\otimes\tau_{B^n} U(\pi)_{A^n}\otimes U(\pi)_{B^n} T^n_{A^nB^n}U(\pi)_{A^n}^\dagger\otimes U(\pi)_{B^n}^\dagger]
\\
&=\tr[\sigma_{A^n}\otimes \tau_{B^n}\hat{T}^n_{A^nB^n}].
\end{align}
Hence, 
\begin{align}
\max_{\substack{\sigma_{A^n} \in \mathcal{S}_{\sym}(A^{\otimes n}),\\ \tau_{B^n} \in \mathcal{S}_{\sym}(B^{\otimes n})}} 
\tr[\sigma_{A^n}\otimes \tau_{B^n}T^n_{A^nB^n}] 
=\max_{\substack{\sigma_{A^n} \in \mathcal{S}_{\sym}(A^{\otimes n}),\\ \tau_{B^n} \in \mathcal{S}_{\sym}(B^{\otimes n})}} 
\tr[\sigma_{A^n}\otimes \tau_{B^n}\hat{T}^n_{A^nB^n}] .
\end{align}
Since $\rho_{AB}^{\otimes n}\in \mathcal{S}_{\sym}((AB)^{\otimes n})$, 
$\tr[\rho_{AB}^{\otimes n}(1-T^n_{A^nB^n})]=\tr[\rho_{AB}^{\otimes n}(1-\hat{T}^n_{A^nB^n})]$.
Since $\hat{T}^n_{A^nB^n}$ is permutation invariant, 
it follows that
\begin{align}\label{eq:a-t-sym}
\hat{\alpha}_{n,\rho}(\mu)
&=\min_{\substack{T^n_{A^nB^n}\in \mathcal{L}_{\sym}((AB)^{\otimes n}) : \\ 0\leq T^n_{A^nB^n}\leq 1}}
\{\tr[\rho_{AB}^{\otimes n}(1-T^n_{A^nB^n})]: 
\max_{\substack{\sigma_{A^n} \in \mathcal{S}_{\sym}(A^{\otimes n}),\\ \tau_{B^n} \in \mathcal{S}_{\sym}(B^{\otimes n})}} 
\tr[\sigma_{A^n}\otimes \tau_{B^n}T^n_{A^nB^n}] \leq \mu \}.
\end{align}
Let $\widetilde{T}^n_{A^nB^n}\in \mathcal{L}_{\sym}((AB)^{\otimes n})$ be positive semidefinite.
Then, for all 
$\sigma_{A^n} \in \mathcal{S}_{\sym}(A^{\otimes n}), \tau_{B^n} \in \mathcal{S}(B^n)$
\begin{align}
\tr[\sigma_{A^n}\otimes \tau_{B^n} \widetilde{T}^n_{A^nB^n}]
&=\frac{1}{|S_n|}\sum_{\pi\in S_n}
\tr[\sigma_{A^n}\otimes \tau_{B^n} (U(\pi)_{A^n}\otimes U(\pi)_{B^n} \widetilde{T}^n_{A^nB^n}  U(\pi)_{A^n}^\dagger\otimes U(\pi)_{B^n}^\dagger)]
\\
&=\frac{1}{|S_n|}\sum_{\pi\in S_n}
\tr[\underbrace{(U(\pi)_{A^n}^\dagger\sigma_{A^n} U(\pi)_{A^n})}_{=\sigma_{A^n}} \otimes (U(\pi)_{B^n}^\dagger\tau_{B^n} U(\pi)_{B^n}) \widetilde{T}^n_{A^nB^n}]
\\
&=
\tr[\sigma_{A^n}\otimes \underbrace{\frac{1}{|S_n|}\sum_{\pi\in S_n}(U(\pi)_{B^n}^\dagger\tau_{B^n} U(\pi)_{B^n})}_{\in \mathcal{S}_{\sym}(B^{\otimes n})} \widetilde{T}^n_{A^nB^n}].
\end{align} 
Hence,
\begin{equation}\label{eq:max-symt}
\max_{\substack{\sigma_{A^n} \in \mathcal{S}_{\sym}(A^{\otimes n}),\\ \tau_{B^n} \in \mathcal{S}(B^n)}}
\tr[\sigma_{A^n}\otimes \tau_{B^n} \widetilde{T}^n_{A^nB^n}]
=\max_{\substack{\sigma_{A^n} \in \mathcal{S}_{\sym}(A^{\otimes n}),\\ \tau_{B^n} \in \mathcal{S}_{\sym}(B^{\otimes n})}}
\tr[\sigma_{A^n}\otimes \tau_{B^n} \widetilde{T}^n_{A^nB^n}].
\end{equation}
We conclude that 
\begin{align}
\hat{\alpha}_{n,\rho}'(\mu)
&\leq \min_{\substack{T^n_{A^nB^n}\in \mathcal{L}_{\sym}((AB)^{\otimes n}): \\ 0\leq T^n_{A^nB^n}\leq 1}}
\{\tr[\rho_{AB}^{\otimes n}(1-T^n_{A^nB^n})]:
\max_{\substack{\sigma_{A^n} \in \mathcal{S}_{\sym}(A^{\otimes n}),\\ \tau_{B^n} \in \mathcal{S}(B^n)}} 
\tr[\sigma_{A^n}\otimes \tau_{B^n}T^n_{A^nB^n}] \leq \mu \}
\\
&= \min_{\substack{T^n_{A^nB^n}\in \mathcal{L}_{\sym}((AB)^{\otimes n}): \\ 0\leq T^n_{A^nB^n}\leq 1}}
\{\tr[\rho_{AB}^{\otimes n}(1-T^n_{A^nB^n})]:
\max_{\substack{\sigma_{A^n} \in \mathcal{S}_{\sym}(A^{\otimes n}),\\ \tau_{B^n} \in \mathcal{S}_{\sym}(B^{\otimes n})}} 
\tr[\sigma_{A^n}\otimes \tau_{B^n}T^n_{A^nB^n}] \leq \mu \}
\label{eq:alpha-proof2}\\
&= \hat{\alpha}_{n,\rho}(\mu).
\label{eq:alpha-proof3}
\end{align}
\eqref{eq:alpha-proof2} follows from~\eqref{eq:max-symt}.
\eqref{eq:alpha-proof3} follows from~\eqref{eq:a-t-sym}.
\end{proof}

The following lemma is a variant of~\cite[Lemma~20]{tomamichel2018operational} that does not require $\phi$ to be differentiable.
\begin{lem}[A property of convex functions]\label{lem:convex}
Let $a,b\in \mathbb{R}$ be such that $0\leq a<b$ and let $\phi:(a,b)\rightarrow\mathbb{R}$ be a convex function. 
Let $\psi(s)\coloneqq s\frac{\partial}{\partial s^+}\phi(s)-\phi(s)$ for all $s\in (a,b)$. 
Then, $\psi$ is monotonically increasing.
\end{lem}
\begin{proof}
Let $s_0,s_1\in \mathbb{R}$ be such that $a<s_0<s_1<b$. Then, 
\begin{align}
\psi(s_1)-\psi(s_0)
&=s_1\frac{\partial}{\partial s^+}\phi(s_1)-s_0\frac{\partial}{\partial s^+}\phi(s_0)-(\phi(s_1)-\phi(s_0)) 
\\
&\geq s_1\frac{\partial}{\partial s^+}\phi(s_1)-s_0\frac{\partial}{\partial s^+}\phi(s_0)-(s_1-s_0) \frac{\partial}{\partial s^+}\phi(s_1)
\\
&=s_0\left(\frac{\partial}{\partial s^+}\phi(s_1)-\frac{\partial}{\partial s^+}\phi(s_0) \right)\geq 0,
\end{align}
where the inequalities follow from the convexity of $\phi$.
\end{proof}

\subsection{Proof of Theorem~\ref{thm:direct}}\label{app:hoeffding}

First, we prove the bounds on $R_{1/2}$. 
It follows from \cite[Theorem~3~(q)]{burri2025prmisrmi1} that $I_s^{\downarrow\downarrow}(A:B)_\rho$ is right-differentiable on $s\in (0,1)$. 
By the monotonicity of the doubly minimized PRMI in the R\'enyi order, see \cite[Theorem~3~(n)]{burri2025prmisrmi1},
\begin{align}\label{eq:r12-0}
\frac{\partial}{\partial s^+}I_s^{\downarrow\downarrow}(A:B)_\rho\geq 0\qquad \forall s\in (0,1).
\end{align}
For any $s\in (0,1)$ and any fixed
$(\sigma_A,\tau_B)\in \argmin_{(\sigma_A',\tau_B')\in \mathcal{S}(A)\times \mathcal{S}(B)}D_s (\rho_{AB}\| \sigma_A'\otimes \tau_B')$, 
the right derivative at $s$ is upper bounded as 
$\frac{\partial}{\partial s^+}I_s^{\downarrow\downarrow}(A:B)_\rho 
\leq \frac{\partial}{\partial s^+}D_s (\rho_{AB}\| \sigma_A\otimes \tau_B)$, 
which implies that  
$0\leq \liminf_{s\rightarrow 0^+}\frac{\partial}{\partial s^+}I_s^{\downarrow\downarrow}(A:B)_\rho
\leq \limsup_{s\rightarrow 0^+}\frac{\partial}{\partial s^+}I_s^{\downarrow\downarrow}(A:B)_\rho <\infty$. Hence, 
\begin{align}
I_0^{\downarrow\downarrow}(A:B)_\rho
&=\lim_{s\rightarrow 0^+} (I_{s}^{\downarrow\downarrow}(A:B)_\rho-s(1-s)\frac{\partial}{\partial s^+} I_{s}^{\downarrow\downarrow}(A:B)_\rho) 
\label{eq:r12-1}\\
&\leq (I_{s}^{\downarrow\downarrow}(A:B)_\rho-s(1-s)\frac{\partial}{\partial s^+} I_{s}^{\downarrow\downarrow}(A:B)_\rho )\big|_{s=1/2}
=R_{1/2}
\label{eq:r12-2}\\
&\leq I_{1/2}^{\downarrow\downarrow}(A:B)_\rho.
\label{eq:r12-3}
\end{align}
\eqref{eq:r12-2} follows from Lemma~\ref{lem:convex} due to the convexity of $(s-1)I_s^{\downarrow\downarrow}(A:B)_\rho$ on $s\in (0,1)$, see \cite[Theorem~3~(q)]{burri2025prmisrmi1}. 
\eqref{eq:r12-3} follows from~\eqref{eq:r12-0}.

Next, we prove the bounds on the right-hand side of~\eqref{eq:direct}. 
For any $R\in (0,\infty)$
\begin{subequations}\label{eq:direct-rhs}
\begin{align}
0=\lim_{s\rightarrow 1^-}\frac{1-s}{s}(I_s^{\downarrow\downarrow}(A:B)_\rho - R)
&\leq \sup_{s\in (\frac{1}{2},1)}\frac{1-s}{s}(I_s^{\downarrow\downarrow}(A:B)_\rho - R)
\\
&\leq \sup_{s\in (\frac{1}{2},1)}\frac{1-s}{s}(I_1^{\downarrow\downarrow}(A:B)_\rho - R)
=\max(0,I(A:B)_\rho -R),
\end{align}
\end{subequations}
where we have used the monotonicity and continuity of the doubly minimized PRMI in the R\'enyi order 
and $I_1^{\downarrow\downarrow}(A:B)_\rho=I(A:B)_\rho$, 
see \cite[Theorem~3~(m), (n), (o)]{burri2025prmisrmi1}. 
On the one hand, the bounds in~\eqref{eq:direct-rhs} imply that for any $R\in [I(A:B)_\rho,\infty)$, the right-hand side of~\eqref{eq:direct} vanishes. 
On the other hand, if $R\in (0,I(A:B)_\rho)$, then 
the right-hand side of~\eqref{eq:direct} is strictly positive 
due to \cite[Theorem~3~(m), (n), (o)]{burri2025prmisrmi1}.

We will now prove the equality in~\eqref{eq:direct}. 
The proof of~\eqref{eq:direct} is divided into two parts: 
a proof of achievability for $\hat{\alpha}_{n,\rho}$ and a proof of optimality for $\hat{\alpha}_{n,\rho}^{\mathrm{iid}}$. 
The assertion in Theorem~\ref{thm:direct} follows from these two parts because 
$\hat{\alpha}_{n,\rho}^{\mathrm{iid}} (\mu)\leq \hat{\alpha}_{n,\rho}(\mu)$ 
for all $\mu\in [0,\infty)$, see Lemma~\ref{lem:order}~(b).

\subsubsection{Proof of achievability}\label{sec:achievability}
Below, we will prove that for any $R\in (0,\infty)$
\begin{align}
\liminf_{n\rightarrow\infty}-\frac{1}{n}\log \hat{\alpha}_{n,\rho}(e^{-nR})
&\geq \sup_{s\in (0,1)}\frac{1-s}{s}(I_s^{\downarrow\downarrow}(A:B)_\rho - R)
\label{eq:ach-01}\\
&\geq \sup_{s\in (\frac{1}{2},1)}\frac{1-s}{s}(I_s^{\downarrow\downarrow}(A:B)_\rho - R).
\label{eq:ach-12}
\end{align}
The inequality in~\eqref{eq:ach-12} is trivially true, so it suffices to prove~\eqref{eq:ach-01}.
\begin{proof}
Let $R\in (0,\infty)$ and $s\in (0,1)$ be arbitrary but fixed. For all $n\in \mathbb{N}_{>0}$, we define
\begin{align}
\lambda_n\coloneqq \frac{1}{s} 
\left(\log g_{n,d_A}+\log g_{n,d_B}+nR-(1-s)D_s(\rho_{AB}^{\otimes n}\| \omega_{A^n}^n \otimes \omega_{B^n}^n )\right)
\label{eq:ach-lambda}
\end{align}
and the test
$T^n_{A^nB^n}\coloneqq \{\rho_{AB}^{\otimes n}\geq e^{\lambda_n}\omega_{A^n}^n\otimes \omega_{B^n}^n\}$.
For this test holds
\begin{subequations}\label{eq:ach-two}
\begin{align}
\max_{\substack{\sigma_{A^n} \in \mathcal{S}_{\sym}(A^{\otimes n}),\\ \tau_{B^n} \in \mathcal{S}_{\sym}(B^{\otimes n})}}
\tr[\sigma_{A^n}\otimes \tau_{B^n} T^n_{A^nB^n}]
&\leq g_{n,d_A}g_{n,d_B} \tr[\omega_{A^n}^n\otimes\omega_{B^n}^n \{e^{-\lambda_n} \rho_{AB}^{\otimes n}\geq \omega_{A^n}^n\otimes\omega_{B^n}^n \} ] 
\label{eq:ach-b0}\\
&\leq g_{n,d_A}g_{n,d_B}\tr[(e^{-\lambda_n} \rho_{AB}^{\otimes n})^s ( \omega_{A^n}^n\otimes\omega_{B^n}^n )^{1-s}]
\label{eq:ach-b1}\\
&=g_{n,d_A}g_{n,d_B}e^{-s\lambda_n} 
\exp\left(-(1-s)D_s(\rho_{AB}^{\otimes n}\| \omega_{A^n}^n\otimes\omega_{B^n}^n )\right)
\\
&=e^{-nR}.
\label{eq:ach-b3}
\end{align}
\end{subequations}
\eqref{eq:ach-b0} follows from \cite[Remark~1~(b)]{burri2025prmisrmi1}.  
\eqref{eq:ach-b1} follows from \cite[Eq.~(2.2)]{hayashi2016correlation}. 
\eqref{eq:ach-b3} follows from~\eqref{eq:ach-lambda}.
Furthermore,
\begin{subequations}\label{eq:ach-one}
\begin{align}
\tr[\rho_{AB}^{\otimes n}(1-T^n_{A^nB^n})]
&=\tr[\rho_{AB}^{\otimes n} \{\rho_{AB}^{\otimes n}<e^{\lambda_n}\omega_{A^n}^n\otimes \omega_{B^n}^n\}]\\
&\leq \tr[(\rho_{AB}^{\otimes n})^s (e^{\lambda_n}\omega_{A^n}^n\otimes \omega_{B^n}^n)^{1-s}]
\label{eq:alpha-lambda}\\
&=e^{\lambda_n (1-s)} 
\exp\left(-(1-s)D_s(\rho_{AB}^{\otimes n} \| \omega_{A^n}^n\otimes \omega_{B^n}^n)\right)\\
&=\exp\left(\frac{1-s}{s}\left(\log g_{n,d_A}+\log g_{n,d_B}
-(D_s(\rho_{AB}^{\otimes n}\| \omega_{A^n}^n \otimes \omega_{B^n}^n) - nR)\right)\right).
\label{eq:alpha-qn}
\end{align}
\end{subequations}
\eqref{eq:alpha-lambda} follows from~\cite[Eq.~(2.2)]{hayashi2016correlation}.
\eqref{eq:alpha-qn} follows from~\eqref{eq:ach-lambda}. 
We conclude that 
\begin{subequations}\label{eq:ach-a}
\begin{align}
\liminf_{n\rightarrow\infty} -\frac{1}{n}\log \hat{\alpha}_{n,\rho}(e^{-nR})
&\geq \liminf_{n\rightarrow\infty} -\frac{1}{n}\log \tr[\rho_{AB}^{\otimes n}(1-T^n_{A^nB^n})]
\label{eq:ach-a0}\\
&\geq \frac{1-s}{s}\left(\liminf_{n\rightarrow\infty} \frac{1}{n} D_s(\rho_{AB}^{\otimes n} \| \omega_{A^n}^n\otimes \omega_{B^n}^n) -R\right)
\label{eq:ach-a1}\\
&= \frac{1-s}{s}(I_s^{\downarrow\downarrow}(A:B)_\rho -R).
\label{eq:ach-a2}
\end{align}
\end{subequations}
\eqref{eq:ach-a0} follows from~\eqref{eq:ach-two}.
\eqref{eq:ach-a1} follows from~\eqref{eq:ach-one} and \cite[Remark~1~(b)]{burri2025prmisrmi1}.
\eqref{eq:ach-a2} follows from \cite[Theorem~3~(l)]{burri2025prmisrmi1}. 
Since $s\in (0,1)$ can be chosen arbitrarily, the assertion in~\eqref{eq:ach-01} follows from~\eqref{eq:ach-a}.
\end{proof}

\subsubsection{Proof of optimality}
Below, we will prove that for any $R\in (R_{1/2},\infty)$
\begin{align}\label{eq:opt3}
\limsup_{n\rightarrow\infty}-\frac{1}{n}\log \hat{\alpha}_{n,\rho}^{\mathrm{iid}} (e^{-nR})
\leq \sup_{s\in (\frac{1}{2},1)}\frac{1-s}{s}(I_s^{\downarrow\downarrow}(A:B)_\rho - R).
\end{align}

\begin{proof}
Let $R\in (R_{1/2},\infty)$ be arbitrary but fixed.
Since $R>R_{1/2}\geq 0$, the converse quantum Hoeffding bound~\cite{nagaoka2006converse,audenaert2008asymptotic} implies that for any $(\sigma_A,\tau_B)\in \mathcal{S}(A)\times \mathcal{S}(B)$
\begin{align}\label{eq:converse-hoeffding}
\limsup_{n\rightarrow\infty} -\frac{1}{n}\log \hat{\alpha}_{n,\rho}^{\mathrm{iid}} (e^{-nR})
&\leq \sup_{s\in (0,1)} \frac{1-s}{s}(D_s(\rho_{AB}\| \sigma_A\otimes \tau_B)-R) .
\end{align}
Case 1: $R\in [I(A:B)_\rho,\infty)$. Then,
\begin{align}
\limsup_{n\rightarrow\infty} -\frac{1}{n}\log \hat{\alpha}_{n,\rho}^{\mathrm{iid}} (e^{-nR})
&\leq \sup_{s\in (0,1)} \frac{1-s}{s}(I_s^{\uparrow\uparrow}(A:B)_\rho-R) 
\label{eq:optimal0}
\\
&=0
\label{eq:optimal1}
\\
&=\sup_{s\in (\frac{1}{2},1)}\frac{1-s}{s}(I_s^{\downarrow\downarrow}(A:B)_\rho - R).
\label{eq:optimal2}
\end{align}
\eqref{eq:optimal0} follows from the evaluation of~\eqref{eq:converse-hoeffding} for $(\sigma_A,\tau_B)=(\rho_A,\rho_B)$. 
\eqref{eq:optimal1} follows from 
$I_1^{\uparrow\uparrow}(A:B)_\rho=I(A:B)_\rho \leq R$ 
and the monotonicity in the R\'enyi order of the non-minimized PRMI. 
\eqref{eq:optimal2} follows from~\eqref{eq:direct-rhs}. 

Case 2: $R\in (R_{1/2},I(A:B)_\rho)$. 
Let us define the following functions of $s\in (\frac{1}{2},1)$.
\begin{align}
\phi(s)&\coloneqq (s-1)I_s^{\downarrow\downarrow}(A:B)_\rho \\
\psi(s)&\coloneqq s\phi'(s)-\phi(s)
=I_s^{\downarrow\downarrow}(A:B)_\rho-s(1-s)\frac{\mathrm{d}}{\mathrm{d} s}I_s^{\downarrow\downarrow}(A:B)_\rho \\
g(s)&\coloneqq \frac{1}{s}((s-1)R-\phi(s))=\frac{1-s}{s}(I_s^{\downarrow\downarrow}(A:B)_\rho -R)
\end{align}
Note that $\lim_{s\rightarrow 1/2^+}\psi (s)=R_{1/2}$ and 
$\lim_{s\rightarrow 1^-}\psi (s)=I_{1}^{\downarrow\downarrow}(A:B)_\rho=I(A:B)_\rho $ due to 
\cite[Theorem~3~(m)]{burri2025prmisrmi1}. 
$\phi$ is convex and continuously differentiable due to \cite[Theorem~3~(p), (q)]{burri2025prmisrmi1}. 
Since $\phi$ is continuously differentiable, $\psi$ is continuous. 
The convexity of $\phi$ implies that $\psi$ is monotonically increasing, see~\cite[Lemma~20]{tomamichel2018operational} or Lemma~\ref{lem:convex}. 
Since $\phi$ is continuously differentiable, also $g$ is continuously differentiable, so 
$g'(s)=\frac{1}{s^2}(R-\psi(s))$ is continuous. 
Since $\psi$ is monotonically increasing, $g'$ is monotonically decreasing. 
Furthermore, $\lim_{s\rightarrow 1/2^+}g'(s)=4(R-R_{1/2})>0$ and 
$\lim_{s\rightarrow 1^-}g'(s)=R-I(A:B)_\rho<0$. 
Thus, $g'(s)=0$ iff $g$ is maximal at $s$.

Let $\hat{s}\in (\frac{1}{2},1)$ be such that $g'(\hat{s})=0$ and 
\begin{align}\label{eq:optimal-0}
\sup_{s\in (\frac{1}{2},1)} \frac{1-s}{s}(I_s^{\downarrow\downarrow}(A:B)_\rho-R) 
=\frac{1-\hat{s}}{\hat{s}}(I_{\hat{s}}^{\downarrow\downarrow}(A:B)_\rho-R).
\end{align}
Such an $\hat{s}$ exists because $R\in (R_{1/2},I(A:B)_\rho)$, 
see also~\cite[Lemma~21]{tomamichel2018operational}.

Let $(\sigma_A^{(\hat{s})},\tau_B^{(\hat{s})})\in \argmin_{(\sigma_A,\tau_B)\in\mathcal{S}(A)\times\mathcal{S}(B)} D_{\hat{s}}(\rho_{AB}\| \sigma_A\otimes\tau_B)$ be the unique minimizer, see \cite[Theorem~3~(j)]{burri2025prmisrmi1}. 
Let us define the following functions of $s\in (0,1)$. 
\begin{align}
\bar{\phi}(s)&\coloneqq (s-1)D_s(\rho_{AB}\| \sigma_A^{(\hat{s})}\otimes \tau_B^{(\hat{s})} ) \\
\bar{\psi}(s)&\coloneqq s\bar{\phi}'(s)-\bar{\phi}(s)
=D_s(\rho_{AB}\| \sigma_A^{(\hat{s})}\otimes \tau_B^{(\hat{s})} )-s(1-s)\frac{\mathrm{d}}{\mathrm{d} s}D_s(\rho_{AB}\| \sigma_A^{(\hat{s})}\otimes \tau_B^{(\hat{s})} ) \\
\bar{g}(s)&\coloneqq \frac{1}{s}((s-1)R-\bar{\phi}(s))
=\frac{1-s}{s}(D_s(\rho_{AB}\| \sigma_A^{(\hat{s})}\otimes \tau_B^{(\hat{s})} ) -R)
\end{align}
$\bar{\phi}$ is convex and continuously differentiable due to the corresponding properties of the Petz divergence (see \cite[Remark~2]{burri2025prmisrmi1}). 
The convexity of $\bar{\phi}$ implies that $\bar{\psi}$ is monotonically increasing, see~\cite[Lemma~20]{tomamichel2018operational} or Lemma~\ref{lem:convex}. 
Since $\bar{\phi}$ is continuously differentiable, also $\bar{g}$ is continuously differentiable.
Hence, $\bar{g}'(s)=\frac{1}{s^2}(R-\bar{\psi}(s))$ is continuous and monotonically decreasing. 
Therefore, if $\bar{g}'(s)=0$ for some $s\in (0,1)$, then $\bar{g}$ is maximal at $s$.

By the definition of $(\sigma_A^{(\hat{s})},\tau_B^{(\hat{s})})$, it is clear that 
$D_{\hat{s}}(\rho_{AB}\| \sigma_A^{(\hat{s})}\otimes \tau_B^{(\hat{s})} )=I_{\hat{s}}^{\downarrow\downarrow}(A:B)_\rho$. 
According to \cite[Theorem~3~(p)]{burri2025prmisrmi1}, 
$\frac{\mathrm{d}}{\mathrm{d} s} D_{s}(\rho_{AB}\| \sigma_A^{(\hat{s})}\otimes \tau_B^{(\hat{s})} )|_{s=\hat{s}}=\frac{\mathrm{d}}{\mathrm{d} s}I_{s}^{\downarrow\downarrow}(A:B)_\rho|_{s=\hat{s}}$. 
Importantly, this implies that $\phi$ and $\bar{\phi}$ are the same up to first order at $\hat{s}$, i.e.,
\begin{align}
\bar{\phi}(\hat{s})&=(\hat{s}-1)I_{\hat{s}}^{\downarrow\downarrow}(A:B)_\rho = \phi(\hat{s})
\qquad\text{and}\qquad 
\bar{\phi}'(\hat{s})=\phi'(\hat{s}).
\label{eq:opt-bar1}
\end{align}
\eqref{eq:opt-bar1} implies that $\bar{g}'(\hat{s})=g'(\hat{s})=0$, so $\bar{g}$ achieves its maximum at $\hat{s}$. 
Therefore,
\begin{subequations}\label{eq:optimal-1}
\begin{align}
\sup_{s\in (0,1)} \frac{1-s}{s}(D_s(\rho_{AB}\| \sigma_A^{(\hat{s})}\otimes \tau_B^{(\hat{s})} )-R) 
&=\frac{1-\hat{s}}{\hat{s}}(D_{\hat{s}}(\rho_{AB}\| \sigma_A^{(\hat{s})}\otimes \tau_B^{(\hat{s})} )-R)
\\
&=\frac{1-\hat{s}}{\hat{s}}(I_{\hat{s}}^{\downarrow\downarrow}(A:B)_\rho -R).
\end{align}
\end{subequations}
By evaluating~\eqref{eq:converse-hoeffding} for $(\sigma_A,\tau_B)=(\sigma_A^{(\hat{s})},\tau_B^{(\hat{s})})$, it follows that
\begin{align}\label{eq:optimal-2}
\limsup_{n\rightarrow\infty} -\frac{1}{n}\log \hat{\alpha}_{n,\rho}^{\mathrm{iid}} (e^{-nR})
&\leq \sup_{s\in (0,1)} \frac{1-s}{s}(D_s(\rho_{AB}\| \sigma_A^{(\hat{s})}\otimes\tau_B^{(\hat{s})} )-R).
\end{align}
The assertion then follows from the combination of~\eqref{eq:optimal-0},~\eqref{eq:optimal-1},  and~\eqref{eq:optimal-2}.
\end{proof}

\subsection{Example for Remark~\ref{rem:permutation_prmi}}\label{app:example_prmi}
Before we begin constructing a counterexample, let us note that for any 
$\rho_{AB}\in \mathcal{S}(AB),n\in \mathbb{N}_{>0},\mu\in [0,\infty)$, 
\begin{align}\label{eq:a-all2}
\hat{\alpha}_{n,\rho}^{\mathrm{ind}}(\mu)
=\min_{\substack{T^n_{A^nB^n}\in \mathcal{L} (A^nB^n): \\ 0\leq T^n_{A^nB^n}\leq 1}}
&\{\tr[\rho_{AB}^{\otimes n}(1-T^n_{A^nB^n})]:
\nonumber \\
&\max_{m\in \mathbb{N}_{>0}}
\max_{\substack{\sigma_{A^n}^{(i)} \in \mathcal{S}(A^n),\tau_{B^n}^{(i)} \in \mathcal{S}(B^n), \\
(p_i)_{i\in [m]}\in [0,1]^{\times m}: \sum\limits_{i\in [m]}p_i=1 }}
\sum_{i\in [m]}p_i
\tr[\sigma_{A^n}^{(i)}\otimes \tau_{B^n}^{(i)} T^n_{A^nB^n}]\leq \mu\}.
\end{align}

We will now construct a counterexample for the consideration in Remark~\ref{rem:permutation_prmi}.
Suppose $d_A\geq 2,d_B\geq 2$, 
and let $\{\ket{i}_A\}_{i=0}^1,\{\ket{i}_B\}_{i=0}^1$ be orthonormal vectors in $A,B$.
Let $p\in (\frac{1}{2},1)$ and let 
$\rho_{AB}\coloneqq p \proj{0,0}_{AB}+(1-p)\proj{1,1}_{AB}$. 
By \cite[Theorem~3~(v)]{burri2025prmisrmi1},
\begin{align}
R_{1/2}
&\coloneqq I_{1/2}^{\downarrow\downarrow}(A:B)_\rho-\frac{1}{4}\frac{\partial}{\partial s^+}I_{s}^{\downarrow\downarrow}(A:B)_\rho \big|_{s=\frac{1}{2}}
=-\log p -\frac{1}{4}(-4\log \max(p,1-p))=0.
\end{align}

Consider now the left-hand side of~\eqref{eq:direct} with $\hat{\alpha}_{n,\rho}$ replaced by 
$\hat{\alpha}^{\mathrm{ind}}_{n,\rho}$.
Since $\rho_{AB}$ is separable between $A$ and $B$, also $\rho_{AB}^{\otimes n}$ is separable between $A^n$ and $B^n$ for any $n\in \mathbb{N}_{>0}$. 
Hence,~\eqref{eq:a-all2} implies that $\hat{\alpha}^{\mathrm{ind}}_{n,\rho}(\mu)\geq 1-\mu$ for all $\mu\in [0,1]$. 
On the other hand, due to Lemma~\ref{lem:order}~(b), we have 
$\hat{\alpha}^{\mathrm{ind}}_{n,\rho}(\mu)\leq 1-\mu$ for all $\mu\in [0,1]$. 
Therefore, $\hat{\alpha}^{\mathrm{ind}}_{n,\rho}(\mu)=1-\mu$ for all $\mu\in [0,1]$. 
This implies that $\lim_{n\rightarrow\infty}\hat{\alpha}_{n,\rho}^{\mathrm{ind}}(e^{-nR})=1$ 
for any $R\in (0,\infty)$. 
Hence, for any $R\in (0,\infty)$
\begin{equation}\label{eq:ex-permutation1}
\lim_{n\rightarrow\infty}-\frac{1}{n}\log \hat{\alpha}_{n,\rho}^{\mathrm{ind}}(e^{-nR})=0.
\end{equation}
Consider now the right-hand side of~\eqref{eq:direct}. 
For any 
$R\in (0,H(A)_\rho)$
\begin{align}\label{eq:ex-permutation2}
\sup_{s\in (\frac{1}{2},1)}\frac{1-s}{s}(I_s^{\downarrow\downarrow}(A:B)_\rho - R)
&=\sup_{s\in (\frac{1}{2},1)}\frac{1-s}{s}(H_{\frac{s}{2s-1}}(A)_\rho - R)
>0.
\end{align} 
The equality in~\eqref{eq:ex-permutation2} follows from \cite[Theorem~3~(v)]{burri2025prmisrmi1}. 
The strict inequality in~\eqref{eq:ex-permutation2} follows from the continuity of the R\'enyi entropy in the R\'enyi order and  $R<H(A)_\rho=H_1(A)_\rho$. 
A comparison of~\eqref{eq:ex-permutation1} and~\eqref{eq:ex-permutation2} reveals that the equality in~\eqref{eq:direct} is violated if $\hat{\alpha}_{n,\rho}$ is replaced by 
$\hat{\alpha}^{\mathrm{ind}}_{n,\rho}$.
Therefore, Theorem~\ref{thm:direct} does not hold if $\hat{\alpha}_{n,\rho}$ is replaced by $\hat{\alpha}_{n,\rho}^{\mathrm{ind}}$.

\subsection{Proof of Corollary~\ref{cor:limit-type1_direct}}\label{app:cor_direct}
\begin{cor:limit-type1_direct_repeated}[Asymptotic minimum type-I error]
Let $\rho_{AB}\in \mathcal{S}(AB)$ and let $R\in (-\infty,I(A:B)_\rho)$. Then 
$\lim_{n\rightarrow\infty}\hat{\alpha}_{n,\rho}(e^{-nR})=0$.
Moreover, the same is true if $\hat{\alpha}_{n,\rho}$ is replaced by 
$\hat{\alpha}_{n,\rho}^{\mathrm{iid}}$.
\end{cor:limit-type1_direct_repeated}

\begin{proof}
Let $\rho_{AB}\in \mathcal{S}(AB)$ and let $R\in (-\infty,I(A:B)_\rho)$. 
By Lemma~\ref{lem:order}~(b), 
\begin{equation}\label{eq:app-cor_direct}
0\leq \liminf_{n\rightarrow\infty}\hat{\alpha}_{n,\rho}(e^{-nR})
\leq \limsup_{n\rightarrow\infty}\hat{\alpha}_{n,\rho}(e^{-nR}).
\end{equation}
We will now prove that $\lim_{n\rightarrow\infty}\hat{\alpha}_{n,\rho}(e^{-nR})=0$ by cases.

Case 1: $R\in (0,I(A:B)_\rho)$. 
By the proof of achievability for Theorem~\ref{thm:direct}, see~\eqref{eq:ach-12},
\begin{equation}\label{eq:proof-cor-liminf_direct}
\liminf_{n\rightarrow\infty}-\frac{1}{n}\log \hat{\alpha}_{n,\rho}(e^{-nR})
\geq \sup_{s\in (\frac{1}{2},1)}\frac{1-s}{s}(I_s^{\downarrow\downarrow}(A:B)_\rho -R)>0,
\end{equation}
where the strict inequality follows from Theorem~\ref{thm:direct}. 
\eqref{eq:proof-cor-liminf_direct} implies that 
$\limsup_{n\rightarrow\infty}\hat{\alpha}_{n,\rho}(e^{-nR})=0$.
By~\eqref{eq:app-cor_direct}, it follows that 
$\lim_{n\rightarrow\infty}\hat{\alpha}_{n,\rho}(e^{-nR})=0$.

Case 2: $R\in (-\infty, 0]$. By Lemma~\ref{lem:order}~(a),~(b),
\begin{align}
\limsup_{n\rightarrow\infty}\hat{\alpha}_{n,\rho}(e^{-nR})
\leq \limsup_{n\rightarrow\infty}\hat{\alpha}_{n,\rho}(e^{0})
=0.
\end{align}
By~\eqref{eq:app-cor_direct}, it follows that 
$\lim_{n\rightarrow\infty}\hat{\alpha}_{n,\rho}(e^{-nR})=0$. 

This completes the proof of $\lim_{n\rightarrow\infty}\hat{\alpha}_{n,\rho}(e^{-nR})=0$. 
The assertion regarding $\hat{\alpha}_{n,\rho}^{\mathrm{iid}}$ follows from this due to Lemma~\ref{lem:order}~(b).
\end{proof}

\subsection{Proof of Corollary~\ref{cor:direct}}\label{app:proof-cor_direct}
\begin{lem}\label{lem:constant_prmi}
Let $\rho_{AB}\in \mathcal{S}(AB)$ and let $R$ be defined as in~\eqref{eq:def-R(s)-direct}. 
Let $a,b\in [\frac{1}{2},1]$ be such that $a<b$. 
Then, $(a,b)\rightarrow [0,\infty),s\mapsto I_s^{\downarrow\downarrow}(A:B)_\rho$ is constant iff 
$(a,b)\rightarrow [0,\infty),s\mapsto R(s)$ is constant. 
\end{lem}
\begin{proof}
Suppose $(a,b)\rightarrow [0,\infty),s\mapsto I_s^{\downarrow\downarrow}(A:B)_\rho$ is constant, i.e., $\exists c\in \mathbb{R}$ such that 
$I_s^{\downarrow\downarrow}(A:B)_\rho=c$ for all $s\in (a,b)$. 
Then, $\frac{\mathrm{d}}{\mathrm{d}s}I_s^{\downarrow\downarrow}(A:B)_\rho=0$ for all $s\in (a,b)$, so $R(s)=I_s^{\downarrow\downarrow}(A:B)_\rho=c$ for all $s\in (a,b)$. 

Now, suppose $(a,b)\rightarrow [0,\infty),s\mapsto R(s)$ is constant instead, 
i.e., $\exists c\in \mathbb{R}$ such that 
$I_s^{\downarrow\downarrow}(A:B)_\rho-s(1-s)\frac{\mathrm{d}}{\mathrm{d}s}I_s^{\downarrow\downarrow}(A:B)_\rho=c$ for all $s\in (a,b)$. 
The solution of this first-order linear ordinary differential equation is 
$I_s^{\downarrow\downarrow}(A:B)_\rho=\frac{c+ks}{1-s}$ where $k\in \mathbb{R}$ is a constant. 
We have
\begin{align}
0&=\lim_{s\rightarrow 1^-}(1-s)I_s^{\downarrow\downarrow}(A:B)_\rho 
=\lim_{s\rightarrow 1^-}(c+ks)
=c+k.
\end{align}
Hence, $k=-c$, which implies that $I_s^{\downarrow\downarrow}(A:B)_\rho =c$ for all $s\in (a,b)$.
\end{proof}

\begin{cor:direct_repeated}[Direct exponent]
Let $\rho_{AB}\in \mathcal{S}(AB)$ and let
\begin{align}\label{eq:def-R(s)-direct}
R:(1/2,1]\rightarrow [0,I(A:B)_\rho],\quad 
s\mapsto 
I_s^{\downarrow\downarrow}(A:B)_\rho-s(1-s)\frac{\mathrm{d}}{\mathrm{d} s}I_s^{\downarrow\downarrow}(A:B)_\rho.
\end{align}
Then $R$ is continuous and monotonically increasing. 
Let $R_{1/2}\coloneqq \lim_{s\rightarrow 1/2^+}R(s)$ and 
$R(1/2)\coloneqq R_{1/2}$. 
Let $s_{1/2}\coloneqq \max\{s\in [\frac{1}{2},1]:R(s)=R_{1/2}\}$ and 
$s_1\coloneqq\min\{s\in [\frac{1}{2},1]:R(s)=I(A:B)_\rho\}$.

Suppose $I_{1/2}^{\downarrow\downarrow}(A:B)_\rho\neq I(A:B)_\rho$. 
Then, $\frac{1}{2}= s_{1/2}<s_1\leq 1$ and for any $s\in (s_{1/2},s_1)$
\begin{equation}\label{eq:direct-cor}
\lim_{n\rightarrow\infty}-\frac{1}{n}\log \hat{\alpha}_{n,\rho}(e^{-nR(s)})
=\frac{1-s}{s}(I_s^{\downarrow\downarrow}(A:B)_\rho - R(s))
=(1-s)^2\frac{\mathrm{d}}{\mathrm{d}s} I_s^{\downarrow\downarrow}(A:B)_\rho.
\end{equation}
Moreover, the same is true if $\hat{\alpha}_{n,\rho}$ in~\eqref{eq:direct-cor} is replaced by 
$\hat{\alpha}_{n,\rho}^{\mathrm{iid}}$.
\end{cor:direct_repeated}

\begin{proof}
Let $\rho_{AB}\in \mathcal{S}(AB)$. 
The proof of Theorem~\ref{thm:direct} then implies that 
$R$ is continuous and monotonically increasing. 

From now on, suppose 
$I_{1/2}^{\downarrow\downarrow}(A:B)_\rho\neq I(A:B)_\rho$. 
Then, Lemma~\ref{lem:constant_prmi} implies that $R(1/2)\neq R(1)$. 
Hence, $\frac{1}{2}\leq s_{1/2}<s_1\leq 1$. 
The proof of Theorem~\ref{thm:direct} then implies that~\eqref{eq:direct-cor} holds for any $s\in (s_{1/2},s_1)$, and that the same is true if $\hat{\alpha}_{n,\rho}$ in~\eqref{eq:direct-cor} is replaced by 
$\hat{\alpha}_{n,\rho}^{\mathrm{iid}}$.

It remains to show that $s_{1/2}=\frac{1}{2}$. We will prove this by contradiction. 
Suppose $s_{1/2}>\frac{1}{2}$. 
Then, $R(s)=R_{1/2}$ for all $s\in [\frac{1}{2},s_{1/2}]$. 
By Lemma~\ref{lem:constant_prmi}, it follows that
\begin{align}\label{eq:proof_const1_prmi}
I_s^{\downarrow\downarrow}(A:B)_\rho=I_{1/2}^{\downarrow\downarrow}(A:B)_\rho
\qquad \forall s\in [1/2,s_{1/2}].
\end{align} 
Since $I_{1/2}^{\downarrow\downarrow}(A:B)_\rho\neq I(A:B)_\rho$, it follows from \cite[Theorem~3~(p)]{burri2025prmisrmi1} that there exists $s_*\in (s_{1/2},s_1)$ such that 
$\frac{\mathrm{d}}{\mathrm{d}s}I_s^{\downarrow\downarrow}(A:B)_\rho|_{s=s_*}\neq 0$. 
Then,
\begin{align}
0
&=\lim_{s\rightarrow s_{1/2}}\frac{\mathrm{d}}{\mathrm{d}s} I_s^{\downarrow\downarrow}(A:B)_\rho
\label{eq:proof_const2_prmi}\\
&=\lim_{s\rightarrow s_{1/2}^+} (1-s)^2 \frac{\mathrm{d}}{\mathrm{d}s} I_s^{\downarrow\downarrow}(A:B)_\rho
\\
&=\lim_{s\rightarrow s_{1/2}^+}\lim_{n\rightarrow\infty}-\frac{1}{n}\log \hat{\alpha}_{n,\rho}(e^{-nR(s)})
\label{eq:proof_const3_prmi}\\
&\geq \lim_{n\rightarrow\infty}-\frac{1}{n}\log \hat{\alpha}_{n,\rho}(e^{-nR(s_*)})
\label{eq:proof_const4_prmi}
\\
&=(1-s_*)^2 \frac{\mathrm{d}}{\mathrm{d}s} I_s^{\downarrow\downarrow}(A:B)_\rho \big|_{s=s^*}
>0.
\label{eq:proof_const5_prmi}
\end{align}
\eqref{eq:proof_const2_prmi} follows from~\eqref{eq:proof_const1_prmi} and \cite[Theorem~3~(p)]{burri2025prmisrmi1}. 
The equalities in~\eqref{eq:proof_const3_prmi} and~\eqref{eq:proof_const5_prmi} follow from above because~\eqref{eq:direct-cor} holds for any $s\in (s_{1/2},s_1)$. 
\eqref{eq:proof_const4_prmi} holds because $R$ is monotonically increasing, $\hat{\alpha}_{n,\rho}$ is monotonically decreasing (see Lemma~\ref{lem:order}), and $s_{1/2}\leq s_*$. 
Since~\eqref{eq:proof_const5_prmi} yields a contradiction, we can conclude that $s_{1/2}=\frac{1}{2}$.
\end{proof}

\subsection{Proof for Remark~\ref{rem:direct}}\label{app:rem_direct}
\begin{proof}
Let $\rho_{AB}\in \mathcal{S}(AB)$ be such that $V(A:B)_\rho\neq 0$. 
By \cite[Theorem~3~(n), (p)]{burri2025prmisrmi1} it follows that there exists $\alpha_0\in (\frac{1}{2},1)$ such that 
$I_\alpha^{\downarrow\downarrow}(A:B)_\rho$ is strictly monotonically increasing for $\alpha\in [\alpha_0,1]$. 
Let $R$ be defined as in~\eqref{eq:def-R(s)-direct}. 
By Lemma~\ref{lem:constant_prmi} and Corollary~\ref{cor:direct}, it follows that $R(\alpha)$ is strictly monotonically increasing for $\alpha\in [\alpha_0,1]$. 
Therefore, the parameter $s_1$ as defined Corollary~\ref{cor:direct} is given by $s_1=1$. 
The assertion in Remark~\ref{rem:direct} then follows from Corollary~\ref{cor:direct}.
\end{proof}

\subsection{Proof of Corollary~\ref{cor:forward-cutoff}}\label{ssec:forward-cutoff}
For simplicity, we will prove the assertion for $\hat{\alpha}_{n,\rho}$ only. 
The proof for $\hat{\alpha}_{n,\rho}^{\mathrm{iid}}$ is analogous due to Lemma~\ref{lem:order} and Corollary~\ref{cor:direct}. 
The proof is divided into two parts: 
a proof of achievability and a proof of optimality. 

\subsubsection{Proof of achievability}
Let $\rho_{AB}\in \mathcal{S}(AB)$.  
In the following, we will show that 
\begin{align}\label{eq:fc-geq-forward}
R_0^{(f)}(\beta)_\rho\geq I_{\frac{1}{1-\beta}}^{\downarrow\downarrow}(A:B)_\rho\quad \forall\beta\in (-1,0).
\end{align}
\begin{proof}
Let $\beta\in (-1,0)$ be arbitrary but fixed. 
Let $R_0\coloneqq I_{\frac{1}{1-\beta}}^{\downarrow\downarrow}(A:B)_\rho$. 
Then, for all $R\in (R_{1/2},\infty)$
\begin{subequations}\label{eq:fc12-forward}
\begin{align}
\liminf_{n\rightarrow\infty}-\frac{1}{n}\log \hat{\alpha}_{n,\rho}(e^{-nR})
&\geq \sup_{s\in (\frac{1}{2},1)}\frac{1-s}{s}(I_s^{\downarrow\downarrow}(A:B)_\rho - R)
\label{eq:fc1-forward}\\
&=\sup_{b\in (-1,0)}b(R-I_{\frac{1}{1-b}}^{\downarrow\downarrow}(A:B)_\rho )
\geq \beta (R-R_0).
\label{eq:fc2-forward}
\end{align}
\end{subequations}
\eqref{eq:fc1-forward} follows from the proof of achievability for Theorem~\ref{thm:direct}, see~\eqref{eq:ach-12}. 
The equality in~\eqref{eq:fc2-forward} follows from~\eqref{eq:fc1-forward} by defining $b\coloneqq \frac{s-1}{s}$. 
\eqref{eq:fc12-forward} implies that 
$R_0^{(f)}(\beta)_\rho\geq R_0 
= I_{\frac{1}{1-\beta}}^{\downarrow\downarrow}(A:B)_\rho$.
\end{proof}

\subsubsection{Proof of optimality}
Let $\rho_{AB}\in \mathcal{S}(AB)$ be such that $I_{1/2}^{\downarrow\downarrow}(A:B)_\rho \neq I(A:B)_\rho$. 
Then, $-1=\beta_{1/2}<\beta_1\leq 0$ follows from Corollary~\ref{cor:direct}. 
In the following, we will show that 
\begin{align}
R_0^{(f)}(\beta)_\rho\leq I_{\frac{1}{1-\beta}}^{\downarrow\downarrow}(A:B)_\rho \quad\forall \beta\in (\beta_{1/2},\beta_1).
\end{align}
\begin{proof}
Let $\beta\in (\beta_{1/2},\beta_1)$ be arbitrary but fixed. 
Let $s\coloneqq \frac{1}{1-\beta}\in (s_{1/2},s_1)$, and let $R(s)\coloneqq I_s^{\downarrow\downarrow}(A:B)_\rho-s(1-s)\frac{\mathrm{d}}{\mathrm{d} s}I_s^{\downarrow\downarrow}(A:B)_\rho$. 
Then, $R(s)\in (R_{1/2},I(A:B)_\rho)\subseteq(R_{1/2},\infty)$ follows from Corollary~\ref{cor:direct}. 
For all $R_0\in (I_{\frac{1}{1-\beta}}^{\downarrow\downarrow}(A:B)_\rho,\infty)$
\begin{subequations}\label{eq:fc34-forward}
\begin{align}
\liminf_{n\rightarrow\infty} -\frac{1}{n}\log \hat{\alpha}_{n,\rho}(e^{-nR(s)})
&= \frac{1-s}{s}(I_s^{\downarrow\downarrow}(A:B)_\rho -R(s))
\label{eq:fc3-forward}\\
&=\beta (R(s)-I_{\frac{1}{1-\beta}}^{\downarrow\downarrow}(A:B)_\rho )
<\beta (R(s)-R_0 ).
\label{eq:fc4-forward}
\end{align}
\end{subequations}
\eqref{eq:fc3-forward} follows from Corollary~\ref{cor:direct}. 
\eqref{eq:fc34-forward} implies that 
$R_0^{(f)}(\beta)_\rho\leq I_{\frac{1}{1-\beta}}^{\downarrow\downarrow}(A:B)_\rho$.
\end{proof}

\subsection{Proof of Theorem~\ref{thm:moderate_direct}}\label{proof:moderate_direct}
The proof of Theorem~\ref{thm:moderate_direct} is divided into two parts: 
a proof of achievability for $\hat{\alpha}_{n,\rho}$ and 
a proof of optimality for $\hat{\alpha}_{n,\rho}^{\mathrm{iid}}$. 
The assertion in Theorem~\ref{thm:moderate_direct} follows from these two parts because 
$\hat{\alpha}_{n,\rho}^{\mathrm{iid}}(\mu)\leq \hat{\alpha}_{n,\rho}(\mu)$ for all $\mu\in [0,\infty)$, see Lemma~\ref{lem:order}.

\subsubsection*{Proof of achievability}
We will show that 
\begin{align}\label{eq:mod-ach}
\liminf_{n\rightarrow\infty}-\frac{1}{na_n^2}\log \hat{\alpha}_{n,\rho}(e^{-nR_n})\geq \frac{1}{2V(A:B)_\rho}.
\end{align}
\begin{proof}
Let $I\coloneqq I(A:B)_\rho$ and $V\coloneqq V(A:B)_\rho$. 
Since $V>0$ by assumption, also $I>0$. 
Thus, for all sufficiently large $n\in \mathbb{N}$, we have $R_n=I-a_n>0$. 
Henceforth, we assume that this is the case.

Consider the first-order Taylor expansion of $I_{1/(1+s)}^{\downarrow\downarrow}(A:B)_\rho$ around $s=0$, and denote the corresponding error term by
\begin{align}\label{eq:mod-eps-def}
\varepsilon(s)
&\coloneqq I_{\frac{1}{1+s}}^{\downarrow\downarrow}(A:B)_\rho
-I_{1}^{\downarrow\downarrow}(A:B)_\rho
+s\frac{\partial}{\partial \alpha}I_{\alpha}^{\downarrow\downarrow}(A:B)_\rho\big|_{\alpha=1} 
\end{align}
for $s\in (0,1)$. 
For the second term, we have $I_1^{\downarrow\downarrow}(A:B)_\rho=I$~\cite{burri2025prmisrmi1}. 
For the last term, we have $\frac{\partial}{\partial \alpha}I_\alpha^{\downarrow\downarrow}(A:B)_\rho|_{\alpha=1}=V/2$~\cite{burri2025prmisrmi1}. 
Using that the doubly minimized PRMI of order $\alpha$ is continuously differentiable on $\alpha\in (1/2,2)$~\cite{burri2025prmisrmi1}, the first term can be expressed as
\begin{align}
I_{\frac{1}{1+s}}^{\downarrow\downarrow}(A:B)_\rho
&=I_{1-\frac{s}{1+s}}^{\downarrow\downarrow}(A:B)_\rho
=I -\int_{1-\frac{s}{1+s}}^1 \frac{\partial}{\partial \alpha} I_{\alpha}^{\downarrow\downarrow}(A:B)_\rho \,  \mathrm{d}\alpha
\\
&=I
- \int_{1-\frac{s}{1+s}}^1 \left(\frac{\partial}{\partial \alpha}I_{\alpha}^{\downarrow\downarrow}(A:B)_\rho -  \frac{\partial}{\partial \beta}I_{\beta}^{\downarrow\downarrow}(A:B)_\rho\big|_{\beta=1} \right) \mathrm{d}\alpha
-\frac{V}{2} \frac{s}{1+s}.
\label{eq:mod-taylors}
\end{align}
To estimate the absolute value of the second term in~\eqref{eq:mod-taylors}, we define for $s\in (0,1)$
\begin{align}\label{eq:mod-delta}
\delta(s)
&\coloneqq \sup_{\alpha\in (1-\frac{s}{1+s},1)} \Big| \frac{\partial}{\partial\alpha}I_\alpha^{\downarrow\downarrow}(A:B)_\rho - \frac{\partial}{\partial\beta}I_\beta^{\downarrow\downarrow}(A:B)_\rho\big|_{\beta=1} \Big| .
\end{align}
Combining these results, 
we conclude that 
\begin{align}\label{eq:eps-abs}
|\varepsilon(s)| 
&\leq \frac{s}{1+s} \delta(s) +\frac{V}{2} \Big|s-\frac{s}{1+s} \Big|
= \frac{s}{1+s}\left(\delta(s)+\frac{sV}{2}\right)
\leq s\left(\delta (s)+\frac{sV}{2}\right) .
\end{align}

We have
\begin{align}
&-\frac{1}{na_n^2}\log \hat{\alpha}_{n,\rho}(e^{-nR_n})
\nonumber\\
&\geq \frac{1}{a_n^2} \sup_{\alpha\in (\frac{1}{2},1)}
\frac{1-\alpha}{\alpha} \left(\frac{1}{n}D_\alpha (\rho_{AB}^{\otimes n}\| \omega_{A^n}^n\otimes \omega_{B^n}^n) -R_n-\frac{\log g_{n,d_A}}{n} -\frac{\log g_{n,d_B}}{n} \right)
\label{eq:mod-dir1}\\
&\geq \frac{1}{a_n^2} \sup_{s\in (0,1)} 
s
\left( I_{\frac{1}{1+s}}^{\downarrow\downarrow}(A:B)_\rho -R_n- \frac{\log g_{n,d_A}}{n} -\frac{\log g_{n,d_B}}{n} \right)
\label{eq:mod-dir2}\\
&\geq \frac{1}{a_n^2} \sup_{s\in (0,1)} 
\left(-\frac{s^2V}{2} -s^2 \left(\delta (s)+\frac{sV}{2}\right)
+sa_n- s\frac{\log g_{n,d_A}}{n}-s\frac{\log g_{n,d_B}}{n} \right)
\label{eq:mod-dir3}
\\
&\geq\frac{1}{2V}
-\frac{1}{V^2}\delta\left(\frac{a_n}{V}\right)
-\frac{a_n}{2V^2}
-\frac{1}{V}\frac{1}{\sqrt{n}a_n} \left(\frac{\log g_{n,d_A}}{\sqrt{n}} + \frac{\log g_{n,d_B}}{\sqrt{n}} \right).
\label{eq:mod-dir4}
\end{align}
\eqref{eq:mod-dir1} follows from the same arguments as in~\eqref{eq:ach-two}--\eqref{eq:ach-one}. 
\eqref{eq:mod-dir2} follows from the previous line by setting $s\coloneqq\frac{1-\alpha}{\alpha}$ and by noting that $\frac{1}{n}D_\alpha (\rho_{AB}^{\otimes n}\| \omega_{A^n}^n\otimes \omega_{B^n}^n) \geq \frac{1}{n} I_\alpha^{\downarrow\downarrow}(A^n:B^n)_{\rho^{\otimes n}}=I_\alpha^{\downarrow\downarrow}(A:B)_\rho$ for $\alpha\in (\frac{1}{2},1)$, where the last equality holds by additivity~\cite{burri2025prmisrmi1}. 
\eqref{eq:mod-dir3} follows from the definition of $R_n$, $I_{1/(1+s)}^{\downarrow\downarrow}(A:B)_\rho =I-sV/2+\varepsilon(s)$, see~\eqref{eq:mod-eps-def}, and from~\eqref{eq:eps-abs}. 
\eqref{eq:mod-dir4} follows from the previous line by setting $s\coloneqq a_n/V$. 

Let us consider the limit of~\eqref{eq:mod-dir4} as $n\rightarrow\infty$. 
The second term in~\eqref{eq:mod-dir4} vanishes because $\delta$ is continuous (as $I_\alpha^{\downarrow\downarrow}(A:B)_\rho$ is continuously differentiable on $\alpha\in (\frac{1}{2},2)$~\cite{burri2025prmisrmi1}) and $\lim_{n\rightarrow\infty}a_n=0$. 
The third term in~\eqref{eq:mod-dir4} vanishes because $\lim_{n\rightarrow\infty}a_n= 0$. 
The last term in~\eqref{eq:mod-dir4} vanishes because $\lim_{n\rightarrow\infty}\sqrt{n}a_n= \infty$ by assumption and $\lim_{n\rightarrow\infty}(\log g_{n,d_A})/\sqrt{n}=0=\lim_{n\rightarrow\infty}(\log g_{n,d_A})/\sqrt{n}$, see~\cite[Remark~1~(b)]{burri2025prmisrmi1}. 
Thus, the claim in~\eqref{eq:mod-ach} follows from~\eqref{eq:mod-dir4}.
\end{proof}

\subsubsection*{Proof of optimality}
We will show that 
\begin{align}
\limsup_{n\rightarrow\infty}-\frac{1}{na_n^2}\log \hat{\alpha}_{n,\rho}^{\mathrm{iid}}(e^{-nR_n})\leq \frac{1}{2V(A:B)_\rho}.
\end{align}
\begin{proof}
Let us define the following function of $\mu\in [0,\infty)$.
\begin{align}\label{eq:def-alpha-mar}
\hat{\alpha}_{n,\rho}^{\mathrm{mar}}(\mu )
&\coloneqq \min_{\substack{T^n_{A^nB^n}\in \mathcal{L} (A^nB^n): \\ 0\leq T^n_{A^nB^n}\leq 1}}
\{\tr[\rho_{AB}^{\otimes n}(1-T^n_{A^nB^n})]:
\tr[\rho_{A}^{\otimes n}\otimes \rho_{B}^{\otimes n}\, T^n_{A^nB^n}]
\leq \mu\}
\end{align}
Then 
$\hat{\alpha}_{n,\rho}^{\mathrm{iid}}(\mu )\geq 
\hat{\alpha}_{n,\rho}^{\mathrm{mar}}(\mu )$
for all $\mu\in [0,\infty)$. Therefore, 
\begin{align}
\limsup_{n\rightarrow\infty}-\frac{1}{na_n^2}\log \hat{\alpha}_{n,\rho}^{\mathrm{iid}}(e^{-nR_n})
&\leq \limsup_{n\rightarrow\infty}-\frac{1}{na_n^2}\log \hat{\alpha}_{n,\rho}^{\mathrm{mar}}(e^{-nR_n})
= \frac{1}{2V(A:B)_\rho}.
\end{align}
The equality follows from~\cite[Theorems~10 and~11]{Cheng2017ModerateDA} (see also~\cite{Chubb_2017}).
\end{proof}

\section{Proofs for Section~\ref{ssec:strongconverse}}

\subsection{Proof of Theorem~\ref{thm:strong-converse}}\label{app:strong-converse}

First, we will derive the bounds on the right-hand side of~\eqref{eq:strong-converse}. 
Let $\rho_{AB}\in \mathcal{S}(AB)$. Then, for any $R\in [0,\infty)$
\begin{align}
0&=\lim_{s\rightarrow 1^+}\frac{s-1}{s}(R-\widetilde{I}_s^{\downarrow\downarrow}(A:B)_\rho )
\leq \sup_{s\in (1,\infty)}\frac{s-1}{s}(R-\widetilde{I}_s^{\downarrow\downarrow}(A:B)_\rho )
\\
&\leq \max(0, \sup_{s\in (1,\infty)}(R-\widetilde{I}_s^{\downarrow\downarrow}(A:B)_\rho ))
= \max(0,R-I(A:B)_\rho)
\end{align}
due to the monotonicity and continuity of the doubly minimized SRMI in the R\'enyi order 
and $\widetilde{I}_1^{\downarrow\downarrow}(A:B)_\rho=I(A:B)_\rho$, 
see \cite[Theorem~4~(k), (l), (m)]{burri2025prmisrmi1}. 
These bounds imply that for any $R\in [0, I(A:B)_\rho]$, the right-hand side of~\eqref{eq:strong-converse} vanishes. 
If $R\in (I(A:B)_\rho,\infty)$ instead, then the right-hand side of~\eqref{eq:strong-converse} is strictly positive 
due to \cite[Theorem~4~(k), (l), (m)]{burri2025prmisrmi1}.

We will now prove the equality in~\eqref{eq:strong-converse}. 
The proof of~\eqref{eq:strong-converse} is divided into two parts: 
a proof of achievability for $\hat{\alpha}_{n,\rho}$ and a proof of optimality for $\hat{\alpha}_{n,\rho}^{\mathrm{iid}}$. 
The assertion follows from these two parts because 
$\hat{\alpha}_{n,\rho}^{\mathrm{iid}}(\mu)\leq \hat{\alpha}_{n,\rho}(\mu)$ for all $\mu\in [0,\infty)$, see Lemma~\ref{lem:order}.
Below, we first give the proof of achievability, followed by the proof of optimality.

\subsubsection{Proof of achievability}\label{proof:achievability}
Let $\rho_{AB}\in \mathcal{S}(AB)$ be such that 
$I(A:B)_\rho\neq \widetilde{I}_\infty^{\downarrow\downarrow}(A:B)_\rho$. 
In the following, we will show that for any $R\in [0,\infty)$ 
\begin{align}\label{eq:ach-0}
\limsup_{n\rightarrow\infty}-\frac{1}{n}\log (1- \hat{\alpha}_{n,\rho}(e^{-nR}))
\leq \sup_{s\in (1,\infty)}\frac{s-1}{s}(R-\widetilde{I}_s^{\downarrow\downarrow}(A:B)_\rho ).
\end{align}
\begin{proof}
Let $R\in [0,\infty)$ be arbitrary but fixed. 
Let 
$R_\infty \coloneqq \lim_{s\rightarrow\infty} ( \widetilde{I}_s^{\downarrow\downarrow}(A:B)_\rho +s(s-1)\frac{\mathrm{d}}{\mathrm{d}s}\widetilde{I}_s^{\downarrow\downarrow}(A:B)_\rho )$. 
Then $R_\infty\in [\widetilde{I}_\infty^{\downarrow\downarrow} (A:B)_\rho,\infty]$ due to the monotonicity in the R\'enyi order of the doubly minimized SRMI, see \cite[Theorem~4~(l)]{burri2025prmisrmi1}.

Case 1: $R\in (I(A:B)_\rho,R_\infty)$. 
First, we will analyze the right-hand side of~\eqref{eq:ach-0}. 
Let us define the following functions of $s\in (1,\infty)$.
\begin{align}
\phi(s)&\coloneqq (s-1)\widetilde{I}_s^{\downarrow\downarrow}(A:B)_\rho 
\label{eq:def-phi}\\
\psi(s)&\coloneqq s\phi'(s)-\phi(s)
=\widetilde{I}_s^{\downarrow\downarrow}(A:B)_\rho+s(s-1)\frac{\mathrm{d}}{\mathrm{d} s}\widetilde{I}_s^{\downarrow\downarrow}(A:B)_\rho
\label{eq:def-psi}\\
g(s)&\coloneqq \frac{1}{s} ((s-1)R-\phi(s))
=\frac{s-1}{s}(R-\widetilde{I}_s^{\downarrow\downarrow}(A:B)_\rho )
\label{eq:def-gs}
\end{align}
$\phi$ is continuously differentiable and convex due to \cite[Theorem~4~(n), (o)]{burri2025prmisrmi1}. 
This implies that $\psi$ is continuous and monotonically increasing~\cite[Lemma~20]{tomamichel2018operational}. 
As a consequence, $g'(s)=\frac{1}{s^2}(R-\psi(s))$ is continuous and monotonically decreasing. 

On the one hand, $\lim_{s\rightarrow 1^+}\psi(s)=\widetilde{I}_1^{\downarrow\downarrow}(A:B)_\rho=I(A:B)_\rho$.
Hence, $\lim_{s\rightarrow 1^+}g'(s)=R-I(A:B)_\rho>0$.
On the other hand, $\lim_{s\rightarrow\infty}\psi(s)=R_\infty$. 
Hence, there exists $t_0\in (1,\infty)$ such that $R<\psi(t_0)$. 
As a consequence, $g'(t_0)=\frac{1}{t_0^2}(R-\psi(t_0))<0$. 
By the continuity and monotonicity of $g'$, we can conclude that there exists $\hat{s}\in (1,t_0)$ such that $g'(\hat{s})=0$ and 
\begin{align}\label{eq:shat-max}
\sup_{s\in (1,\infty)}\frac{s-1}{s}(R-\widetilde{I}_s^{\downarrow\downarrow}(A:B)_\rho )
=\frac{\hat{s}-1}{\hat{s}}(R-\widetilde{I}_{\hat{s}}^{\downarrow\downarrow}(A:B)_\rho ).
\end{align}

Let us define the following function of $t\in (0,\infty)$. 
\begin{align}\label{eq:def-lambda}
\Lambda(t)&\coloneqq 
\phi(t+1)-\frac{t}{\hat{s}}\left(\phi(\hat{s})+R\right)
=\phi(t+1)-t\phi'(\hat{s})
\end{align}
For the second equality in~\eqref{eq:def-lambda}, we have used that $g'(\hat{s})=0$. 
For any $t\in (1,\infty)$, the derivative of $\Lambda$ at $t-1$ is given by
\begin{align}
\Lambda'(t-1)
&=\phi'(t)-\frac{1}{\hat{s}}(\phi(\hat{s})+R)
=\widetilde{I}_t^{\downarrow\downarrow}(A:B)_\rho+(t-1)\frac{\mathrm{d}}{\mathrm{d}t}\widetilde{I}_t^{\downarrow\downarrow}(A:B)_\rho-\frac{1}{\hat{s}}(\phi(\hat{s})+R)
\label{eq:lambda-derivative1}\\
&=\phi'(t)-\frac{1}{\hat{s}}(\hat{s}\phi'(\hat{s})-R+R)
=\phi'(t)-\phi'(\hat{s})
=\frac{1}{t}(\psi(t)+\phi(t))-\frac{1}{\hat{s}}(\psi(\hat{s})+\phi(\hat{s})) ,
\label{eq:lambda-derivative2}
\end{align}
where we have used in the second line that $g'(\hat{s})=0$. 
We have 
\begin{subequations}\label{eq:ge-conditions}
\begin{align}
\lim_{t\rightarrow 1^+}\Lambda'(t-1)
&=\widetilde{I}_1^{\downarrow\downarrow}(A:B)_\rho-\frac{1}{\hat{s}}(\phi(\hat{s})+R)
<\frac{\hat{s}-1}{\hat{s}}(I_1^{\downarrow\downarrow}(A:B)_\rho-I_{\hat{s}}^{\downarrow\downarrow}(A:B)_\rho)\leq 0,
\label{eq:lambda-1}
\\
\Lambda'(t_0-1)&=\frac{1}{t_0}(\psi(t_0)+\phi(t_0))-\frac{1}{\hat{s}}(\psi(\hat{s})+\phi(\hat{s}) ) 
\label{eq:lambda-2}
\\
&>(R-g(t_0))-(R-g(\hat{s}))=g(\hat{s})-g(t_0)\geq 0.
\label{eq:lambda-3}
\end{align}
\end{subequations}
\eqref{eq:lambda-1} follows from~\eqref{eq:lambda-derivative1} and 
$R>I(A:B)_\rho=I_1^{\downarrow\downarrow}(A:B)_\rho$, see \cite[Theorem~4~(k), (l)]{burri2025prmisrmi1}. 
\eqref{eq:lambda-2} follows from~\eqref{eq:lambda-derivative2}. 
The first inequality in~\eqref{eq:lambda-3} follows from $\psi(t_0)>R$ and $g'(\hat{s})=0$. 
The second inequality in~\eqref{eq:lambda-3} follows from~\eqref{eq:shat-max}. 

We will now analyze the left-hand side of~\eqref{eq:ach-0}. 
For any $n\in \mathbb{N}_{>0}$, let us define the test
\begin{align}\label{eq:test-np}
T^n_{A^nB^n}\coloneqq \{\mathcal{P}_{\omega_{A^n}^n\otimes \omega_{B^n}^n}(\rho_{AB}^{\otimes n})
\geq e^{\mu_n}\omega_{A^n}^n\otimes \omega_{B^n}^n\},
\end{align}
where $\mu_n\in \mathbb{R}$ is a trade-off parameter that will be specified later on.  
Let $\{\ket{\phi_{x_n}}\}_{x_n\in [d_A^nd_B^n]}$ be an orthonormal basis of $A^{\otimes n}\otimes B^{\otimes n}$ that diagonalizes both $\mathcal{P}_{\omega_{A^n}^n\otimes \omega_{B^n}^n}(\rho_{AB}^{\otimes n})$ and $\omega_{A^n}^n\otimes \omega_{B^n}^n$, and let us define the PMFs $P_n$ and $Q_n$ as follows.
\begin{align}
[d_A^nd_B^n]\rightarrow [0,1],\quad 
x_n\mapsto
P_n(x_n)&\coloneqq \bra{\phi_{x_n}} \mathcal{P}_{\omega_{A^n}^n\otimes \omega_{B^n}^n}(\rho_{AB}^{\otimes n}) \ket{\phi_{x_n}}
\label{eq:test-p}\\
[d_A^nd_B^n]\rightarrow (0,1],\quad 
x_n\mapsto
Q_n(x_n)&\coloneqq \bra{\phi_{x_n}} \omega_{A^n}^n\otimes \omega_{B^n}^n \ket{\phi_{x_n}}
\label{eq:test-q}
\end{align}

Let $X_n$ be the random variable over the alphabet $[d_A^nd_B^n]$ whose PMF is $P_n$. 
Then
\begin{subequations}\label{eq:ach-at}
\begin{align}
&\tr[\rho_{AB}^{\otimes n}(1-T^n_{A^nB^n})]
\\
&=\tr[\mathcal{P}_{\omega_{A^n}^n\otimes \omega_{B^n}^n}(\rho_{AB}^{\otimes n}) \{ \mathcal{P}_{\omega_{A^n}^n\otimes \omega_{B^n}^n}(\rho_{AB}^{\otimes n}) <e^{\mu_n}\omega_{A^n}^n\otimes \omega_{B^n}^n\}]\\
&= \sum_{x_n\in [d_A^nd_B^n]}\bra{\phi_{x_n}} \mathcal{P}_{\omega_{A^n}^n\otimes \omega_{B^n}^n}(\rho_{AB}^{\otimes n}) \proj{\phi_{x_n}}  \{ \mathcal{P}_{\omega_{A^n}^n\otimes \omega_{B^n}^n}(\rho_{AB}^{\otimes n}) <e^{\mu_n}\omega_{A^n}^n\otimes \omega_{B^n}^n\} \ket{\phi_{x_n}} \\
&=\sum_{x_n\in [d_A^nd_B^n]}P_n(x_n) \delta(P_n(x_n)<e^{\mu_n} Q_n(x_n)) \\
&=\mathrm{Pr}[P_n(X_n)<e^{\mu_n}Q_n(X_n)].
\end{align}
\end{subequations}
It follows from~\eqref{eq:ach-at} that
\begin{subequations}\label{eq:ach-alpha}
\begin{align}
\tr[\rho_{AB}^{\otimes n}T^n_{A^nB^n}]
&=\mathrm{Pr}[P_n(X_n)\geq e^{\mu_n}Q_n(X_n)]\\
&=\mathrm{Pr}\left[\frac{1}{n} (\log P_n(X_n)- \log Q_n(X_n)-\mu_n) \geq 0\right] 
=\mathrm{Pr}[Z_n\geq 0].
\end{align}
\end{subequations}
For the last equality, we defined the random variable
\begin{equation}\label{eq:def-zn}
Z_n\coloneqq \frac{1}{n} (\log P_n(X_n)- \log Q_n(X_n)-\mu_n),
\end{equation}
where we use the convention that $\log P_n(x_n)=-\infty$ if $P_n(x_n)=0$.

Let $X_n'$ be the random variable over the alphabet $[d_A^nd_B^n]$ whose PMF is $Q_n$. Then
\begin{subequations}\label{eq:ach-gg0}
\begin{align}
\sup_{\substack{\sigma_{A^n}\in \mathcal{S}_{\sym}(A^{\otimes n}),\\ \tau_{B^n}\in \mathcal{S}_{\sym}(B^{\otimes n})}} \tr[\sigma_{A^n}\otimes \tau_{B^n} T^n_{A^nB^n}]
&\leq g_{n,d_A}g_{n,d_B}\tr[\omega_{A^n}^n\otimes \omega_{B^n}^nT^n_{A^nB^n}]\\
&= g_{n,d_A}g_{n,d_B} \sum_{x_n\in [d_A^nd_B^n]}\bra{\phi_{x_n}}\omega_{A^n}^n\otimes \omega_{B^n}^n \proj{\phi_{x_n}} T^n_{A^nB^n} \ket{\phi_{x_n}}\\
&= g_{n,d_A}g_{n,d_B} \sum_{x_n\in [d_A^nd_B^n]} Q_n(x_n)\delta(P_n(x_n)\geq e^{\mu_n}Q_n(x_n) )\\
&= g_{n,d_A}g_{n,d_B}\mathrm{Pr}[e^{-\mu_n}P_n(X_n')\geq Q_n(X_n')].
\label{eq:ach-gg}
\end{align}
\end{subequations}

Let us now define 
\begin{align}
\mu_n\coloneqq \frac{1}{\hat{s}} (\log g_{n,d_A}+\log g_{n,d_B}+nR+(\hat{s}-1)D_{\hat{s}}(P_n\| Q_n) ).
\label{eq:def-mun}
\end{align}
Then, 
\begin{subequations}\label{eq:ach-beta}
\begin{align}
\sup_{\substack{\sigma_{A^n}\in \mathcal{S}_{\sym}(A^{\otimes n}),\\ \tau_{B^n}\in \mathcal{S}_{\sym}(B^{\otimes n})}} \tr[\sigma_{A^n}\otimes \tau_{B^n} T^n_{A^nB^n}]
&\leq g_{n,d_A}g_{n,d_B}\sum_{x_n\in [d_A^nd_B^n]}(e^{-\mu_n} P_n(x_n))^{\hat{s}} Q_n(x_n)^{1-\hat{s}}
\label{eq:ach-beta1}\\
&=g_{n,d_A}g_{n,d_B}e^{-\hat{s}\mu_n} \exp((\hat{s}-1)D_{\hat{s}}(P_n\| Q_n)) 
=e^{-nR}.
\label{eq:ach-beta2}
\end{align}
\end{subequations}
\eqref{eq:ach-beta1} follows from~\eqref{eq:ach-gg0} and~\cite[Eq.~(2.2)]{hayashi2016correlation}. 
\eqref{eq:ach-beta2} follows from~\eqref{eq:def-mun}. 
\eqref{eq:ach-alpha} and~\eqref{eq:ach-beta} imply that
\begin{equation}\label{eq:alpha-prz}
\limsup_{n\rightarrow\infty}-\frac{1}{n}\log (1- \hat{\alpha}_{n,\rho}(e^{-nR}))
\leq \limsup_{n\rightarrow\infty}-\frac{1}{n}\log \tr[\rho_{AB}^{\otimes n}T^n_{A^nB^n}]
= \limsup_{n\rightarrow\infty} -\frac{1}{n}\log \mathrm{Pr}[Z_n\geq 0] .
\end{equation}

We will now show that the asymptotic cumulant generating function of $Z_n$ coincides with $\Lambda$. 
For any $t\in (0,\infty)$,
\begin{subequations}\label{eq:z0}
\begin{align}
\Lambda(t)
&=t \left( \widetilde{I}_{1+t}^{\downarrow\downarrow}(A:B)_\rho -\frac{\hat{s}-1}{\hat{s}} \widetilde{I}_{\hat{s}}^{\downarrow\downarrow}(A:B)_\rho \right) - \frac{t}{\hat{s}}R
\\
&=t\lim_{n\rightarrow\infty} \left(\frac{1}{n} D_{1+t}(P_n\| Q_n)-\frac{\hat{s}-1}{\hat{s}}\frac{1}{n} D_{\hat{s}}(P_n\| Q_n)\right) - \frac{t}{\hat{s}}R
\label{eq:z1}\\
&=\lim_{n\rightarrow\infty}\left(\frac{1}{n}\log \mathbb{E}\left[\frac{P_n(X_n)^t}{Q_n(X_n)^t}\right]
-\frac{t}{\hat{s}}\frac{\log(g_{n,d_A}g_{n,d_B})}{n}-\frac{t}{\hat{s}} \frac{(\hat{s}-1)}{n}D_{\hat{s}}(P_n\| Q_n) \right)-\frac{t}{\hat{s}}R 
\label{eq:z11}\\
&= \lim_{n\rightarrow\infty}\frac{1}{n}\log \mathbb{E}[\exp(ntZ_n)].
\label{eq:z2}
\end{align}
\end{subequations}
\eqref{eq:z1} follows from the asymptotic optimality of universal permutation invariant state for the doubly minimized SRMI \cite[Theorem~4~(j)]{burri2025prmisrmi1}. 
\eqref{eq:z11} follows from \cite[Remark~1~(b)]{burri2025prmisrmi1}.
\eqref{eq:z2} follows from~\eqref{eq:def-zn}.

Next, we apply the G\"{a}rtner-Ellis lower bound from~\cite[Proposition 17]{hayashi2016correlation} (see also~\cite[Theorem 3.6]{chen2000generalization}).
The conditions for applying this proposition are fulfilled since 
$\lim_{t\rightarrow 0^+}\Lambda'(t)<0$ and $\Lambda'(t_0-1)>0$, see~\eqref{eq:ge-conditions}.
Thus, we can infer from the combination of \cite[Proposition 17]{hayashi2016correlation} with~\eqref{eq:z0} that 
\begin{subequations}\label{eq:ge}
\begin{align}
\limsup_{n\rightarrow\infty} -\frac{1}{n}\log \mathrm{Pr}[Z_n\geq 0] 
&\leq \sup_{t\in (0,t_0-1)} -\Lambda(t)
=\sup_{t\in (1,t_0)} -\Lambda(t-1)
\label{eq:ge1}\\
&=\sup_{t\in (1,t_0)}(-\phi(t)+(t-1)\phi'(\hat{s}))
\label{eq:ge2}\\
&=-\phi(\hat{s})+(\hat{s}-1)\phi'(\hat{s})
\label{eq:ge3}
\\
&=-\phi(\hat{s})+\frac{\hat{s}-1}{\hat{s}} (R+\phi(\hat{s}))
=\frac{\hat{s}-1}{\hat{s}} (R-\widetilde{I}_{\hat{s}}^{\downarrow\downarrow}(A:B)_\rho ).
\label{eq:ge4}
\end{align}
\end{subequations}
\eqref{eq:ge2} follows from~\eqref{eq:def-lambda}.
\eqref{eq:ge3} holds because the objective function $t\mapsto (t-1)\phi'(s)-\phi(t)$ 
is concave in $t\in (1,\infty)$, 
and its first derivative at $t=\hat{s}\in (1,t_0)$ is zero. 
\eqref{eq:ge4} follows from $g'(\hat{s})=0$. 
The combination of~\eqref{eq:shat-max},~\eqref{eq:alpha-prz}, and~\eqref{eq:ge} implies the assertion in~\eqref{eq:ach-0}. 
This completes the proof for case 1.

Case 2: $R_\infty<\infty$ and $R\in [R_\infty,\infty)$. 
Let $R'\in (I(A:B)_\rho,R_\infty)$.
Let $T^n_{A^n B^n}(R')$ denote the test that was defined in case 1 (where $R$ in case 1 is replaced by $R'$). 
Let us define the test
\begin{align}\label{eq:def-trr}
T^n_{A^n B^n}(R,R')\coloneqq e^{-n(R-R')}T^n_{A^n B^n}(R').
\end{align}
Since $R'\leq R_\infty\leq R$, we have $e^{-n(R-R')}\in [0,1]$, so $0\leq T^n_{A^n B^n}(R,R')\leq 1$. 
By~\eqref{eq:ach-beta},
\begin{align}\label{eq:c20}
\sup_{\substack{\sigma_{A^n}\in \mathcal{S}_{\sym}(A^{\otimes n}),\\ \tau_{B^n}\in \mathcal{S}_{\sym}(B^{\otimes n})}} &\tr[\sigma_{A^n}\otimes \tau_{B^n} T^n_{A^nB^n}(R,R')]
\leq e^{-n(R-R')} e^{-nR'}=e^{-nR}.
\end{align}
Therefore,
\begin{align}
\limsup_{n\rightarrow\infty}-\frac{1}{n}\log (1- \hat{\alpha}_{n,\rho}(e^{-nR}))
&\leq \limsup_{n\rightarrow\infty}-\frac{1}{n}\log \tr[\rho_{AB}^{\otimes n}T^n_{A^nB^n}(R,R')]
\label{eq:c21}\\
&= \limsup_{n\rightarrow\infty}-\frac{1}{n}\log \tr[\rho_{AB}^{\otimes n}T^n_{A^nB^n}(R')]
+R-R'
\label{eq:c22}\\
&\leq \sup_{s\in (1,\infty)}\frac{s-1}{s}(R'-\widetilde{I}_s^{\downarrow\downarrow}(A:B)_\rho )+R-R'.
\label{eq:c23}
\end{align}
\eqref{eq:c21} follows from~\eqref{eq:c20}. 
\eqref{eq:c22} follows from~\eqref{eq:def-trr}. 
\eqref{eq:c23} follows from the proof for case 1. 
By the proof for case 1, the supremum in~\eqref{eq:c23} is achieved by $s\rightarrow\infty$ in the limit where $R'\rightarrow R_\infty$ from below. 
Therefore, 
\begin{align}
\limsup_{n\rightarrow\infty}-\frac{1}{n}\log (1- \hat{\alpha}_{n,\rho}(e^{-nR}))
&\leq\lim_{s\rightarrow\infty}\frac{s-1}{s}(R_\infty-\widetilde{I}_s^{\downarrow\downarrow}(A:B)_\rho ) +R-R_\infty
\\
&=R_\infty-\widetilde{I}_{\infty}^{\downarrow\downarrow}(A:B)_\rho +R-R_\infty 
=R-\widetilde{I}_{\infty}^{\downarrow\downarrow}(A:B)_\rho
\\
&=\lim_{s\rightarrow\infty}\frac{s-1}{s}(R-\widetilde{I}_s^{\downarrow\downarrow}(A:B)_\rho ) 
\label{eq:c25}\\
&\leq \sup_{s\in (1,\infty)}\frac{s-1}{s}(R-\widetilde{I}_s^{\downarrow\downarrow}(A:B)_\rho ), 
\label{eq:c26}
\end{align}
where we have used the continuity of the doubly minimized SRMI in the R\'enyi order, see \cite[Theorem~4~(m)]{burri2025prmisrmi1}.

Case 3: $R\in [0,I(A:B)_\rho]$. Then,
\begin{align}
\limsup_{n\rightarrow\infty}-\frac{1}{n}\log (1- \hat{\alpha}_{n,\rho}(e^{-nR}))
&\leq \inf_{R'\in (I(A:B)_\rho,\widetilde{I}_\infty^{\downarrow\downarrow}(A:B)_\rho)} \limsup_{n\rightarrow\infty}-\frac{1}{n}\log (1- \hat{\alpha}_{n,\rho}(e^{-nR'}))
\label{eq:c3-1}\\
&\leq \inf_{R'\in (I(A:B)_\rho,\widetilde{I}_\infty^{\downarrow\downarrow}(A:B)_\rho)} 
\sup_{s\in (1,\infty)} \frac{s-1}{s} (R'-\widetilde{I}_{s}^{\downarrow\downarrow}(A:B)_\rho )
\label{eq:c3-2}\\
&\leq \inf_{R'\in (I(A:B)_\rho,\widetilde{I}_\infty^{\downarrow\downarrow}(A:B)_\rho)} 
\sup_{s\in (1,\infty)} (R'-\widetilde{I}_{s}^{\downarrow\downarrow}(A:B)_\rho )
=0
\label{eq:c3-22}\\
&= \lim_{s\rightarrow 1^+}\frac{s-1}{s}(R-\widetilde{I}_s^{\downarrow\downarrow}(A:B)_\rho )
\label{eq:c3-3}\\
&\leq \sup_{s\in (1,\infty)}\frac{s-1}{s}(R-\widetilde{I}_s^{\downarrow\downarrow}(A:B)_\rho ).
\end{align}
\eqref{eq:c3-1} follows from the monotonicity of the minimum type-I error, see Lemma~\ref{lem:order}. 
\eqref{eq:c3-2} follows from case 1 because 
$I(A:B)_\rho\neq \widetilde{I}_\infty^{\downarrow\downarrow}(A:B)_\rho$. 
\eqref{eq:c3-22} and~\eqref{eq:c3-3} follow from 
$\widetilde{I}_1^{\downarrow\downarrow}(A:B)_\rho=I(A:B)_\rho$ 
and the monotonicity and continuity of the doubly minimized SRMI in the R\'enyi order, see \cite[Theorem~4~(k), (l), (m)]{burri2025prmisrmi1}.
\end{proof}

\subsubsection{Proof of optimality}\label{sec:optimality}
Let $\rho_{AB}\in \mathcal{S}(AB)$. 
In the following, we will show that for any $R\in [0,\infty)$ 
\begin{equation}\label{eq:sc-optimality}
\liminf_{n\rightarrow\infty}-\frac{1}{n}\log (1- \hat{\alpha}_{n,\rho}^{\mathrm{iid}}(e^{-nR}))
\geq \sup_{s\in (1,\infty)}\frac{s-1}{s}(R-\widetilde{I}_s^{\downarrow\downarrow}(A:B)_\rho ).
\end{equation}

\begin{proof}
Let $R\in [0,\infty)$ be arbitrary but fixed.

Case 1: $\rho_{AB}\neq \rho_A\otimes \rho_B$. 
According to the strong converse bound in~\cite[Lemma~4.7]{mosonyi2014quantum}, we have for any $(\sigma_A,\tau_B)\in \mathcal{S}(A)\times \mathcal{S}(B)$ such that $\rho_{AB}\ll \sigma_A\otimes \tau_B$
\begin{align}
\liminf_{n\rightarrow\infty}-\frac{1}{n}\log (1- \hat{\alpha}_{n,\rho}^{\mathrm{iid}}(e^{-nR}))
\geq \sup_{s\in (1,\infty)}\frac{s-1}{s}(R-\widetilde{D}_s(\rho_{AB}\| \sigma_A\otimes \tau_B) ).
\end{align}
By taking the supremum over all such states, it follows that~\eqref{eq:sc-optimality} holds.

Case 2: $\rho_{AB}= \rho_A\otimes \rho_B$. Then 
$\hat{\alpha}_{n,\rho}^{\mathrm{iid}}(\mu)=1-\mu$ for all $\mu\in [0,1]$. 
Hence,
\begin{align}
\liminf_{n\rightarrow\infty}-\frac{1}{n}\log (1- \hat{\alpha}_{n,\rho}^{\mathrm{iid}}(e^{-nR}))
=R.
\end{align}
By \cite[Theorem~4~(p)]{burri2025prmisrmi1}, 
$\widetilde{I}_s^{\downarrow\downarrow}(A:B)_\rho=0$ for all $s\in (1,\infty)$. Hence, 
\begin{align}
\sup_{s\in (1,\infty)}\frac{s-1}{s}(R-\widetilde{I}_s^{\downarrow\downarrow}(A:B)_\rho )
=\sup_{s\in (1,\infty)}\frac{s-1}{s}R
=R.
\end{align}
\end{proof}

\subsection{Example for Remark~\ref{rem:permutation_srmi}}\label{app:example_srmi}

Suppose $d_A\geq 2,d_B\geq 2$. 
Let $\rho_{AB}\in \mathcal{S}(AB)$ be separable and such that $\rho_{AB}\neq \rho_A\otimes \rho_B$ and 
$I(A:B)_\rho\neq \widetilde{I}_\infty^{\downarrow\downarrow}(A:B)_\rho$. 

Consider now the left-hand side of~\eqref{eq:strong-converse} with $\hat{\alpha}_{n,\rho}$ replaced by $\hat{\alpha}^{\mathrm{ind}}_{n,\rho}$.
Since $\rho_{AB}$ is separable with respect to $A$ and $B$, also $\rho_{AB}^{\otimes n}$ is separable with respect to $A^n$ and $B^n$ for any $n\in \mathbb{N}_{>0}$.
Thus, $\hat{\alpha}^{\mathrm{ind}}_{n,\rho}(\mu)= 1-\mu$ for all $\mu\in [0,1]$, see Appendix~\ref{app:example_prmi}. 
This implies that for any $R\in [0,\infty)$
\begin{equation}
\lim_{n\rightarrow\infty}-\frac{1}{n}\log (1-\hat{\alpha}_{n,\rho}^{\mathrm{ind}}(e^{-nR}))=R.
\end{equation}
Consider now the right-hand side of~\eqref{eq:strong-converse}.
By Theorem~\ref{thm:strong-converse}, for any $R\in [I(A:B)_\rho,\infty)$
\begin{align}\label{eq:ex_rem}
\sup_{s\in (1,\infty)}\frac{s-1}{s}(R-\widetilde{I}_s^{\downarrow\downarrow}(A:B)_\rho)
&\leq R-I(A:B)_\rho
< R.
\end{align} 
The strict inequality follows from $\rho_{AB}\neq \rho_A\otimes\rho_B$. 
Therefore, the equality in~\eqref{eq:strong-converse} is violated if $\hat{\alpha}_{n,\rho}$ is replaced by 
$\hat{\alpha}^{\mathrm{ind}}_{n,\rho}$.
Thus, Theorem~\ref{thm:strong-converse} does not hold if $\hat{\alpha}_{n,\rho}$ is replaced by 
$\hat{\alpha}^{\mathrm{ind}}_{n,\rho}$.

\subsection{Proof of Corollary~\ref{cor:limit-type1_strongconverse}}\label{app:cor_strongconverse}
\begin{cor:limit-type1_strongconverse_repeated}[Asymptotic minimum type-I error]
Let $\rho_{AB}\in \mathcal{S}(AB)$ and let $R\in (I(A:B)_\rho,\infty)$. Then  
$\lim_{n\rightarrow\infty}\hat{\alpha}_{n,\rho}(e^{-nR})=1$.
Moreover, the same is true if $\hat{\alpha}_{n,\rho}$ is replaced by 
$\hat{\alpha}_{n,\rho}^{\mathrm{iid}}$.
\end{cor:limit-type1_strongconverse_repeated}

\begin{proof}
Let $\rho_{AB}\in \mathcal{S}(AB)$ and let $R\in (I(A:B)_\rho,\infty)$. Then
\begin{equation}\label{eq:app-cor_strongconverse}
1\geq \limsup_{n\rightarrow\infty}\hat{\alpha}_{n,\rho}^{\mathrm{iid}}(e^{-nR})
\geq \liminf_{n\rightarrow\infty}\hat{\alpha}_{n,\rho}^{\mathrm{iid}}(e^{-nR}).
\end{equation}
By the proof of optimality for Theorem~\ref{thm:strong-converse}, see~\eqref{eq:sc-optimality},
\begin{equation}\label{eq:proof-cor-liminf_strongconverse}
\liminf_{n\rightarrow\infty}-\frac{1}{n}\log (1-\hat{\alpha}_{n,\rho}^{\mathrm{iid}}(e^{-nR}))
\geq \sup_{s\in (1,\infty)}\frac{s-1}{s}(R-\widetilde{I}_s^{\downarrow\downarrow}(A:B)_\rho )>0,
\end{equation}
where the strict inequality follows from Theorem~\ref{thm:strong-converse} because $R>I(A:B)_\rho$. 
\eqref{eq:proof-cor-liminf_strongconverse} implies that 
$\liminf_{n\rightarrow\infty}\hat{\alpha}_{n,\rho}^{\mathrm{iid}}(e^{-nR})=1$.
By~\eqref{eq:app-cor_strongconverse}, this implies that 
$\lim_{n\rightarrow\infty}\hat{\alpha}_{n,\rho}^{\mathrm{iid}}(e^{-nR})=1$. 

The assertion regarding $\hat{\alpha}_{n,\rho}$ follows from this 
because 
$\hat{\alpha}_{n,\rho}^{\mathrm{iid}}(\mu)\leq \hat{\alpha}_{n,\rho}(\mu)\leq 1$ for all $\mu\in [0,\infty)$ and $n\in \mathbb{N}_{>0}$, see Lemma~\ref{lem:order}.
\end{proof}

\subsection{Proof of Corollary~\ref{cor:strong-converse}}\label{app:proof-cor_strong-converse}
\begin{lem}\label{lem:constant_srmi}
Let $\rho_{AB}\in \mathcal{S}(AB)$ and let $R$ be defined as in~\eqref{eq:def-R(s)-strongconverse}. 
Let $a,b\in [1,\infty]$ be such that $a<b$. 
Then, $(a,b)\rightarrow [0,\infty),s\mapsto \widetilde{I}_s^{\downarrow\downarrow}(A:B)_\rho$ is constant iff 
$(a,b)\rightarrow [0,\infty),s\mapsto R(s)$ is constant. 
\end{lem}
\begin{proof}
Suppose $(a,b)\rightarrow [0,\infty),s\mapsto \widetilde{I}_s^{\downarrow\downarrow}(A:B)_\rho$ is constant, 
i.e., $\exists c\in \mathbb{R}$ such that 
$\widetilde{I}_s^{\downarrow\downarrow}(A:B)_\rho=c$ for all $s\in (a,b)$. 
Then, $\frac{\mathrm{d}}{\mathrm{d}s}\widetilde{I}_s^{\downarrow\downarrow}(A:B)_\rho=0$ for all $s\in (a,b)$, so $R(s)=\widetilde{I}_s^{\downarrow\downarrow}(A:B)_\rho=c$ for all $s\in (a,b)$. 

Now, suppose $(a,b)\rightarrow [0,\infty),s\mapsto R(s)$ is constant instead, 
i.e., $\exists c\in \mathbb{R}$ such that 
$\widetilde{I}_s^{\downarrow\downarrow}(A:B)_\rho+s(s-1)\frac{\mathrm{d}}{\mathrm{d}s}\widetilde{I}_s^{\downarrow\downarrow}(A:B)_\rho=c$ for all $s\in (a,b)$. 
The solution of this first-order linear ordinary differential equation is 
$\widetilde{I}_s^{\downarrow\downarrow}(A:B)_\rho=\frac{c+ks}{1-s}$ where $k\in \mathbb{R}$ is a constant. 
We have
\begin{align}
0&=\lim_{s\rightarrow 1^+}(1-s)\widetilde{I}_s^{\downarrow\downarrow}(A:B)_\rho 
=\lim_{s\rightarrow 1^+}(c+ks)
=c+k.
\end{align}
Hence, $k=-c$, which implies that $\widetilde{I}_s^{\downarrow\downarrow}(A:B)_\rho =c$ for all $s\in (a,b)$.
\end{proof}

\begin{cor:strong-converse_repeated}[Strong converse exponent]
Let $\rho_{AB}\in \mathcal{S}(AB)$ and let
\begin{align}\label{eq:def-R(s)-strongconverse}
R:(1,\infty)\rightarrow 
[I(A:B)_\rho,\infty),\quad  
s\mapsto 
\widetilde{I}_s^{\downarrow\downarrow}(A:B)_\rho+s(s-1)\frac{\mathrm{d}}{\mathrm{d} s}\widetilde{I}_s^{\downarrow\downarrow}(A:B)_\rho.
\end{align}
Then $R$ is continuous and monotonically increasing. 
Let $R(1)\coloneqq \lim_{s\rightarrow 1^+}R(s)=I(A:B)_\rho$, 
$R_\infty\coloneqq \lim_{s\rightarrow\infty}R(s)\in [\widetilde{I}_\infty^{\downarrow\downarrow}(A:B)_\rho,\infty]$, and 
$R(\infty)\coloneqq R_\infty$. 
Let $s_1\coloneqq \max\{s\in [1,\infty]:R(s)=I(A:B)_\rho\}$ 
and
$s_\infty \coloneqq\min\{s\in [1,\infty]:R(s)=R_\infty\}$. 

Suppose $I(A:B)_\rho\neq \widetilde{I}_\infty^{\downarrow\downarrow}(A:B)_\rho$. 
Then, $1\leq s_1<s_\infty = \infty$ and for any $s\in (s_1,s_\infty)$
\begin{equation}\label{eq:sc1}
\lim_{n\rightarrow\infty} -\frac{1}{n}\log (1-\hat{\alpha}_{n,\rho}(e^{-nR(s)}))
=\frac{s-1}{s}(R(s)-\widetilde{I}_s^{\downarrow\downarrow}(A:B)_\rho )
=(s-1)^2\frac{\mathrm{d}}{\mathrm{d}s} \widetilde{I}_s^{\downarrow\downarrow}(A:B)_\rho,
\end{equation}
and if $R_\infty<\infty$, then for any $R'\in [R_\infty,\infty)$
\begin{align}\label{eq:sc2}
\lim_{n\rightarrow\infty} -\frac{1}{n}\log (1-\hat{\alpha}_{n,\rho}(e^{-nR'}))
=R'-\widetilde{I}_\infty^{\downarrow\downarrow}(A:B)_\rho.
\end{align}
Moreover, the same is true if $\hat{\alpha}_{n,\rho}$ in~\eqref{eq:sc1} and~\eqref{eq:sc2} is replaced by 
$\hat{\alpha}_{n,\rho}^{\mathrm{iid}}$.
\end{cor:strong-converse_repeated}

\begin{proof}
Let $\rho_{AB}\in \mathcal{S}(AB)$. 
The proof of Theorem~\ref{thm:strong-converse} then implies that 
$R$ is continuous and monotonically increasing. 

From now on, suppose 
$I(A:B)_\rho\neq \widetilde{I}_{\infty}^{\downarrow\downarrow}(A:B)_\rho$. 
Then, Lemma~\ref{lem:constant_srmi} implies that $R(1)\neq R(\infty)$. 
Hence, $1\leq s_{1}<s_\infty\leq \infty$. 
The proof of Theorem~\ref{thm:strong-converse} then implies that~\eqref{eq:sc1} holds for any $s\in (s_1,s_\infty)$, and that the same is true if $\hat{\alpha}_{n,\rho}$ in~\eqref{eq:sc1} is replaced by 
$\hat{\alpha}_{n,\rho}^{\mathrm{iid}}$.

It remains to show that $s_\infty= \infty$. We will prove this by contradiction. 
Suppose $s_{\infty}<\infty$. 
Then, $R(s)=R_{\infty}$ for all $s\in [s_\infty, \infty]$. 
By Lemma~\ref{lem:constant_srmi}, it follows that
\begin{align}\label{eq:proof_const1_srmi}
\widetilde{I}_s^{\downarrow\downarrow}(A:B)_\rho=\widetilde{I}_{\infty}^{\downarrow\downarrow}(A:B)_\rho
\qquad \forall s\in [s_{\infty},\infty].
\end{align} 
Since $I(A:B)_\rho\neq \widetilde{I}_{\infty}^{\downarrow\downarrow}(A:B)_\rho$, 
it follows from \cite[Theorem~4~(n)]{burri2025prmisrmi1} that there exists $s_*\in (s_1,s_\infty)$ such that 
$\frac{\mathrm{d}}{\mathrm{d}s}\widetilde{I}_s^{\downarrow\downarrow}(A:B)_\rho|_{s=s_*}\neq 0$. 
Then,
\begin{align}
0
&=\lim_{s\rightarrow s_\infty}\frac{\mathrm{d}}{\mathrm{d}s} \widetilde{I}_s^{\downarrow\downarrow}(A:B)_\rho
\label{eq:proof_const2_srmi}\\
&=\lim_{s\rightarrow s_\infty^-} (s-1)^2 \frac{\mathrm{d}}{\mathrm{d}s} \widetilde{I}_s^{\downarrow\downarrow}(A:B)_\rho
\\
&=\lim_{s\rightarrow s_\infty^-}\lim_{n\rightarrow\infty}-\frac{1}{n}\log (1- \hat{\alpha}_{n,\rho}(e^{-nR(s)}))
\label{eq:proof_const3_srmi}\\
&\geq \lim_{n\rightarrow\infty}-\frac{1}{n}\log (1- \hat{\alpha}_{n,\rho}(e^{-nR(s_*)}))
\label{eq:proof_const4_srmi}
\\
&=(s_*-1)^2 \frac{\mathrm{d}}{\mathrm{d}s} \widetilde{I}_s^{\downarrow\downarrow}(A:B)_\rho \big|_{s=s^*}
>0.
\label{eq:proof_const5_srmi}
\end{align}
\eqref{eq:proof_const2_srmi} follows from~\eqref{eq:proof_const1_srmi} and \cite[Theorem~4~(n)]{burri2025prmisrmi1}. 
The equalities in~\eqref{eq:proof_const3_srmi} and~\eqref{eq:proof_const5_srmi} follow from above because~\eqref{eq:sc1} holds for any $s\in (s_1,s_\infty)$. 
\eqref{eq:proof_const4_srmi} holds because $R$ is monotonically increasing, $\hat{\alpha}_{n,\rho}$ is monotonically decreasing (see Lemma~\ref{lem:order}), and $s_\infty\geq s_*$. 
Since~\eqref{eq:proof_const5_srmi} yields a contradiction, we can conclude that $s_{\infty}=\infty$.
\end{proof}

\subsection{Proof for Remark~\ref{rem:strong-converse}}\label{app:rem_strong-converse}
\begin{proof}
Let $\rho_{AB}\in \mathcal{S}(AB)$ be such that $V(A:B)_\rho\neq 0$. 
By \cite[Theorem~4~(l), (n)]{burri2025prmisrmi1} it follows that there exists $\alpha_0\in (1,\infty)$ such that 
$I_\alpha^{\downarrow\downarrow}(A:B)_\rho$ is strictly monotonically increasing for $\alpha\in [1,\alpha_0]$. 
Let $R$ be defined as in~\eqref{eq:def-R(s)-strongconverse} and let $R(1)\coloneqq I(A:B)_\rho$. 
By Lemma~\ref{lem:constant_srmi} and Corollary~\ref{cor:strong-converse}, it follows that $R(\alpha)$ is strictly monotonically increasing for $\alpha\in [1,\alpha_0]$. 
Therefore, the parameter $s_1$ as defined Corollary~\ref{cor:strong-converse} is given by $s_1=1$. 
The assertion in Remark~\ref{rem:strong-converse} then follows from Corollary~\ref{cor:strong-converse}.
\end{proof}

\subsection{Proof of Corollary~\ref{cor:reverse-cutoff}}\label{ssec:reverse-cutoff}
For simplicity, we will prove the assertion for $\hat{\alpha}_{n,\rho}$ only. 
The proof for $\hat{\alpha}_{n,\rho}^{\mathrm{iid}}$ is analogous due to Corollary~\ref{cor:strong-converse}.  
The proof is divided into two parts: 
a proof of achievability and a proof of optimality. 

\subsubsection{Proof of achievability}
Let $\rho_{AB}\in \mathcal{S}(AB)$. 
In the following, we will show that 
\begin{align}\label{eq:fc-geq-reverse}
R_0^{(r)}(\beta)_\rho\leq \widetilde{I}_{\frac{1}{1-\beta}}^{\downarrow\downarrow}(A:B)_\rho
\quad \forall\beta\in (0,1).
\end{align}
\begin{proof}
Let $\beta\in (0,1)$ be arbitrary but fixed. 
Let $R_0\coloneqq \widetilde{I}_{\frac{1}{1-\beta}}^{\downarrow\downarrow}(A:B)_\rho$. 
Then, for all $R\in (0,\infty)$
\begin{subequations}\label{eq:fc12-reverse}
\begin{align}
\liminf_{n\rightarrow\infty}-\frac{1}{n}\log (1-\hat{\alpha}_{n,\rho}(e^{-nR}))
&\geq \sup_{s\in (1,\infty)}\frac{s-1}{s}(R-\widetilde{I}_s^{\downarrow\downarrow}(A:B)_\rho)
\label{eq:fc1-reverse}\\
&=\sup_{b\in (0,1)}b(R-\widetilde{I}_{\frac{1}{1-b}}^{\downarrow\downarrow}(A:B)_\rho )
\geq \beta (R-R_0).
\label{eq:fc2-reverse}
\end{align}
\end{subequations}
\eqref{eq:fc1-reverse} follows from the proof of optimality for Theorem~\ref{thm:strong-converse}, see~\eqref{eq:sc-optimality}, and Lemma~\ref{lem:order}. 
\eqref{eq:fc2-reverse} follows from~\eqref{eq:fc1-reverse} by defining $b\coloneqq \frac{s-1}{s}$. 
\eqref{eq:fc12-reverse} implies that 
$R_0^{(r)}(\beta)_\rho\leq R_0
=\widetilde{I}_{\frac{1}{1-\beta}}^{\downarrow\downarrow}(A:B)_\rho$. 
\end{proof}

\subsubsection{Proof of optimality}
Let $\rho_{AB}\in \mathcal{S}(AB)$ be such that $I(A:B)_\rho\neq \widetilde{I}_\infty^{\downarrow\downarrow}(A:B)_\rho$. 
Then, $0\leq \beta_1 <\beta_\infty =1$ follows from Corollary~\ref{cor:strong-converse}. 
In the following, we will show that 
\begin{align}
R_0^{(r)}(\beta)_\rho\geq \widetilde{I}_{\frac{1}{1-\beta}}^{\downarrow\downarrow}(A:B)_\rho \quad\forall \beta\in (\beta_{1},\beta_\infty).
\end{align}
\begin{proof}
Since $I(A:B)_\rho\neq \widetilde{I}_\infty^{\downarrow\downarrow}(A:B)_\rho$, \cite[Theorem~4]{burri2025prmisrmi1} implies that $\rho_{AB}$ is not a product state and $0<I(A:B)_\rho\leq \widetilde{I}_{\frac{1}{1-\beta}}^{\downarrow\downarrow}(A:B)_\rho \leq \widetilde{I}_\infty^{\downarrow\downarrow}(A:B)_\rho$ for all $\beta\in (0,1)$.

Let $\beta\in (\beta_{1},\beta_\infty)$ be arbitrary but fixed. 
Let $s\coloneqq \frac{1}{1-\beta}\in (s_{1},s_\infty)$, and let 
$R(s)\coloneqq \widetilde{I}_s^{\downarrow\downarrow}(A:B)_\rho+s(s-1)\frac{\mathrm{d}}{\mathrm{d} s}\widetilde{I}_s^{\downarrow\downarrow}(A:B)_\rho$. 
Then, $R(s)\in (I(A:B)_\rho,\infty)\subseteq (0,\infty)$ follows from Corollary~\ref{cor:strong-converse}. 
For all $R_0\in (0,\widetilde{I}_{\frac{1}{1-\beta}}^{\downarrow\downarrow}(A:B)_\rho)$
\begin{subequations}\label{eq:fc34-reverse}
\begin{align}
\liminf_{n\rightarrow\infty} -\frac{1}{n}\log (1- \hat{\alpha}_{n,\rho}(e^{-nR(s)}))
&=\frac{s-1}{s}(R(s)-\widetilde{I}_s^{\downarrow\downarrow}(A:B)_\rho )
\label{eq:fc3-reverse}\\
&=\beta (R(s)-\widetilde{I}_{\frac{1}{1-\beta}}^{\downarrow\downarrow}(A:B)_\rho )
<\beta (R(s)-R_0 ).
\label{eq:fc4-reverse}
\end{align}
\end{subequations}
\eqref{eq:fc3-reverse} follows from Corollary~\ref{cor:strong-converse}. 
\eqref{eq:fc34-reverse} implies that 
$R_0^{(r)}(\beta)_\rho\geq \widetilde{I}_{\frac{1}{1-\beta}}^{\downarrow\downarrow}(A:B)_\rho$. 
\end{proof}

\subsection{Proof of Theorem~\ref{thm:moderate_strongconverse}}\label{proof:moderate_strongconverse}

The proof of Theorem~\ref{thm:moderate_strongconverse} is divided into two parts: 
a proof of achievability for $\hat{\alpha}_{n,\rho}$ and 
a proof of optimality for $\hat{\alpha}_{n,\rho}^{\mathrm{iid}}$. 
The assertion in Theorem~\ref{thm:moderate_strongconverse} follows from these two parts because 
$\hat{\alpha}_{n,\rho}^{\mathrm{iid}}(\mu)\leq \hat{\alpha}_{n,\rho}(\mu)$ for all $\mu\in [0,\infty)$, see Lemma~\ref{lem:order}.

\subsubsection*{Proof of achievability}

We will show that 
\begin{align}
\limsup_{n\rightarrow\infty}-\frac{1}{na_n^2}\log (1-\hat{\alpha}_{n,\rho}(e^{-nR_n}))
\leq \frac{1}{2V(A:B)_\rho}.
\end{align}

\begin{proof}
Let $I\coloneqq I(A:B)_\rho$ and $V\coloneqq V(A:B)_\rho$. 

Let $n \in \mathbb{N}_{>0}$ be fixed. 
In the following, we assume that $n$ is sufficiently large whenever required. 
Let $P_n$ and $Q_n$ be defined as in \eqref{eq:test-p}--\eqref{eq:test-q}. 
Let $s_n\in (1,\infty)$ to be specified later. 
Then, by the same arguments as in~\eqref{eq:ach-at}--\eqref{eq:ach-beta} (with $s_n$ instead of $\hat{s}$), 
\begin{align}\label{eq:mod-ach-ln}
1-\hat{\alpha}_{n,\rho}(e^{-nR_n}) 
&\geq \mathrm{Pr}_{x_n\sim P_n}[\mu_n\leq L_n(x_n)],
\end{align}
where $L_n(x_n)\coloneqq \log P_n(x_n) -\log Q_n(x_n)$ for $x_n\in \mathcal{X}_n\coloneqq \{y_n\in [d_A^nd_B^n]:P_n(y_n)\neq 0\}$,  
\begin{align}\label{eq:mod-mun-def}
\mu_n\coloneqq \frac{1}{s_n}(nR_n+\phi_n(s_n)+\log g_{n,d_A}+\log g_{n,d_B}), 
\end{align}
and 
$\phi_n(s)\coloneqq (s-1)D_s(P_n\|Q_n)$ for $s\in (1,\infty)$. 
Moreover, let us define the exponentially tilted probability distributions 
\begin{align}
P_{n,s}(x_n)
&\coloneqq e^{(s-1)L_n(x_n)-\phi_n(s)}P_n(x_n)
=\frac{P_n(x_n)^sQ_n(x_n)^{1-s}}{\sum\limits_{y_n\in [d_A^nd_B^n]}P_n(y_n)^{s}Q_n(y_n)^{1-s}}
\label{eq:pns-pn}
\end{align}
for $s\in (1,\infty)$. 

Let $\Delta_n\in (0,\infty)$ to be specified later. 
For all $s\in (1,\infty)$ 
\begin{align}
1-\hat{\alpha}_{n,\rho}(e^{-nR_n})
&\geq \mathrm{Pr}_{x_n\sim P_n}[\mu_n\leq L_n(x_n)< \mu_n+\Delta_n]
\label{eq:l-mu0}\\
&=\sum_{\substack{x_n\in \mathcal{X}_n:\\ \mu_n\leq L_n(x_n)<\mu_n+\Delta_n}}
P_n(x_n)
\\
&=\sum_{\substack{x_n\in \mathcal{X}_n:\\ \mu_n\leq L_n(x_n)<\mu_n+\Delta_n}}
P_{n,s}(x_n) e^{-(s-1)L_n(x_n)+\phi_n(s)}
\label{eq:l-mu2}\\
&\geq \sum_{\substack{x_n\in \mathcal{X}_n:\\ \mu_n\leq L_n(x_n)<\mu_n+\Delta_n}}
P_{n,s}(x_n) e^{-(s-1)(\mu_n+\Delta_n)+\phi_n(s)}
\\
&=e^{-(s-1)(\mu_n+\Delta_n)+\phi_n(s)} 
\mathrm{Pr}_{x_n\sim P_{n,s}}[\mu_n\leq L_n(x_n)<\mu_n+\Delta_n].
\label{eq:l-mu4}
\end{align}
\eqref{eq:l-mu0} follows from~\eqref{eq:mod-ach-ln}.
\eqref{eq:l-mu2} follows from~\eqref{eq:pns-pn}.  
Let us denote the probability on the right-hand side of~\eqref{eq:l-mu4} as 
\begin{align}\label{eq:mod-def-pns}
\mathcal{P}_{n,s}\coloneqq \mathrm{Pr}_{x_n\sim P_{n,s}}[\mu_n\leq L_n(x_n)<\mu_n+\Delta_n].
\end{align}

We now refine our choices of $s_n$ and $\Delta_n$. 
Let $s_n\coloneqq 1+c_n\frac{a_n}{V}$ where $c_n\coloneqq 1+\delta_n$, $\delta_n\in (0,1)$ and $\delta_n\rightarrow 0$ as $n\rightarrow\infty$. 
(The choice of $\delta_n$ will be specified later.) 
Let $\varepsilon\in (0,\infty)$ be arbitrary but fixed 
and let $\Delta_n\coloneqq 2\varepsilon na_n$. 
Then, 
\begin{align}
&-\frac{1}{n a_n^2}\log\bigl(1-\hat{\alpha}_{n,\rho}(e^{-nR_n})\bigr)
\nonumber\\
&\leq -\frac{1}{na_n^2} (-(s_n-1)\mu_n+\phi_n(s_n)-(s_n-1)\Delta_n +\log \mathcal{P}_{n,s_n})
\label{eq:log1-1}\\
&=\frac{c_n}{V s_n a_n}
\left(
R_n + \frac{\log g_{n,d_A}}{n} + \frac{\log g_{n,d_B}}{n}
- \frac{1}{n}D_{1 + c_n a_n / V}(P_n \| Q_n)
\right)
+\frac{c_n \Delta_n}{n a_n V}
-\frac{1}{n a_n^2}\log \mathcal{P}_{n,s_n}
\label{eq:log1-2}\\
&=
\frac{c_n}{V s_n}
\left(
1 - \frac{c_n}{2}
\right)
+\frac{c_n}{V s_n a_n}
\left(
o(a_n) + O\left(\frac{\log n}{n}\right)
\right)
+\frac{c_n \Delta_n}{n a_n V}
-\frac{1}{n a_n^2}\log \mathcal{P}_{n,s_n}
\label{eq:log1-3}\\
&\leq
\frac{c_n}{2V}
+\frac{o(a_n)}{a_n}
+O\left(\frac{1}{\sqrt{n}a_n} \frac{\log n}{\sqrt{n}}\right) 
+\frac{2 \varepsilon c_n}{V}
-\frac{1}{n a_n^2}\log \mathcal{P}_{n,s_n}.
\label{eq:log1-4}
\end{align}
\eqref{eq:log1-1} follows from~\eqref{eq:l-mu4}. 
\eqref{eq:log1-2} follows from the definitions of $s_n$, $\mu_n$, and $\phi_n$. 
For~\eqref{eq:log1-3}, we have used the definition of $R_n$, 
$\log g_{n,d_A}=O(\log n)$~\cite{burri2025prmisrmi1}, 
and we have used that
\begin{align}\label{eq:dpq-i1}
\frac{1}{n}D_{s_n}(P_n\| Q_n)
&=\frac{1}{n}\widetilde{I}_{s_n}^{\downarrow\downarrow}(A:B)_{\rho} 
+O\left( \frac{\log n}{n}\right)
\\
&=I+\frac{V}{2}(s_n-1)+o(s_n-1)+O\left(\frac{\log n}{n}\right).
\label{eq:d-anv4}
\end{align}
\eqref{eq:dpq-i1} holds because for $\alpha\in [1,\infty)$
\begin{align}
\frac{1}{n}D_\alpha (P_n\| Q_n)
&=\frac{1}{n}\widetilde{D}_\alpha (\mathcal{P}_{\omega_{A^n}^n\otimes \omega_{B^n}^n}(\rho_{AB}^{\otimes n})\|\omega_{A^n}^n\otimes \omega_{B^n}^n )
=\widetilde{I}_\alpha^{\downarrow\downarrow}(A:B)_\rho +O\left(\frac{\log n}{n}\right)
\label{eq:dpq-i}
\end{align}
as $n\rightarrow\infty$, 
where the last term is independent of $\alpha$~\cite[Proof of Theorem~4~(j)]{burri2025prmisrmi1}. 
\eqref{eq:d-anv4} holds because 
$\widetilde{I}_\alpha^{\downarrow\downarrow}(A:B)_\rho$ is continuously differentiable on $\alpha\in [1,\infty)$~\cite[Theorem~4~(n)]{burri2025prmisrmi1}, which implies for $\alpha\in [1,\infty)$
\begin{align}\label{eq:i-v2}
\widetilde{I}_\alpha^{\downarrow\downarrow}(A:B)_\rho=I+\frac{V}{2}(\alpha-1)+o(\alpha-1).
\end{align}
\eqref{eq:log1-4} follows from $s_n\geq 1$, $c_n\in (1,2)$, $c_n\rightarrow 1$ and $s_n\rightarrow 1$ as $n\rightarrow\infty$, and the identity $\Delta_n=2\varepsilon na_n$.

We proceed by analyzing the individual terms appearing in~\eqref{eq:log1-4}. 
The second and the third term in~\eqref{eq:log1-4} vanish as $n\rightarrow\infty$. 
In the remainder of this proof, we will show that also the last term in~\eqref{eq:log1-4} is negligible as $n\rightarrow\infty$, i.e., 
\begin{align}\label{eq:limsup_p_n}
\liminf_{n\rightarrow\infty}\frac{1}{na_n^2}\log \mathcal{P}_{n,s_n} \stackrel{?}{\geq} 0.
\end{align}
The claim then follows from~\eqref{eq:log1-4} since $\varepsilon\in (0,\infty)$ can be chosen arbitrarily small and $c_n\rightarrow 1$ as $n\rightarrow\infty$. 
Thus, it remains to prove~\eqref{eq:limsup_p_n}. 

We proceed to derive an estimate for $\mu_n$. 
Using~\eqref{eq:dpq-i} and~\eqref{eq:i-v2}, we have for $x\in [0,\infty)$
\begin{align}
\phi_n(1+x)
=nIx+n\frac{V}{2}x^2+xO(\log n)+nr(x)x^2,
\label{eq:mod-phi-expansion}
\end{align}
where $O(\log n)$ is independent of $x$, 
$r$ is a function such that $r(x)\rightarrow 0$ as $x\rightarrow 0^+$, and $r$ is independent of $n$. 
Let $b_n\coloneqq s_n-1=c_na_n/V>0$. 
Using~\eqref{eq:mod-phi-expansion} and the definition of $\mu_n$ in~\eqref{eq:mod-mun-def}, we obtain
\begin{align}
\mu_n
&=(1-b_n+\underbrace{O(b_n^2)}_{=O(a_n^2)})
(nI+na_n+nIb_n+n\frac{V}{2}b_n^2
+\underbrace{nr(b_n)b_n^2}_{=o(nb_n^2)=o(na_n^2)}+O(\log n))
\\
&=nI + n a_n + O(n a_n^2) + O(\log n)
=nI+na_n+o(na_n).
\label{eq:mod_mu}
\end{align}

Next, we derive an estimate for $m_n\coloneqq \mathbb{E}_{x_n\sim P_{n,s_n}}[L_n(x_n)] $. 
We have 
$\mathbb{E}_{x_n\sim P_{n,s}}[L_n(x_n)] = \phi_n'(s)$ for any $s\in (1,\infty)$. 
In particular, 
$m_n = \phi_n'(s_n)=\phi_n'(1+b_n)$. 
To derive an estimate for $m_n$, we now analyze this derivative. 

Let 
$\varepsilon_n \coloneqq \sup_{u\in (0,3b_n)}|r(u)|$, 
$\theta_n \coloneqq \max\{\varepsilon_n^{1/4},b_n\}$, 
$h_n \coloneqq b_n\theta_n>0$. 
Since $b_n\rightarrow 0$ as $n\rightarrow\infty$ and $r(x)\rightarrow 0$ as $x\rightarrow 0^+$, 
we have $\varepsilon_n\rightarrow 0$ and $\theta_n\rightarrow 0$ as $n\rightarrow\infty$. 
Therefore, 
$h_n=o(b_n)$. 
In particular, we have $b_n-h_n>0$ and $b_n+h_n<2b_n$ for all sufficiently large $n$. 
Since $\phi_n$ is convex and differentiable, its derivative is bounded from below and above by the corresponding left and right difference quotients. Therefore,
\begin{align}
\frac{\phi_n(1+b_n)-\phi_n(1+b_n-h_n)}{h_n}
\leq \phi_n'(1+b_n)
\le
\frac{\phi_n(1+b_n+h_n)-\phi_n(1+b_n)}{h_n}.
\label{eq:mod-convexity-mn}
\end{align}
We now insert~\eqref{eq:mod-phi-expansion} into the two quotients in~\eqref{eq:mod-convexity-mn}. 
The contribution of the first term in~\eqref{eq:mod-phi-expansion} is $nI$, 
and the contribution of the third term in~\eqref{eq:mod-phi-expansion} is $O(\log n)$, because these terms are linear in $x$. 
Note that $O(\log n)=o(na_n)$ because $\sqrt{n}a_n\rightarrow\infty$ and $(\log n)/\sqrt{n}\rightarrow 0$ as $n\rightarrow\infty$. 
The contribution of the second term in~\eqref{eq:mod-phi-expansion} is
\begin{align}
\pm \frac{n\frac{V}{2}\big((b_n\pm h_n)^2-b_n^2\big)}{h_n}
= nVb_n + O(nh_n)
=nVb_n+o(nb_n)
=c_n na_n+o(na_n).
\end{align}
For the contribution of the last term in~\eqref{eq:mod-phi-expansion}, we use that $b_n\pm h_n\in (0,2b_n)$ for all sufficiently large $n$. 
Thus,
\begin{align}
&\left|
\frac{nr(b_n\pm h_n)(b_n\pm h_n)^2-nr(b_n)b_n^2}{h_n}
\right|
\nonumber\\
&\leq n\varepsilon_n \frac{(b_n+h_n)^2+b_n^2}{h_n}
\leq \text{const. } n\varepsilon_n \frac{b_n^2}{h_n}
\leq \text{const. } n b_n \varepsilon_n^{3/4}
=o(nb_n)
=o(na_n).
\end{align}
The second inequality holds because $h_n=o(b_n)$. 
Combining the contributions of the four terms yields an estimate for $m_n$ given by
\begin{align}
m_n 
&=\phi_n'(s_n)
=\phi_n'(1+b_n)
= nI + c_n n a_n + o(n a_n).
\label{eq:mod_mn}
\end{align}

Let us define 
$Z_n(x_n) \coloneqq \frac{1}{n a_n}(L_n(x_n) - m_n )$ for $x_n\in \mathcal{X}_n$ and 
let $z_n \coloneqq (\mu_n - m_n)/(n a_n)$. 
Then, 
\begin{align}
\mathcal{P}_{n,s_n}
\stackrel{\eqref{eq:mod-def-pns}}{=}\mathrm{Pr}_{x_n\sim P_{n,s_n}}[z_n \leq Z_n(x_n) < z_n + 2 \varepsilon ].
\label{eq:mod_pns_2e}
\end{align}
We now show that $z_n$, which is the left endpoint in~\eqref{eq:mod_pns_2e}, tends to $0$. 
By~\eqref{eq:mod_mu}, we have 
$\mu_n = nI + n a_n + o(n a_n)$. 
By~\eqref{eq:mod_mn}, we have
$m_n = nI + c_n n a_n + o(n a_n)$. 
Therefore,
$z_n=1-c_n+o(1)=-\delta_n +o(1)$. 
We now specify the choice of $\delta_n$. 
Since $z_n=-\delta_n +o(1)$ and $\delta_n\in (0,1)$, we can choose $\delta_n\rightarrow 0$ sufficiently slowly such that 
$z_n<0$ for all sufficiently large $n$.  
(For instance, one can choose $\delta_n\coloneqq \sqrt{o(1)}$. 
Then, $\delta_n\rightarrow 0$ and $o(1)/\delta_n\rightarrow 0$ as $n\rightarrow\infty$, hence $z_n<0$ for all sufficiently large $n$.) 
Then, for all sufficiently large $n$, we have $-\varepsilon < z_n<0$, so that 
$(0,\varepsilon)\subseteq [z_n,z_n+2\varepsilon)$. 
Combining this with~\eqref{eq:mod_pns_2e}, we obtain 
$\mathcal{P}_{n,s_n} \geq \mathrm{Pr}_{x_n\sim P_{n,s_n}}[ Z_n(x_n) \in (0,\varepsilon)]$ 
for all sufficiently large $n$.
To prove~\eqref{eq:limsup_p_n}, it therefore suffices to prove that
\begin{align}
\liminf_{n \to \infty}
\frac{1}{n a_n^2}
\log \mathrm{Pr}_{x_n\sim P_{n,s_n}}[Z_n(x_n) \in (0,\varepsilon)]
\stackrel{?}{\geq} 0.
\label{eq:mod_toprove_e}
\end{align}

To this end, we determine the asymptotic cumulant generating function of $Z_n$ under $P_{n,s_n}$ with speed $na_n^2$.
For any $\theta \in \mathbb{R}$, let
\begin{align}
\Lambda_n(\theta)
&\coloneqq
\frac{1}{n a_n^2}
\log
\mathbb{E}_{x_n\sim P_{n,s_n}}
[ e^{\theta n a_n^2 Z_n(x_n)} ]
\\
&=
\frac{1}{n a_n^2}
\log
\mathbb{E}_{x_n\sim P_{n,s_n}}
[ e^{\theta a_n L_n(x_n) - \theta a_n m_n} ]
\label{eq:mod-lambdan2}\\
&=
\frac{1}{n a_n^2}
\left( \phi_n(s_n + \theta a_n)
-\phi_n(s_n)
-\theta a_n \phi_n'(s_n)
\right).
\label{eq:mod_lambdan_0}
\end{align}
\eqref{eq:mod-lambdan2} follows from the definition of $Z_n$. 
\eqref{eq:mod_lambdan_0} follows from~\eqref{eq:mod_mn} and from the definition of $P_{n,s_n}$ in~\eqref{eq:pns-pn}. 
Next, we estimate the three terms occurring in~\eqref{eq:mod_lambdan_0} separately. 
For the last term in~\eqref{eq:mod_lambdan_0}, \eqref{eq:mod_mn} implies
\begin{align}
\theta a_n \phi_n'(s_n) 
&= \theta a_n ( nI + c_n n a_n + o(n a_n))
\\
&=nI\theta a_n+\theta c_n na_n^2+o(na_n^2).
\label{eq:mod_lambdan_3}
\end{align}
For the first two terms in~\eqref{eq:mod_lambdan_0}, \eqref{eq:mod-phi-expansion} implies for any fixed $\theta\in \mathbb{R}$
\begin{align}
\phi_n(s_n + \theta a_n)
-\phi_n(s_n)
&=nI\theta a_n
+n\frac{V}{2}(2b_n\theta a_n+\theta^2a_n^2)
+\theta a_nO(\log n)
+o(na_n^2)
\\
&=nI \theta a_n
+\frac{V}{2}\left(\frac{2c_n\theta}{V}+\theta^2\right)na_n^2
+o(na_n^2),
\label{eq:mod_lambdan_2}
\end{align}
where we have used the identity $b_n=c_na_n/V$.  
Inserting~\eqref{eq:mod_lambdan_3} and~\eqref{eq:mod_lambdan_2} into~\eqref{eq:mod_lambdan_0} yields 
for any fixed $\theta \in \mathbb{R}$ 
\begin{align}
\Lambda_n(\theta)
&=\frac{V}{2}
\left( \frac{2c_n\theta}{V}+\theta^2 \right)
-\theta c_n
+o(1)
=\frac{V\theta^2}{2} 
+o(1).
\end{align}

Thus, the limiting cumulant generating function exists and is given by
$\Lambda(\theta) \coloneqq \lim_{n\to\infty}\Lambda_n(\theta) = V\theta^2/2$ 
for all $\theta \in \mathbb{R}$. 
Its Legendre-Fenchel transform is 
$\Lambda^*(z)
\coloneqq
\sup_{\theta \in \mathbb{R}} \{ \theta z - \Lambda(\theta) \}
=z^2/(2V)$ for any $z\in \mathbb{R}$. 
We now apply the G\"{a}rtner-Ellis lower bound from~\cite[Theorem II.6.1.]{Ellis1975entropy} 
to $Z_n$ under $P_{n,s_n}$ with speed $n a_n^2$. 
This bound is applicable because $\Lambda$ is finite and differentiable on $\mathbb{R}$. 
The G\"{a}rtner-Ellis lower bound yields
\begin{align}
\liminf_{n \to \infty}
\frac{1}{n a_n^2}
\log \mathrm{Pr}_{x_n\in P_{n,s_n}}[Z_n(x_n) \in (0,\varepsilon)]
\geq -\inf_{z \in (0,\varepsilon)}
\Lambda^*(z).
\end{align}
The right-hand side is given by
$-\inf_{z \in (0,\varepsilon)} \Lambda^*(z)
=-\inf_{z \in (0,\varepsilon)} z^2/(2V)
=0$. 
Thus, we obtain~\eqref{eq:mod_toprove_e}.
\end{proof}

\subsubsection*{Proof of optimality}

We will show that 
\begin{align}
\liminf_{n\rightarrow\infty}-\frac{1}{na_n^2}\log (1-\hat{\alpha}_{n,\rho}^{\mathrm{iid}}(e^{-nR_n}))
\geq \frac{1}{2V(A:B)_\rho}.
\end{align}

\begin{proof} 
Let us define the function $\hat{\alpha}_{n,\rho}^{\mathrm{mar}}(\mu )$ as in~\eqref{eq:def-alpha-mar}. 
Then, 
\begin{align}
\liminf_{n\rightarrow\infty}-\frac{1}{na_n^2}\log (1-\hat{\alpha}_{n,\rho}^{\mathrm{iid}}(e^{-nR_n}))
\geq \liminf_{n\rightarrow\infty}-\frac{1}{na_n^2}\log (1-\hat{\alpha}_{n,\rho}^{\mathrm{mar}}(e^{-nR_n}))
=\frac{1}{2V(A:B)_\rho}.
\end{align}
The equality follows from~\cite[Theorem~1]{Chubb_2017}.
\end{proof}

\section{Proof for Section~\ref{ssec:second}}
\subsection{Proof of Theorem~\ref{thm:second}}\label{app:second}
The proof of Theorem~\ref{thm:second} is divided into two parts: a proof of achievability for $\hat{\alpha}_{n,\rho}$ and a proof of optimality for $\hat{\alpha}_{n,\rho}^{\mathrm{iid}}$. 
The assertion in Theorem~\ref{thm:second} follows from these two parts because 
$\hat{\alpha}_{n,\rho}^{\mathrm{iid}}(\mu)\leq \hat{\alpha}_{n,\rho}(\mu)$ for all $\mu\in [0,\infty)$, see Lemma~\ref{lem:order}.

\subsubsection{Proof of achievability}
We will show that for any $r\in \mathbb{R}$ 
\begin{align}
\limsup_{n\rightarrow\infty} \hat{\alpha}_{n,\rho}(e^{-nR_n })
\leq \Phi\left(\frac{r}{\sqrt{V(A:B)_{\rho} }}\right).
\end{align}
\begin{proof}
Let $r\in \mathbb{R}$. 
For any $n\in \mathbb{N}_{>0}$, let 
\begin{align}
\mu_n&\coloneqq nI(A:B)_\rho +\sqrt{n}r+\log g_{n,d_A}+\log g_{n,d_B}.
\label{eq:def-mu-n}
\end{align}
In the following, we consider again the test $T^n_{A^nB^n}$, the PMFs $P_n,Q_n$, and the random variables $X_n,X_n'$ from the proof of achievability in Appendix~\ref{proof:achievability}, see \eqref{eq:test-np}--\eqref{eq:test-q}, where $\mu_n$ is now given by~\eqref{eq:def-mu-n} instead of~\eqref{eq:def-mun}. 
Note that the relations in~\eqref{eq:ach-at}--\eqref{eq:ach-gg0} still apply.
Let $n\in \mathbb{N}_{>0}$ be arbitrary but fixed. Then
\begin{subequations}\label{eq:ach-x'}
\begin{align}
e^{\mu_n}\mathrm{Pr}[P_n(X_n')\geq e^{\mu_n}Q_n(X_n')]
&= \tr[T^n_{A^nB^n}e^{\mu_n}\omega_{A^n}^n\otimes \omega_{B^n}^n]\\
&\leq \tr[T^n_{A^nB^n}\mathcal{P}_{\omega_{A^n}^n\otimes \omega_{B^n}^n}(\rho_{AB}^{\otimes n}) ]
\label{eq:ach-x'0}\\
&=\mathrm{Pr}[P_n(X_n)\geq e^{\mu_n}Q_n(X_n)].
\end{align}
\end{subequations}
\eqref{eq:ach-x'0} follows from the definition of the test in~\eqref{eq:test-np}. 
We have 
\begin{subequations}\label{eq:proof-second-feasible}
\begin{align}
\sup_{\substack{\sigma_{A^n}\in \mathcal{S}_{\sym}(A^{\otimes n}),\\ \tau_{B^n}\in \mathcal{S}_{\sym}(B^{\otimes n})}} \tr[\sigma_{A^n}\otimes \tau_{B^n} T^n_{A^nB^n}]
&\leq g_{n,d_A}g_{n,d_B}\mathrm{Pr}[P_n(X_n')\geq e^{\mu_n}Q_n(X_n')]
\label{eq:proof-ach-sym0}\\
&\leq g_{n,d_A}g_{n,d_B}e^{-\mu_n}\mathrm{Pr}[P_n(X_n)\geq e^{\mu_n}Q_n(X_n)]
\label{eq:proof-ach-sym1}\\
&=e^{-nR_n}\mathrm{Pr}[P_n(X_n)\geq e^{\mu_n}Q_n(X_n)]\leq e^{-nR_n}.
\label{eq:proof-ach-sym2}
\end{align}
\end{subequations}
\eqref{eq:proof-ach-sym0} follows from~\eqref{eq:ach-gg0}.
\eqref{eq:proof-ach-sym1} follows from~\eqref{eq:ach-x'}. 
\eqref{eq:proof-ach-sym2} follows from the definition of $R_n$ and $\mu_n$, see~\eqref{eq:def-mu-n}.
By~\eqref{eq:ach-at}, 
\begin{subequations}\label{eq:proof-second-y}
\begin{align}
\tr[\rho_{AB}^{\otimes n}(1-T^n_{A^nB^n})]
&=\mathrm{Pr}[P_n(X_n)<e^{\mu_n}Q_n(X_n)]
\\
&=\mathrm{Pr}[\log P_n(X_n)<\mu_n+\log Q_n(X_n)]\\
&=\mathrm{Pr}[\log P_n(X_n)-\log Q_n(X_n)<nR_n + \log g_{n,d_A}+\log g_{n,d_B}]\\
&=\mathrm{Pr}[Y_n<r].
\end{align}
\end{subequations}
For the last equality, we defined the random variable 
\begin{equation}\label{eq:def-yn}
Y_n\coloneqq \frac{1}{\sqrt{n}}(\log P_n(X_n)-\log Q_n(X_n)-n I(A:B)_\rho -\log g_{n,d_A}-\log g_{n,d_B}).
\end{equation}
\eqref{eq:proof-second-feasible} and~\eqref{eq:proof-second-y} imply that 
\begin{align}\label{eq:limsup-a-yx}
\limsup_{n\rightarrow\infty}\hat{\alpha}_{n,\rho}(e^{-nR_n})
\leq \limsup_{n\rightarrow\infty}\tr[\rho_{AB}^{\otimes n}(1-T^n_{A^nB^n})]
=\limsup_{n\rightarrow\infty}\mathrm{Pr}[Y_n<r].
\end{align}
For any $n\in \mathbb{N}_{>0}$,
let us define $M_n(t)\coloneqq \mathbb{E}[e^{tY_n}]$ for $t\in [0,\infty)$. 
Then, for all $t\in [0,\infty)$
\begin{align}
\log M_n(t)
&=\log \mathbb{E}[e^{tY_n}] \\
&=\log \mathbb{E}[\left(\frac{P_n(X_n)}{Q_n(X_n)} \right)^{\frac{t}{\sqrt{n}}} (g_{n,d_A}g_{n,d_B})^{-\frac{t}{\sqrt{n}} } e^{-t\sqrt{n}I(A:B)_\rho } ] \\
&=t\sqrt{n} \left(\frac{1}{n} D_{1+\frac{t}{\sqrt{n}}}(P_n\| Q_n)-I(A:B)_\rho \right)
-\frac{t}{\sqrt{n}}\log (g_{n,d_A}g_{n,d_B})
\\
&=t\sqrt{n}\left(\frac{1}{n} D_{1+\frac{t}{\sqrt{n}}}(\mathcal{P}_{\omega_{A^n}^n\otimes \omega_{B^n}^n}(\rho_{AB}^{\otimes n}) \| \omega_{A^n}^n\otimes \omega_{B^n}^n ) -I(A:B)_\rho \right)
-\frac{t}{\sqrt{n}}\log (g_{n,d_A}g_{n,d_B}).
\label{eq:ach-tgg}
\end{align}
By \cite[Remark~1~(b)]{burri2025prmisrmi1}, the last term in~\eqref{eq:ach-tgg} vanishes in the limit $n\rightarrow\infty$. 
The limit as $n\rightarrow\infty$ of the first term in~\eqref{eq:ach-tgg} has been determined in~\cite[Theorem~4~(j)]{burri2025prmisrmi1} by means of the mutual information variance of $\rho_{AB}$. 
Thus, for all $t\in [0,\infty)$
\begin{equation}\label{eq:lim-ln-mn}
\lim_{n\rightarrow\infty}\log M_n(t)
=\frac{t^2}{2}V(A:B)_\rho.
\end{equation}

Let $Y$ be a normally distributed random variable with mean $\mu\coloneqq 0$ and variance $\sigma^2\coloneqq V(A:B)_\rho $. Its moment generating function is given for all $t\in \mathbb{R}$ by
\begin{align}\label{eq:mgf-normal}
M(t)\coloneqq \mathbb{E}[e^{tY}]=\exp\left(t \mu +\frac{t^2}{2}\sigma^2\right)
=\exp\left(\frac{t^2}{2}V(A:B)_\rho\right).
\end{align}
A comparison of~\eqref{eq:lim-ln-mn} with~\eqref{eq:mgf-normal} shows that 
$M(t)=\lim_{n\rightarrow\infty}M_n(t)$ for all $t\in (0,\infty)$. 
We can now apply Curtiss' theorem~\cite{mukherjea2006note} (see also~\cite[Lemma 20]{hayashi2016correlation}).
According to \cite[Theorem~2]{mukherjea2006note}, if a sequence of moment generating functions $M_n(t)$ converges pointwise to a moment generating function $M(t)$ for all $t$ in some open interval of the positive real axis, then the corresponding sequence of distribution functions converges weakly to the distribution function corresponding to $M(t)$. 
Thus,
$\limsup_{n\rightarrow\infty}\mathrm{Pr}[Y_n<r]=\mathrm{Pr}[Y<r]$. 
By~\eqref{eq:limsup-a-yx}, we can conclude that 
\begin{align}
\limsup_{n\rightarrow\infty}\hat{\alpha}_{n,\rho}(e^{-nR_n})
&\leq \mathrm{Pr}[Y<r]
=\mathrm{Pr}\left[\frac{Y}{\sigma}<\frac{r}{\sigma}\right]
=\Phi\left(\frac{r}{\sigma}\right)
=\Phi\left(\frac{r}{\sqrt{V(A:B)_\rho}}\right).
\end{align}
\end{proof}

\subsubsection{Proof of optimality}
We will show that for any $r\in \mathbb{R}$ 
\begin{equation}
\liminf_{n\rightarrow\infty} \hat{\alpha}_{n,\rho}^{\mathrm{iid}}(e^{-nR_n})
\geq \Phi\left(\frac{r}{\sqrt{V(A:B)_{\rho} }}\right).
\end{equation}
\begin{proof}
Let $r\in \mathbb{R}$. 
Let us define the function $\hat{\alpha}_{n,\rho}^{\mathrm{mar}}(\mu )$ as in~\eqref{eq:def-alpha-mar}. 
Then, 
\begin{align}\label{eq:proof-opts2}
\liminf_{n\rightarrow\infty} \hat{\alpha}_{n,\rho}^{\mathrm{iid}}(e^{-nR_n })
&\geq \lim_{n\rightarrow\infty} \hat{\alpha}_{n,\rho}^{\mathrm{mar}}(e^{-nR_n})
= \Phi\left(\frac{r}{\sqrt{V(A:B)_{\rho} }}\right).
\end{align}
For the last equality, we have used the results in~\cite{tomamichel2013hierarchy,li2014second} on the second-order asymptotics for i.i.d. quantum hypothesis testing.
\end{proof}

\begin{rem}[Extensions of Corollary~\ref{cor:stein} and Theorem~\ref{thm:second}]
By~\eqref{eq:proof-opts2}, it is clear that Theorem~\ref{thm:second} also holds if 
$\hat{\alpha}_{n,\rho}$ is replaced by $\hat{\alpha}_{n,\rho}^{\mathrm{mar}}$ as defined in~\eqref{eq:def-alpha-mar}.
Furthermore, note that Corollary~\ref{cor:stein} also holds if 
$\hat{\alpha}_{n,\rho}$ is replaced by $\hat{\alpha}_{n,\rho}^{\mathrm{mar}}$ due to the quantum Stein's lemma~\cite{hiai1991proper,ogawa2005strong}.
\end{rem}

\bibliographystyle{arxiv_fullname}
\bibliography{bibfile}

\begin{thebibliography}{10}

\bibitem{tomamichel2018operational}
Marco Tomamichel and Masahito Hayashi.
\newblock {Operational Interpretation of R\'enyi Information Measures via
  Composite Hypothesis Testing Against Product and Markov Distributions}.
\newblock {\em IEEE Transactions on Information Theory}, 64(2):1064--1082,
  2018.
\newblock
  \texttt{\href{http://dx.doi.org/10.1109/TIT.2017.2776900}{DOI:\,10.1109/TIT.2017.2776900}}.

\bibitem{lapidoth2019two}
Amos Lapidoth and Christoph Pfister.
\newblock {Two Measures of Dependence}.
\newblock {\em Entropy}, 21(778), 2019.
\newblock
  \texttt{\href{http://dx.doi.org/10.3390/e21080778}{DOI:\,10.3390/e21080778}}.

\bibitem{hoeffding1965asymptotically}
Wassily Hoeffding.
\newblock Asymptotically optimal tests for multinomial distributions.
\newblock {\em The Annals of Mathematical Statistics}, 36(2):369--401, 1965.
\newblock
  \texttt{\href{http://dx.doi.org/10.1214/AoMS/1177700150}{DOI:\,10.1214/AoMS/1177700150}}.

\bibitem{hoeffding1965probabilities}
Wassily Hoeffding.
\newblock On probabilities of large deviations.
\newblock {\em Proceedings of the Fifth Berkeley Symposium on Mathematical
  Statistics and Probability}, 5(1):203--219, 1967.

\bibitem{csiszar1971error}
Imre Csisz\'ar and Giuseppe Longo.
\newblock On the error exponent for source coding and for testing simple
  statistical hypotheses.
\newblock {\em Studia Scientiarum Mathematicarum Hungarica}, 6:181--191, 1971.

\bibitem{blahut1974hypothesis}
Richard Blahut.
\newblock {Hypothesis Testing and Information Theory}.
\newblock {\em IEEE Transactions on Information Theory}, 20(4):405--417, 1974.
\newblock
  \texttt{\href{http://dx.doi.org/10.1109/TIT.1974.1055254}{DOI:\,10.1109/TIT.1974.1055254}}.

\bibitem{audenaert2008asymptotic}
Koenraad M.~R. Audenaert, Michael Nussbaum, Arleta Szko\l{}a, and Frank
  Verstraete.
\newblock {Asymptotic Error Rates in Quantum Hypothesis Testing}.
\newblock {\em Communications in Mathematical Physics}, 279(1):251--283, 2008.
\newblock
  \texttt{\href{http://dx.doi.org/10.1007/s00220-008-0417-5}{DOI:\,10.1007/s00220-008-0417-5}}.

\bibitem{han1989strong}
Te~Sun Han and Kingo Kobayashi.
\newblock The strong converse theorem for hypothesis testing.
\newblock {\em IEEE Transactions on Information Theory}, 35(1):178--180, 1989.
\newblock
  \texttt{\href{http://dx.doi.org/10.1109/18.42188}{DOI:\,10.1109/18.42188}}.

\bibitem{nakagawa1993converse}
Kenji Nakagawa and Fumio Kanaya.
\newblock {On the Converse Theorem in Statistical Hypothesis Testing}.
\newblock {\em IEEE Transactions on Information Theory}, 39(2):623--628, 1993.
\newblock
  \texttt{\href{http://dx.doi.org/10.1109/18.212293}{DOI:\,10.1109/18.212293}}.

\bibitem{lapidoth2018testing}
Amos Lapidoth and Christoph Pfister.
\newblock {Testing Against Independence and a R\'enyi Information Measure}.
\newblock In {\em 2018 IEEE Information Theory Workshop (ITW)}, 2018.
\newblock
  \texttt{\href{http://dx.doi.org/10.1109/ITW.2018.8613520}{DOI:\,10.1109/ITW.2018.8613520}}.

\bibitem{wilde2014strong}
Mark~M. Wilde, Andreas Winter, and Dong Yang.
\newblock {Strong Converse for the Classical Capacity of Entanglement-Breaking
  and Hadamard Channels via a Sandwiched R\'enyi Relative Entropy}.
\newblock {\em Communications in Mathematical Physics}, 331(2):593--622, 2014.
\newblock
  \texttt{\href{http://dx.doi.org/10.1007/s00220-014-2122-x}{DOI:\,10.1007/s00220-014-2122-x}}.

\bibitem{hayashi2016correlation}
Masahito Hayashi and Marco Tomamichel.
\newblock {Correlation detection and an operational interpretation of the
  R\'enyi mutual information}.
\newblock {\em Journal of Mathematical Physics}, 57(10):102201, 2016.
\newblock
  \texttt{\href{http://dx.doi.org/10.1063/1.4964755}{DOI:\,10.1063/1.4964755}}.

\bibitem{hayashi2007error}
Masahito Hayashi.
\newblock Error exponent in asymmetric quantum hypothesis testing and its
  application to classical-quantum channel coding.
\newblock {\em Physical Review A}, 76(062301), 2007.
\newblock
  \texttt{\href{http://dx.doi.org/10.1103/PhysRevA.76.062301}{DOI:\,10.1103/PhysRevA.76.062301}}.

\bibitem{nagaoka2006converse}
Hiroshi Nagaoka.
\newblock {The Converse Part of The Theorem for Quantum Hoeffding Bound}, 2006.
\newblock
  \texttt{\href{http://dx.doi.org/10.48550/arXiv.quant-ph/0611289}{DOI:\,10.48550/arXiv.quant-ph/0611289}}.

\bibitem{mosonyi2014quantum}
Mil\'an Mosonyi and Tomohiro Ogawa.
\newblock {Quantum Hypothesis Testing and the Operational Interpretation of the
  Quantum R\'enyi Relative Entropies}.
\newblock {\em Communications in Mathematical Physics}, 334(3):1617--1648,
  2015.
\newblock
  \texttt{\href{http://dx.doi.org/10.1007/s00220-014-2248-x}{DOI:\,10.1007/s00220-014-2248-x}}.

\bibitem{mosonyi2015two}
Mil\'an Mosonyi and Tomohiro Ogawa.
\newblock {Two Approaches to Obtain the Strong Converse Exponent of Quantum
  Hypothesis Testing for General Sequences of Quantum States}.
\newblock {\em IEEE Transactions on Information Theory}, 61(12):6975--6994,
  2015.
\newblock
  \texttt{\href{http://dx.doi.org/10.1109/TIT.2015.2489259}{DOI:\,10.1109/TIT.2015.2489259}}.

\bibitem{berta2021composite}
Mario Berta, Fernando G. S.~L. Brand\~{a}o, and Christoph Hirche.
\newblock {On Composite Quantum Hypothesis Testing}.
\newblock {\em Communications in Mathematical Physics}, 385(1):55--77, 2021.
\newblock
  \texttt{\href{http://dx.doi.org/10.1007/s00220-021-04133-8}{DOI:\,10.1007/s00220-021-04133-8}}.

\bibitem{schmitt2026tumulainformationdoublyminimized}
Lukas Schmitt, Filippo Girardi, and Laura Burri.
\newblock {Tumula information and doubly minimized Petz R\'enyi lautum
  information}, 2026.
\newblock
  \texttt{\href{http://dx.doi.org/10.48550/arXiv.2603.17005}{DOI:\,10.48550/arXiv.2603.17005}}.

\bibitem{fang2025errorexponentsquantumstate}
Kun Fang and Masahito Hayashi.
\newblock Error exponents of quantum state discrimination with composite
  correlated hypotheses, 2025.
\newblock Available online: \url{https://arxiv.org/abs/2508.12901}.

\bibitem{Brand_o_2010}
Fernando G. S.~L. Brandão and Martin~B. Plenio.
\newblock {A Generalization of Quantum Stein’s Lemma}.
\newblock {\em Communications in Mathematical Physics}, 295(3):791--828, 2010.
\newblock
  \texttt{\href{http://dx.doi.org/10.1007/s00220-010-1005-z}{DOI:\,10.1007/s00220-010-1005-z}}.

\bibitem{Hayashi2025GeneralizedQSL}
Masahito Hayashi and Hayata Yamasaki.
\newblock {The generalized quantum Stein's lemma and the second law of quantum
  resource theories}.
\newblock {\em Nature Physics}, 21(12):1988--1993, 2025.
\newblock
  \texttt{\href{http://dx.doi.org/10.1038/s41567-025-03047-9}{DOI:\,10.1038/s41567-025-03047-9}}.

\bibitem{Lami_2025}
Ludovico Lami.
\newblock {A Solution of the Generalized Quantum Stein’s Lemma}.
\newblock {\em IEEE Transactions on Information Theory}, 71(6):4454--4484,
  2025.
\newblock
  \texttt{\href{http://dx.doi.org/10.1109/tit.2025.3543610}{DOI:\,10.1109/tit.2025.3543610}}.

\bibitem{renner2006security}
Renato Renner.
\newblock {Security of Quantum Key Distribution}, 2006.
\newblock
  \texttt{\href{http://dx.doi.org/10.48550/arXiv.quant-ph/0512258}{DOI:\,10.48550/arXiv.quant-ph/0512258}}.

\bibitem{christandl2009postselection}
Matthias Christandl, Robert K\"onig, and Renato Renner.
\newblock {Postselection Technique for Quantum Channels with Applications to
  Quantum Cryptography}.
\newblock {\em Physical Review Letters}, 102(2), 2009.
\newblock
  \texttt{\href{http://dx.doi.org/10.1103/PhysRevLett.102.020504}{DOI:\,10.1103/PhysRevLett.102.020504}}.

\bibitem{tomamichel2013hierarchy}
Marco Tomamichel and Masahito Hayashi.
\newblock {A Hierarchy of Information Quantities for Finite Block Length
  Analysis of Quantum Tasks}.
\newblock {\em IEEE Transactions on Information Theory}, 59(11):7693--7710,
  2013.
\newblock
  \texttt{\href{http://dx.doi.org/10.1109/TIT.2013.2276628}{DOI:\,10.1109/TIT.2013.2276628}}.

\bibitem{li2014second}
Ke~Li.
\newblock Second-order asymptotics for quantum hypothesis testing.
\newblock {\em The Annals of Statistics}, 42(1):171--189, 2014.
\newblock
  \texttt{\href{http://dx.doi.org/10.1214/13-AoS1185}{DOI:\,10.1214/13-AoS1185}}.

\bibitem{petz1986quasi}
D\'enes Petz.
\newblock Quasi-entropies for finite quantum systems.
\newblock {\em Reports on Mathematical Physics}, 23(1):57--65, 1986.
\newblock
  \texttt{\href{http://dx.doi.org/10.1016/0034-4877(86)90067-4}{DOI:\,10.1016/0034-4877(86)90067-4}}.

\bibitem{mueller2013quantum}
Martin M\"uller-Lennert, Fr\'ed\'eric Dupuis, Oleg Szehr, Serge Fehr, and Marco
  Tomamichel.
\newblock {On quantum R\'enyi entropies: A new generalization and some
  properties}.
\newblock {\em Journal of Mathematical Physics}, 54(12):122203, 2013.
\newblock
  \texttt{\href{http://dx.doi.org/10.1063/1.4838856}{DOI:\,10.1063/1.4838856}}.

\bibitem{mosonyi2022error}
Mil\'an Mosonyi, Zsombor Szil\'agyi, and Mih\'aly Weiner.
\newblock {On the Error Exponents of Binary State Discrimination With Composite
  Hypotheses}.
\newblock {\em IEEE Transactions on Information Theory}, 68(2):1032--1067,
  2022.
\newblock
  \texttt{\href{http://dx.doi.org/10.1109/TIT.2021.3125683}{DOI:\,10.1109/TIT.2021.3125683}}.

\bibitem{burri2025prmisrmi1}
Laura Burri.
\newblock {Doubly minimized Petz and sandwiched R\'enyi mutual information:
  Properties}, 2026.
\newblock
  \texttt{\href{http://dx.doi.org/10.48550/arXiv.2406.01699}{DOI:\,10.48550/arXiv.2406.01699}}.

\bibitem{csiszar1995generalized}
Imre Csisz\'ar.
\newblock {Generalized Cutoff Rates and R\'enyi's Information Measures}.
\newblock {\em IEEE Transactions on Information Theory}, 41(1):26--34, 1995.
\newblock
  \texttt{\href{http://dx.doi.org/10.1109/18.370121}{DOI:\,10.1109/18.370121}}.

\bibitem{alajaji2004csiszars}
Fady Alajaji, Po-Ning Chen, and Ziad Rached.
\newblock {Csisz\'ar's Cutoff Rates for the General Hypothesis Testing
  Problem}.
\newblock {\em IEEE Transactions on Information Theory}, 50(4):663--678, 2004.
\newblock
  \texttt{\href{http://dx.doi.org/10.1109/TIT.2004.825040}{DOI:\,10.1109/TIT.2004.825040}}.

\bibitem{Cheng2017ModerateDA}
Hao-Chung Cheng and Min-Hsiu Hsieh.
\newblock {Moderate Deviation Analysis for Classical-Quantum Channels and
  Quantum Hypothesis Testing}.
\newblock {\em IEEE Transactions on Information Theory}, 64:1385--1403, 2017.
\newblock
  \texttt{\href{http://dx.doi.org/10.1109/TIT.2017.2781254}{DOI:\,10.1109/TIT.2017.2781254}}.

\bibitem{Chubb_2017}
Christopher~T. Chubb, Vincent Y.~F. Tan, and Marco Tomamichel.
\newblock {Moderate Deviation Analysis for Classical Communication over Quantum
  Channels}.
\newblock {\em Communications in Mathematical Physics}, 355(3):1283--1315,
  2017.
\newblock
  \texttt{\href{http://dx.doi.org/10.1007/s00220-017-2971-1}{DOI:\,10.1007/s00220-017-2971-1}}.

\bibitem{Ellis1975entropy}
Richard~S. Ellis.
\newblock {\em Entropy, Large Deviations, and Statistical Mechanics}.
\newblock Springer New York, NY, 1985.
\newblock
  \texttt{\href{http://dx.doi.org/10.1007/978-1-4613-8533-2}{DOI:\,10.1007/978-1-4613-8533-2}}.

\bibitem{chen2000generalization}
Po-Ning Chen.
\newblock {Generalization of G\"artner-Ellis theorem}.
\newblock {\em IEEE Transactions on Information Theory}, 46(7):2752--2760,
  2000.
\newblock
  \texttt{\href{http://dx.doi.org/10.1109/18.887893}{DOI:\,10.1109/18.887893}}.

\bibitem{mukherjea2006note}
Arunava Mukherjea, Murali Rao, and Stephen Suen.
\newblock A note on moment generating functions.
\newblock {\em Statistics \& Probability Letters}, 76(11):1185--1189, 2006.
\newblock
  \texttt{\href{http://dx.doi.org/10.1016/J.SPL.2005.12.026}{DOI:\,10.1016/J.SPL.2005.12.026}}.

\bibitem{hiai1991proper}
Fumio Hiai and D\'enes Petz.
\newblock The proper formula for relative entropy and its asymptotics in
  quantum probability.
\newblock {\em Communications in Mathematical Physics}, 143:99--114, 1991.
\newblock
  \texttt{\href{http://dx.doi.org/10.1007/BF02100287}{DOI:\,10.1007/BF02100287}}.

\bibitem{ogawa2005strong}
Tomohiro Ogawa and Hiroshi Nagaoka.
\newblock {Strong Converse and Stein's Lemma in Quantum Hypothesis Testing}.
\newblock {\em IEEE Transactions on Information Theory}, 46(7):2428--2433,
  2000.
\newblock
  \texttt{\href{http://dx.doi.org/10.1109/18.887855}{DOI:\,10.1109/18.887855}}.

\end{thebibliography}

\end{document}